\documentclass[10pt, letterpaper,reqno]{amsart}
\usepackage[left=1in,right=1in,bottom=0.5in,top=0.8in]{geometry}
\usepackage{amsfonts}
\usepackage{amsmath, amssymb}
\usepackage{graphicx}
\usepackage[font=small,labelfont=bf]{caption}
\usepackage{epstopdf}
\usepackage{xcolor}
\usepackage{amsthm}
\usepackage{float}
\usepackage{enumitem}
\usepackage{pgfplots}
\usepackage{listings}
\usepackage{longtable}
\usepackage{mathrsfs}
\usepackage{dsfont}
\usepackage{mathtools}
\usepackage{dsfont}
\usepackage{tipa}
\usepackage{xfrac}  
\usepackage[pdfpagelabels,hyperindex]{hyperref}
\hypersetup{linkbordercolor=green}

\newcommand*{\rom}[1]{\expandafter\@slowromancap\romannumeral #1@}
\DeclareMathAlphabet{\mathpzc}{OT1}{pzc}{m}{it} 
 \colorlet{lgray}{white!80!black}
\colorlet{lred}{white!85!red}
\colorlet{lgreen}{white!60!green}
\colorlet{dgreen}{black!30!green}
\colorlet{lpurple}{white!60!purple}
\colorlet{lblue}{white!60!blue}
\definecolor{green}{rgb}{0.1,0.8,0.1}
\definecolor{yellow}{rgb}{1.0,0.85,0.25}
\definecolor{purple}{rgb}{1.0, 0, 1.0}
\definecolor{blue}{rgb}{0, 0, 1.0} 
\tikzstyle{unfused}=[lgray, line width=1.5pt, ->]
\tikzstyle{fused}=[lgray, line width=4pt, ->]
\tikzstyle{dual}=[black, line width=1pt, dashed]
\tikzstyle{lightdual}=[black, line width=0.5pt, dashed]
\tikzstyle{cut}=[black, line width=1.0pt]
 \renewcommand{\tikz}[2]{
\begin{tikzpicture}[scale=#1,baseline=(current bounding box.center),>=stealth]
#2
\end{tikzpicture}}

\theoremstyle{plain}
\newtheorem{theorem}{Theorem}[section]
\newtheorem{lemma}[theorem]{Lemma}
\newtheorem{proposition}[theorem]{Proposition}
\newtheorem{corollary}[theorem]{Corollary}

\theoremstyle{definition}
\newtheorem{definition}[theorem]{Definition}
\newtheorem{remark}[theorem]{Remark}

\newtheorem{claim}[theorem]{Claim}
\numberwithin{equation}{section} 
\colorlet{shadecolor}{gray!20}
\pgfplotsset{compat=1.9}
\usetikzlibrary{shapes.multipart}
\usetikzlibrary{patterns}
\usetikzlibrary{shapes.multipart}
\usetikzlibrary{arrows}
\usetikzlibrary{decorations.markings}
\usepgflibrary{decorations.shapes}
\usetikzlibrary{decorations.shapes}
\usepgflibrary{shapes.symbols}
\usetikzlibrary{shapes.symbols}
\usetikzlibrary{decorations.pathreplacing} 
\tikzstyle{fleche}=[>=stealth', postaction={decorate}, thick]
\tikzstyle{axis}=[->, >=stealth', thick, gray]
\tikzstyle{paths}=[>->, >=stealth', thick]
\tikzstyle{path}=[->, >=stealth', thick]
\tikzstyle{grille}=[dotted, gray]
\newcommand{\pathR}{\raisebox{-25pt}{\begin{tikzpicture}[scale=0.6,>=stealth]
\draw[thick,->] (0,-1) -- (0,1);
\draw[thick,->] (-1,0) -- (1,0); 
\node[below] at (0,-1) {\small $a$};
\node[left] at (-1,0) {\small $b$};
\node[above] at (0,1) {\small $c$};
\node[right] at (1,0) {\small $d$};
\end{tikzpicture}}}
\newcommand{\pathKl}{\raisebox{-15pt}{\begin{tikzpicture}[scale=0.8,>=stealth]
\draw[thick,->] (0,-1) -- (0,0)--(1,0); 
\node[below] at (0,-1) {\small $a$};
\node[right] at (1,0) {\small $d$};
\end{tikzpicture}}}
\newcommand{\pathKr}{\raisebox{-15pt}{\begin{tikzpicture}[scale=0.8,>=stealth]
\draw[thick,->] (-1,0) -- (0,0)--(0,1); 
\node[left] at (-1,0) {\small $b$};
\node[above] at (0,1) {\small $c$};
\end{tikzpicture}}}
\DeclareFontFamily{U}{mathx}{}
\DeclareFontShape{U}{mathx}{m}{n}{<-> mathx10}{}
\DeclareSymbolFont{mathx}{U}{mathx}{m}{n}
\DeclareMathAccent{\widehat}{0}{mathx}{"70}
\DeclareMathAccent{\widecheck}{0}{mathx}{"71} 
\def \lw {\biggl< W\bigg{|}}
\def \rv {\bigg{|}V\biggr>}
\def \PP {\mathbb{P}}
\def \C {\mathbb{C}}
\def \z {\mathbb{Z}}
\def \be {\begin{equation}}
\def \ee {\end{equation}}
\def \de {\delta}
\def \ep {\varepsilon}
\def \D {\mathbf{D}}
\def \E {\mathbf{E}}
\def \k {\kappa}
\def \ll {\langle}
\def \rr {\rangle}
\def \lb {\left(}
\def \rb {\right)}
\def \lbe {\left[}
\def \rbe {\right]}
\def \u {\uparrow}
\def \r {\rightarrow}
\def \M {\mathbf{M}}

\def \hA {\mathcal{A}}
\newcommand{\UU}{{\mathbb U}}
\def \o {\otimes}
\def \cR {\widecheck{R}}
\def \bK {\overline{K}}
\def \sR {\mathsf{R}}
\def \sK {\mathsf{K}}
\def \sbK {\overline{\sK}}
\def \scR {\widecheck{\sR}}
\def \sM {\mathsf{M}}
\def \sbM {\boldsymbol{\mathsf{M}}}
\def \fV{\mathfrak{V}}
\def \v {\upsilon}
\def \k {\kappa}
\def \d {\mathbf{d}}
\def \e {\mathbf{e}}
\def \aa{\mathtt{a}}
\def \bb{\mathtt{b}}
\def \cc{\mathtt{c}}
\def \dd{\mathtt{d}}

\def \hA {\mathcal{A}} 
\def \eb{\mathrm{e}}

\def \soR{\mathring{\sR}^I}
\def \soRb{\mathring{\scR}^I}
\def \soK{\mathring{\sK}^I}
\def \sobK{\mathring{\sbK}^I}

\def \hQ{\mathcal{Q}}

\def \R {\mathrm{R}}
\def \lK {_{\ell}\mathrm{K}}
\def \rK {_{r}\mathrm{K}}
\def \lra{\longrightarrow}
\def \lla{\longleftarrow}

\def \th {\theta} 
\def \hP {\mathcal{P}}
\def \aaa {\mathpzc{a}}
\def \bbb {\mathpzc{b}}
\def \ccc {\mathpzc{c}}
\def \ddd {\mathpzc{d}}
\def \xx {\mathtt{x}}
\def \yy {\mathtt{y}}
\def \AA {\mathtt{A}}
\def \BB {\mathtt{B}}
\def \CC {\mathtt{C}}
\def \DD {\mathtt{D}}
\def \ttA{A\sqrt{t}}
\def \ttB{B\sqrt{t}}
\def \ttC{C/\sqrt{t}}
\def \ttD{D/\sqrt{t}}
\def \Pih{\widehat{\Pi}}

\def \Co {\C^2}
\def \CI {\C^{I+1}}
\def \CoI {\lb\C^2\rb^{\o I}}
 
\def \pt {\partial_t}
\def \po {|_{t=1}}
\def\qp#1{\lb#1;q\rb_\infty}
\def\qps#1{\lb#1;q\rb}

\newcommand{\upperRomannumeral}[1]{\uppercase\expandafter{\romannumeral#1}}
\DeclareMathOperator{\dehp}{DEHP}
\DeclareMathOperator{\usw}{USW}

\DeclareMathOperator{\End}{End}
\DeclareMathOperator{\Sym}{Sym}
\DeclareMathOperator{\inv}{inv}
\DeclareMathOperator{\tinv}{\overline{inv}}
\DeclareMathOperator{\id}{id}

\DeclareMathOperator{\zf}{ZF}
\DeclareMathOperator{\gz}{GZ}
\def \sy {\Sym_q^I}

\usepgflibrary{fpu}
\makeatletter 
\let\NAT@parse\undefined
\makeatother

\title{Stationary measures for higher spin vertex models on a strip}
\author[Z. Yang]{Zongrui Yang}
\address{Z. Yang, Columbia University,
	Department of Mathematics,
	2990 Broadway,
	New York, NY 10027, USA.}
\email{zy2417@columbia.edu}
\begin{document}
\begin{abstract} 
We introduce a higher spin vertex model on a strip with fused vertex weights. This model can be regarded as a generalization of both the unfused six-vertex model on a strip \cite{Y} and  an `integrable two-step Floquet dynamics' model introduced in \cite{Vanicat_fused_MPA}. We solve for the stationary measure using a fused version of the matrix product ansatz and then characterize it in terms of the Askey--Wilson process. Using this characterization, we obtain the limits of the mean density along an arbitrary down-right path. It turns out that all these models share a common phase diagram, which, after an appropriate mapping, matches the phase diagram of open ASEP. This provides evidence for the universality of this phase diagram. 
\end{abstract}  
\maketitle
\section{Introduction and main results} 
\subsection{Preface}   \label{subsec:preface}
The higher spin vertex model plays a central role among probabilistic systems in the Kardar-Parisi-Zhang (KPZ) universality class, since it can be degenerated into many other systems in this class, including interacting particle systems and polymer models. For   summaries of its  degenerations, see \cite[Figure 1]{CP}, \cite[Figures 1 and 2]{kuan2018algebraic}.  
While most studies focus on  vertex models in full space, recent progress has been made towards such models with open boundary, see for example \cite{barraquand2018stochastic,Jimmy_He,Y}.

On a separate note, the matrix product ansatz method, introduced by \cite{DEHP93}, has been extensively adopted to study stationary measures for Markov chains, particularly for interacting particle systems. This method involves expressing the stationary measure as a product of matrices, one for each occupation number. These matrices need to satisfy certain consistency relations. In the case of open asymmetric simple exclusion process (ASEP), the matrix ansatz is related to Askey--Wilson polynomials \cite{USW04} and processes \cite{BW17}, enabling a rigorous derivation of the phase diagram, density profile and fluctuations \cite{BW17,BW19,wang2023askey}.

A physics paper \cite{Vanicat_fused_MPA} introduced a class of higher spin interacting particle systems called the `two-step Floquet dynamics'.  
We refer to the spin of an interacting particle system as $\frac{I}{2}$ if up to $I$ many particles are allowed to occupy a single site. 
The stationary measures of the spin-$\frac{I}{2}$ version of `two-step Floquet dynamics' can be  solved  by a fused version of matrix product ansatz. The matrices that are involved  
are obtained in \cite{Vanicat_fused_MPA} through developing a fusion procedure for the so-called Zamolodchikov-Faddeev (ZF) and Ghoshal-Zamolodchikov (GZ) relations. It is known in the physics literature \cite{ZZ,F,SW97,GZ94,Review_and_progress} that, for integrable systems with two open boundaries (in the sense of \cite{sklyanin1988boundary}), the ZF and GZ relations are connected to the consistency relations of the matrix ansatz. \cite{Vanicat_fused_MPA} constructed such systems and their stationary measures for $I\in\{1,2\}$ cases, and algebraic formulas for certain physical quantities were obtained.

A recent work \cite{Y} studied the stationary measure of the unfused stochastic six-vertex
model on a strip (when $I=1$). In this paper we study a higher spin  generalization of this model and its stationary measure. In the spin-$\frac{I}{2}$ version of such model, up to $I$ many arrows are allowed to occupy a single edge. The higher spin vertex model on a strip has vertex weights given by the fused $R$ and $K$ matrices, which are constructed from the (standard) fusion procedure that goes back to \cite{kirillov1987exact}. The stationary measure  of such a  model can be solved using the matrix product ansatz. The matrices involved in this matrix ansatz can be obtained from the fused solutions of ZF and GZ relations that are  generalized from \cite{Vanicat_fused_MPA}. We then utilize with modifications the techniques from \cite{USW04,BW17} to characterize the stationary measure in terms of the Askey--Wilson processes. Using this description, we investigate the limits of a basic (macroscopic) physical quantity of the system  known as the mean density, as the size of the system  going to infinity. The limits are given by different formulas within different regions, from which we obtain the phase diagram of the system.

We remark that the Markov chains defined by the fused vertex model on a strip are indexed by down-right paths on the strip. In this paper we are able to study the stationary measure corresponding to an arbitrary  path.  
It is interesting that the systems indexed by different down-right paths share the same phase diagram (but with different limits of mean density). This phenomenon has not been observed  previously in \cite{Y}, since the asymptotics was only obtained therein for the horizontal path.
Moreover, we observe that the system corresponding to a specific down-right path (the zig-zag path) coincides with one of the `two-step Floquet dynamics' models in \cite{Vanicat_fused_MPA}. Therefore our
fused vertex model on a strip can be considered as a generalization of this `two-step Floquet dynamics' system. Our results, in particular, answer an open question raised in \cite[Section 4]{Vanicat_fused_MPA} of the mean density and the phase diagram of such (a class of) systems. 

The family of Markov chains studied in this paper are parameterized by the following: the system size $N$, the number $I$ (which controls the spin), bulk and boundary parameters $q,\k,\aaa,\bbb,\ccc,\ddd$, and the shape of a down-right path $\hP$. As size $N$ approaches infinity, this class of models share a common phase diagram, which, after an appropriate mapping, matches the phase diagram of open ASEP. We remark that the open ASEP models are parameterized by the system size $N$ and  $q,\alpha,\beta,\gamma,\delta$.
This provides evidence for the universality of the open ASEP phase diagram, in the sense that a family of systems shares this common phase diagram.

It is possible that, as mentioned in the first paragraph, the higher spin vertex model on a strip studied in this paper could also be degenerated and analytically continued into particle systems and polymer models (see \cite{Jimmy_He} for an example in half space). If such procedure could also be done for the fused matrix ansatz, then one may get a description of its stationary measure. We plan to explore this direction in future research.

\subsection*{Outline of the introduction}
In Section \ref{subsec:higher spin six-vertex model on a strip intro} we introduce the higher spin vertex model on a strip with general (unspecified) vertex weights and define a Markov chain corresponding to an arbitrary down-right path. We solve the stationary measure of this Markov chain using the matrix product ansatz in Section \ref{subsec:MPA intro}, assuming certain consistency relations. In Section \ref{subsec:fused vertex model on a strip intro} we  introduce the fusion procedure for the $R$ and $K$ matrices and define the fused vertex model on a strip. We state the fusion for the ZF and GZ relations in Section \ref{subsec:fusion of ZF and GZ intro}, which gives a concrete matrix ansatz for the fused vertex model on a strip. In Section  \ref{subsec:AW intro} we offer an alternative expression of the stationary measure as the Askey--Wilson Markov processes. We state the limits of the mean density on any down-right path (which exhibits the phase diagram) in  Section \ref{subsec:mean density intro}. In Section \ref{subsec: Integrable discrete-time two-step Floquet dynamics intro} we demonstrate that one of the `two-step Floquet dynamics' models in \cite{Vanicat_fused_MPA} can be regarded as a special case of our fused vertex model on a strip corresponding to the down-right zig-zag path. Therefore our result in particular implies the limits of mean   density and phase diagram for this  model in \cite{Vanicat_fused_MPA}.

\subsection{Higher spin vertex model on a strip} \label{subsec:higher spin six-vertex model on a strip intro}
We introduce the stochastic spin-$\frac{I}{2}$ vertex model on a strip with general vertex weights and define its stationary measure on any down-right path.

Suppose $I\in\z_+$.
We consider certain configurations of arrows on the edges of the strip
\be\label{eq:strip}
\left\{(x,y)\in\z^2: 0\leq y\leq x\leq y+N\right\},
\ee
where each edge can contain $0$ up to $I$ arrows.
For all $y\in\mathbb{Z}_{\geq 0}$, we refer vertices $(y,y)$ as left boundary vertices and $(y+N,y)$ as right boundary vertices. Other vertices on the strip are referred to as bulk vertices.  
For each vertex of the strip, its left and/or bottom edges are called its incoming edges, and its  right and/or top edges are called its outgoing edges.

We will use the word ‘down-right path’ to refer to a path $\mathcal{P}$ that goes from a left boundary vertex of the strip to a right boundary vertex of the strip, with each step going downwards or rightwards by $1$. Every down-right path on the strip has length $N$, and there are $N$ outgoing up/right edges emanating from the path. In the configurations that we will be interested in, each of the outgoing edges of $\mathcal{P}$ can be occupied by $0$ up to $I$ arrows, which gives $(I+1)^N$ possible `outgoing configurations' of $\mathcal{P}$. We label the $N$ outgoing edges of $\mathcal{P}$ from the up-left start of the path to the down-right end of the path: $p_1,\dots,p_N\in\left\{\u,\r\right\}$, where $\u$ denotes a vertical edge and $\r$ denotes a horizontal edge. 
The $(I+1)^N$ outgoing configurations of $\mathcal{P}$ can be encoded as occupation variables $\tau=\tau_{\mathcal{P}}=(\tau_1,\dots,\tau_N)\in[[0,I]]^N$, where $0\leq\tau_i\leq I$ denote the number of arrows occupying edge $p_i$, for $1\leq i\leq N$.
Let $\mathcal{Q}$ be any down-right path sitting above $\mathcal{P}$,
which may contain edges coinciding with edges of $\mathcal{P}$. 
We denote by $\mathbb{U}(\mathcal{P},\mathcal{Q})$ the set of vertices between $\mathcal{P}$ and $\mathcal{Q}$, including those on $\mathcal{Q}$ but excluding those on $\mathcal{P}$. Figure \ref{fig:outgoing arrows down-right path} illustrates these definitions.

\begin{figure} 
\centering
\begin{tikzpicture}[scale=0.82]
\draw[dotted] (-0.5,-0.5)--(4.5,4.5);
\draw[dotted] (4.5,-0.5)--(9.5,4.5);
\draw[dotted] (0,0)--(5,0);
\draw[dotted] (1,1)--(6,1);
\draw[dotted] (2,2)--(7,2);
\draw[dotted] (3,3)--(8,3);
\draw[dotted] (4,4)--(9,4);
\draw[dotted] (0,-0.5)--(0,0);
\draw[dotted] (1,-0.5)--(1,1);
\draw[dotted] (2,-0.5)--(2,2);
\draw[dotted] (3,-0.5)--(3,3);
\draw[dotted] (4,-0.5)--(4,4);
\draw[dotted] (5,0)--(5,4.5);
\draw[dotted] (6,1)--(6,4.5);
\draw[dotted] (7,2)--(7,4.5);
\draw[dotted] (8,3)--(8,4.5);
\draw[dotted] (9,4)--(9,4.5);
\node at (1.7,2) {$\mathcal{P}$};
\node at (3.7,4) {$\mathcal{Q}$};
\draw[ultra thick] (2,2)--(2,1)--(3,1)--(3,0)--(5,0);
\draw[ultra thick] (4,4)--(4,3)--(8,3);
\draw[ultra thick,lgray] (2,2)--(3,2);
\draw[ultra thick,lgray] (3,1)--(3,2);
\draw[ultra thick,lgray] (3,1)--(4,1);
\draw[ultra thick,lgray] (4,0)--(4,1);
\draw[ultra thick,lgray] (5,0)--(5,1);
\draw[ultra thick,lgray] (4,4)--(5,4);
\draw[ultra thick,lgray](5,3)--(5,4);
\draw[ultra thick,lgray](6,3)--(6,4);
\draw[ultra thick,lgray](7,3)--(7,4);
\draw[ultra thick,lgray](8,3)--(8,4);
\node at (4,1) {$\bullet$};
\node at (5,1) {$\bullet$};
\node at (6,1) {$\bullet$};
\node at (3,2) {$\bullet$};
\node at (4,2) {$\bullet$};
\node at (5,2) {$\bullet$};
\node at (6,2) {$\bullet$};
\node at (7,2) {$\bullet$};
\node at (3,3) {$\bullet$};
\node at (4,3) {$\bullet$};
\node at (5,3) {$\bullet$};
\node at (6,3) {$\bullet$};
\node at (7,3) {$\bullet$};
\node at (8,3) {$\bullet$};
\node at (4,4) {$\bullet$};
\end{tikzpicture} 
\caption{Outgoing edges of $\hP$ and $\hQ$ and set of vertices $\UU(\hP,\hQ)$, for $N=5$ and for the down-right (thick) paths $\hP$ and $\hQ$ as depicted. The gray edges are outgoing edges of $\hP$ and $\hQ$. Outgoing edges of $\hP$ are labelled from the up-left to the down-right: $p_1=\r$, $p_2=\u$, $p_3=\r$, $p_4=\u$, $p_5=\u$. The thick nodes are vertices in $\mathbb{U}(\hP,\hQ)$. Initially we have a (deterministic) outgoing configuration of $\hP$. We inductively sample through all vertices in $\UU(\hP,\hQ)$ and get a probability measure on the set of all outgoing configurations of $\hQ$. This figure is the same as \cite[Figure 1]{Y}}
\label{fig:outgoing arrows down-right path}
\end{figure}
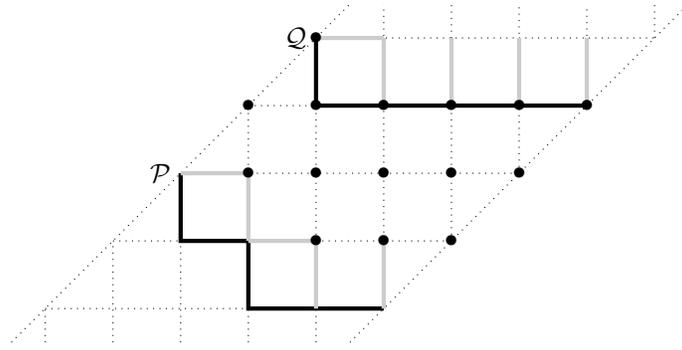

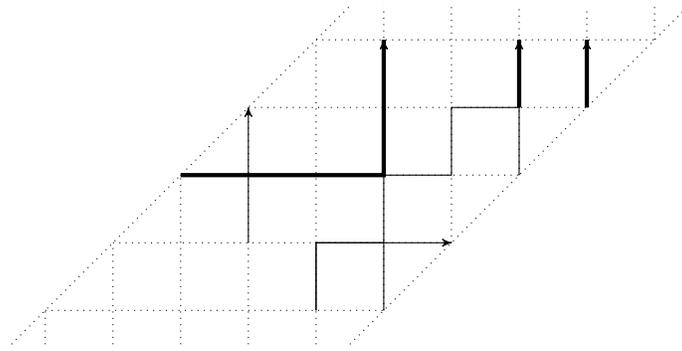
\begin{figure} 
\centering
\begin{tikzpicture}[scale=0.82]
\draw[dotted] (-0.5,-0.5)--(4.5,4.5);
\draw[dotted] (4.5,-0.5)--(9.5,4.5);
\draw[dotted] (0,0)--(5,0);
\draw[dotted] (1,1)--(6,1);
\draw[dotted] (2,2)--(7,2);
\draw[dotted] (3,3)--(8,3);
\draw[dotted] (4,4)--(9,4);
\draw[dotted] (0,-0.5)--(0,0);
\draw[dotted] (1,-0.5)--(1,1);
\draw[dotted] (2,-0.5)--(2,2);
\draw[dotted] (3,-0.5)--(3,3);
\draw[dotted] (4,-0.5)--(4,4);
\draw[dotted] (5,0)--(5,4.5);
\draw[dotted] (6,1)--(6,4.5);
\draw[dotted] (7,2)--(7,4.5);
\draw[dotted] (8,3)--(8,4.5);
\draw[dotted] (9,4)--(9,4.5);
\draw[ultra thick](2,2)--(5,2)--(5,4);\draw[path,thin](5,3)--(5,4);
\draw[thin](3,1)--(3,2);\draw[thin](5,0)--(5,1);\draw[ultra thick](7,3)--(7,4);
\draw[path,thin](8,3)--(8,4); 
\draw[path,thin](4,0)--(4,1)--(6,1);
\draw[path,thin](4,0)--(4,1)--(5,1)--(5,4);
\draw[path,thin](2,2)--(6,2)--(6,3)--(7,3)--(7,4);
\draw[path,thin](7,2)--(7,4);\draw[path,thin](2,2)--(3,2)--(3,3);\draw[ultra thick](8,3)--(8,4);
\end{tikzpicture}  
\caption{An example of sampling the spin-$\frac{I}{2}$ vertex model for $N=5$ and $I=2$. Down-right paths $\hP$ and $\hQ$ are the same as in Figure \ref{fig:outgoing arrows down-right path} and are omitted.
The edges being occupied by one arrow are depicted as thin edges and edges being occupied by two arrows are depicted as thick edges. The outgoing edges of $\hP$ and $\hQ$ are respectively $\tau_\hP=(2,1,0,1,1)$ and $\tau_\hQ=(0,2,0,2,2)$.}
\label{fig:sample configuration higher spin}
\end{figure}

We write $[[0,x]]:=\z\cap[0,x]$.
Suppose that there are three  matrices:
$$
\R=\lb\R_{a,b}^{c,d}\rb_{a,b,c,d\in[[0,I]]},\quad\quad \lK=\lb\lK_a^d\rb_{a,d\in[[0,I]]},\quad\quad \rK=\lb\rK_b^c\rb_{b,c\in[[0,I]]},
$$
which will later play the role of vertex weights respectively at the bulk/left boundary/right boundary. We will always require that these matrices satisfy the following conditions: 
\be  \label{eq:stochasticity intro}
\begin{split}
&{}\R_{a,b}^{c,d}\geq 0,\quad \lK_a^d\geq 0,\quad \rK_b^c\geq 0,\quad \text{for all }0\leq a,b,c,d\leq I,\\
&{}\sum_{c,d=0}^I\R_{a,b}^{c,d}=1,\quad \sum_{d=0}^I {\lK_a^d}=1,\quad \sum_{c=0}^I {\rK_b^c}=1, \quad \text{for all }0\leq a,b\leq I,\\
&{}\R_{a,b}^{c,d}=0,\quad\text{for all }0\leq a,b,c,d\leq I\text{ such that }a+b\neq c+d.
\end{split}
\ee 
The first and second conditions above together mean that the matrices are stochastic, so that all the vertices in the system are probabilistic. As will become clear later on, the third condition above will serve as the `conservation of arrows' in the  bulk of the system: the total number of arrows exiting a bulk vertex is equal to the total number of arrows entering this vertex.
Later we will specify $\R$, $\lK$ and $\rK$ to be the fused vertex weights, however, in order to define  the model, we only need them to satisfy condition \eqref{eq:stochasticity intro}.

The stochastic spin-$\frac{I}{2}$ vertex model is a Markovian sampling procedure that generates configurations. An `initial condition' is given as a down-right path $\hP$ and an outgoing configuration on it. At each vertex of the strip, we inductively sample through the following probabilities given by $\R$, $\lK$ and $\rK$:
\be\label{eq:vertex weights higher spin}
\PP\left(\pathR\right)=\R_{a,b}^{c,d},\quad\quad
\PP\left(\pathKl\right)=\lK_a^d,\quad\quad
\PP\left(\pathKr\right)=\rK_b^c,
\ee
where $0\leq a,b,c,d\leq I$ indicate the number of arrows contained in those edges. More precisely, suppose $\hP$ and $\hQ$ are down-right paths such that $\hQ$ sits above $\hP$. An `initial condition' is given by an outgoing configuration of $\hP$, i.e. initially there are arrows assigned to the outgoing edges of $\hP$. Suppose we have arrived at a vertex $(x,y)\in\UU(\hP,\hQ)$ and have sampled through all the vertices $(x',y')\in\UU(\hP,\hQ)$ such that either $y'<y$ or $y'=y$ and $x'<x$. Then we have already   assigned arrows to the incoming edges of vertex $(x,y)$. We then sample the outgoing edges of $(x,y)$ according to three probabilities given in \eqref{eq:vertex weights higher spin} respectively, in the cases when $(x,y)$ is a bulk/left boundary/right boundary vertex. After we sample through all the vertices in $\UU(\hP,\hQ)$, we get a probability measure on the set of all outgoing configurations of $\hQ$, whose randomness comes from the sampling procedure. See Figure \ref{fig:sample configuration higher spin} for an example of a configuration generated by this sampling.
We will encode the above sampling procedure as a transition probability matrix $P_{\hP,\hQ}(\tau,\tau')$, where $\tau,\tau'\in[[0,I]]^N$ are respectively the occupation variables of outgoing edges of $\hP$ and $\hQ$.
\begin{definition}\label{def:interacting particle system}
Assume the condition \eqref{eq:stochasticity intro} on the vertex weights.
Suppose $\hP$ is a down-right path on the strip. Denote by $\Upsilon_k\hP$ the up-right translation of $\hP$ by $(k,k)$, for all $k\in\mathbb{Z}_{\geq 0}$.  
We look at the outgoing configurations of down-right paths $\Upsilon_k\hP$   and regard $k\in\mathbb{Z}_{\geq 0}$ as time, which gives us a 
time-homogeneous Markov chain $(\tau(k))_{k\geq 0}$ on the state space $[[0,I]]^N$. This Markov chain has initial condition given by an outgoing configuration $\tau(0)\in[[0,I]]^N$ of $\hP$ and with the same transition probability matrix $P_{\Upsilon_k\hP,\Upsilon_{k+1}\hP}(\tau,\tau')=P_{\hP,\Upsilon_1\hP}(\tau,\tau')$ in each step $k\mapsto k+1$. 
When the vertex weights (additionally) satisfy:
\be\label{eq:condition irreducible}\begin{split}
     &\R_{a,b}^{c,d}\in(0,1),\quad\text{for all } a,b,c,d\in[[0,I]]\text{ such that } a+b=c+d,\\
     &\lK_a^d \text{ and } \rK_b^c\in(0,1),\quad\text{for all } a,b,c, d\in[[0,I]], 
\end{split}
\ee 
one can observe that this Markov chain  is irreducible.  
We will be interested in the (unique) stationary measure of this  system, which we  refer to as the stationary measure of spin-$\frac{I}{2}$ vertex model on a strip on $\hP$.

\end{definition} 
\begin{remark}\label{rmk:stochasticity and irreducibility}
    The conditions \eqref{eq:stochasticity intro} and \eqref{eq:condition irreducible} can be understood as follows: (1) Each vertex in the system is probabilistic. 
    (2) At the bulk, anything happens with positive probability as long as the `conservation of arrows' property holds, i.e. the total number of arrows exiting a bulk vertex equals the total number of arrows entering this vertex. (3) At the left and right boundaries, anything  happens with positive probability.
    In summary, arrows are conserved in the bulk, but can enter or exit the system at two open boundaries. 
    
    We will make all our assumptions explicit in the statements of the results in this paper.
Our main object of study is the fused vertex model on a strip, which is introduced in Definition~\ref{def: fused model intro} and, as shown in Proposition~\ref{prop:stochasticity irreducible}, satisfies conditions \eqref{eq:stochasticity intro} and \eqref{eq:condition irreducible}. 
\end{remark}

\subsection{Matrix product ansatz of stationary measure}
\label{subsec:MPA intro} 
We will develop a matrix product ansatz based on the so-called local moves of down-right paths (which will be defined in \eqref{eq:local moves intro} below), in order to solve for the stationary measure of the spin-$\frac{I}{2}$ vertex model on a strip. This matrix product ansatz directly generalizes the recent work \cite{Y} in the spin-$\frac{1}{2}$ ($I=1$) case to the higher spin cases. 

Assume $\left\{\mu_\hP\right\}$ is a collection of probability measures indexed by down-right paths $\hP$ on the strip, where each $\mu_{\hP}$ is supported on the set (with cardinality $(I+1)^N$) of all outgoing configurations of $\hP$. 
We will make the assumption that, for any pair of down-right paths $\hP$ and $\hQ$ such that $\hQ$ sits above $\hP$, the measure $\mu_{\hP}$ is updated to $\mu_{\hQ}$ under the evolution of the vertex model on a strip: For all $\tau'\in[[0,I]]^N$,
\be\label{eq:transition probability intro}\sum_{\tau\in[[0,I]]^N} P_{\hP,\mathcal{Q}}\lb\tau,\tau'\rb\mu_\hP(\tau)=\mu_{\mathcal{Q}}(\tau').\ee
By taking $\hQ=\Upsilon_1\hP$, we get that $\mu_{\hP}$ is the stationary measure of the spin-$\frac{I}{2}$ vertex model on $\hP$. 

We introduce three types of local moves of down-right paths: 
\be\label{eq:local moves intro} 
\begin{tikzpicture}[scale=0.8]
		\draw[dotted] (0,0) -- (1,0)--(1,1)--(0,1)--(0,0);
		\draw[ultra thick] (0,1) -- (0,0) -- (1,0);
		\end{tikzpicture}
  \quad 
\raisebox{5pt}{$\longmapsto$}
\quad 
\begin{tikzpicture}[scale=0.8]
		\draw[dotted] (0,0) -- (1,0)--(1,1)--(0,1)--(0,0);
		\draw[ultra thick] (0,1) -- (1,1) -- (1,0);
		\end{tikzpicture}
,\quad\quad
\begin{tikzpicture}[scale=0.8]
		\draw[dotted] (0,0) -- (0,1)--(-1,0)--(0,0);
            \draw[ultra thick] (0,0) -- (-1,0);
		\end{tikzpicture}
  \quad 
\raisebox{5pt}{$\longmapsto$}
\quad 
\begin{tikzpicture}[scale=0.8]
		\draw[dotted] (0,0) -- (0,1)--(-1,0)--(0,0);
		\draw[ultra thick] (0,0) -- (0,1);
		\end{tikzpicture}
,\quad\quad
\begin{tikzpicture}[scale=0.8]
		\draw[dotted] (0,0) -- (1,0)--(0,-1)--(0,0);
            \draw[ultra thick] (0,0) -- (0,-1);
		\end{tikzpicture}
  \quad 
\raisebox{5pt}{$\longmapsto$}
\quad 
\begin{tikzpicture}[scale=0.8]
		\draw[dotted] (0,0) -- (1,0)--(0,-1)--(0,0);
		\draw[ultra thick] (0,0) -- (1,0);
		\end{tikzpicture},
\ee
where the thick lines denote locally the down-right paths. As will be observed in Section \ref{subsec:proof of MPA}, 
condition \eqref{eq:transition probability intro} can be guaranteed by its special case where $\hQ$ is taken to be a local move of $\hP$.

We propose an ansatz of the form that $\mu_{\hP}$ could take:
Suppose $M_0^\u,\dots,M_I^\u$, $M_0^\r,\dots,M_I^\r$ are elements in a (possibly noncommutative) abstract algebra $\hA$ and $\ll W|\in H^*$ and $|V\rr\in H$ are two boundary vectors, where $H$ is a linear representation space of $\hA$. 
We define $\mu_\hP$ by the following matrix product states: 
\be\label{eq:matrix ansatz in introduction}
\mu_{\mathcal{P}}(\tau_1,\dots,\tau_N)=\frac{\ll W|M_{\tau_1}^{p_1}\times\dots\times M_{\tau_N}^{p_N}|V\rr}{\ll W|( \sum_{j=0}^IM^{p_1}_j)\times\dots\times (\sum_{j=0}^IM^{p_N}_j)|V\rr},
\ee
where $p_i\in\left\{\u,\r\right\}$, $1\leq i\leq N$ are outgoing edges of $\hP$ labeled from the up-left of $\hP$ to the down-right of $\hP$, and $0\leq\tau_1,\dots,\tau_N\leq I$ are occupation variables indicating the number of arrows occupying these edges. Three types of local moves \eqref{eq:local moves intro} provide the following consistency relations: 
For all $0\leq c,d\leq I$, we have
\be\label{eq:compatibility relations in introduction} 
    M_c^\u M_d^\r=\sum_{a,b=0}^I\R_{a,b}^{c,d}M_b^\r M_a^\u, \quad
    \ll W|M_d^\r=\sum_{a=0}^I\lb\lK_a^d\rb\ll W|M_a^\u, \quad
    M_c^\u|V\rr=\sum_{b=0}^I\lb\rK_b^c\rb M_b^\r|V\rr. 
\ee

We summarize this matrix product ansatz as the following theorem:

\begin{theorem}[Matrix product ansatz]\label{thm:general matrix ansatz higher spin}
In this paper  we will always use $\hA$ to denote a  (possibly noncommutative) algebra over $\mathbb{C}$, which admits a linear representation on a vector space $H$ over $\C$
with a finite or countable basis (the elements in $H$ are finite linear combinations of basis vectors). We use $H^*$ to denote the dual of $H$, i.e. the space of linear functions from $H$ to $\mathbb{C}$. 
We will implicitly identify elements of $\hA$ with elements in $\End(H)$, which is the space of linear transformations from $H$ to itself.

Assume the vertex weights $\R$, $\lK$ and $\rK$ satisfy \eqref{eq:stochasticity intro} and \eqref{eq:condition irreducible}.
Suppose  there are elements $M_{j}^\u$ and $M_{j}^\r$ for $0\leq j\leq I$ in $\hA$ 
and two boundary vectors $\ll W|\in H^*$ and $|V\rr\in H$, satisfying consistency relations  \eqref{eq:compatibility relations in introduction}. 
Assume that the denominator of \eqref{eq:matrix ansatz in introduction} is nonzero.

Consider the spin-$\frac{I}{2}$ stochastic vertex model on a strip with width $N$ and with sampling probabilities given by \eqref{eq:vertex weights higher spin}. Then for any down-right path $\hP$ on the strip with outgoing edges $p_1,\dots,p_N\in\left\{\u,\r\right\}$,  the matrix product ansatz \eqref{eq:matrix ansatz in introduction} 
gives the stationary measure of the spin-$\frac{I}{2}$ stochastic vertex model on a strip on the down-right path $\hP$,  
where $0\leq\tau_1,\dots,\tau_N\leq I$ are occupation variables.
\end{theorem}

The above theorem will be proved in Section \ref{subsec:proof of MPA}.
\begin{remark}
In the numerator of \eqref{eq:matrix ansatz in introduction},  $M_{\tau_1}^{p_1}\times\dots\times M_{\tau_N}^{p_N}\in\hA$ is implicitly identified to an element in $\End(H)$, which gives a scalar after pairing with $\ll W|$ and $|V\rr$. The denominator of \eqref{eq:matrix ansatz in introduction} is also a scalar.
\end{remark}

\subsection{The fused vertex model on a strip}
\label{subsec:fused vertex model on a strip intro}
For full and half space higher spin vertex models, one particular choice of vertex weights is 
 the most extensively studied. They are referred to as the `fused' vertex weights and are constructed through the `fusion procedure'. 
The fusion procedure allows for the construction of spin-$\frac{I}{2}$ solutions of the Yang-Baxter equation (referred to as fused $R$ matrices) and the reflection equation (referred to as fused $K$ matrices) from their spin-$\frac{1}{2}$ (i.e. unfused) counterparts. 
 The fusion procedure was introduced in the representation-theoretic context for $R$ matrices in \cite{karowski1979bound, kulish1981yang, jimbo1985q, kirillov1987exact, kulish2005quantum}, and for $K$ matrices in \cite{mezincescu1992fusion,kulish1992algebraic,frappat2007complete}.
In more recent years, explicit formulas for  fused matrices appear in the physics and probability literature, see for example \cite{mangazeev2014yang,kuniba2016stochastic,BM16,CP,borodin2018higher,borodin2017family,BW,BGW} for  fused $R$ matrices and \cite{ML,Jimmy_He} for  fused $K$ matrices.

The spin-$\frac{I}{2}$ vertex models with fused vertex weights (along with their “colored” or multi-species versions) exhibit a rich exactly solvable or integrable structure. As mentioned in Section~\ref{subsec:preface}, these models can also be analytically continued and degenerated into many other models in the KPZ class, including particle systems and polymer models; see, for example, \cite[Figure 1]{CP} and \cite[Figures 1 and 2]{kuan2018algebraic} for summaries in the full-space setting, and \cite{Jimmy_He} for an example in half-space.
The present paper focuses on the stationary measures of fused vertex models on a strip, prior to analytic continuation. In the future, however, the methods developed here may be extended to study the stationary measures of a broader class of open boundary models arising through degeneration by taking various parameter limits. Advancing in this direction would require a better understanding of (1) the analytic continuation of the general fused $K$ matrices and their degeneration to open boundary models, and (2) the analytic continuation of the fused matrix product ansatz—specifically, whether there exists a linear representation of the algebra in Proposition~\ref{prop:unfused ZF and GZ solution} that admits the analytic continuation of the matrices appearing in the ansatz. We defer this direction to future research.

To define the fused vertex model on a strip, in Section \ref{subsec:Fusion of vertex weights} we will introduce the fusion procedure in detail and define the fused $R$ and $K$ matrices.
We will only provide a brief introduction  in this subsection. 

We consider the fundamental solution   of the Yang-Baxter equation given as follows:
$$
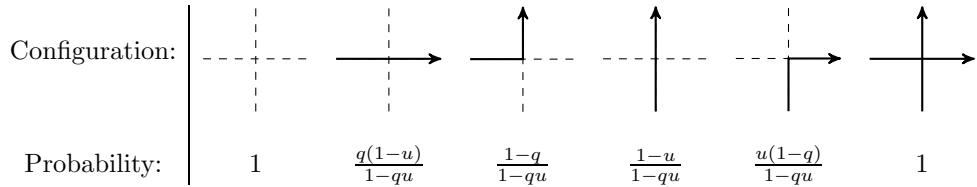
\begin{tabular}{c|cccccc}
Configuration:& \tikz{0.7}{
	\draw[dashed] (-1,0) -- (1,0);
	\draw[dashed] (0,-1) -- (0,1); 
} & \tikz{0.7}{
	\draw[dashed] (-1,0) -- (1,0);
	\draw[dashed] (0,-1) -- (0,1); 
        \draw[path, thick] (-1,0)--(1,0); 
} & \tikz{0.7}{
	\draw[dashed] (-1,0) -- (1,0);
	\draw[dashed] (0,-1) -- (0,1); 
        \draw[path, thick] (-1,0)--(0,0)--(0,1); 
} & \tikz{0.7}{
	\draw[dashed] (-1,0) -- (1,0);
	\draw[dashed] (0,-1) -- (0,1); 
         \draw[path, thick] (0,-1)--(0,1);
}& \tikz{0.7}{
	\draw[dashed] (-1,0) -- (1,0);
	\draw[dashed] (0,-1) -- (0,1); 
       \draw[path, thick] (0,-1)--(0,0)--(1,0); 
}& \tikz{0.7}{
	\draw[path, thick] (-1,0) -- (1,0);
	\draw[path, thick] (0,-1) -- (0,1); 
}\\ 
\\[0.01cm]
Probability:& $1$ & $\frac{q(1-u)}{1-qu}$ & $\frac{1-q}{1-qu}$ & $\frac{1-u}{1-qu}$ & $\frac{u(1-q)}{1-qu}$ & 1 \\ 
\end{tabular} 
$$
We refer to this solution   as the unfused $R$ matrix $\sR(u)=\sR^1(u)\in\End\lb\C^2\o\C^2\rb$, where $\sR(u)_{a,b}^{c,d}$ for $a,b,c,d\in\{0,1\}$ denotes the probability (given  above) that the numbers of outgoing arrows being  $c$ (upwards) and $d$ (rightwards), conditioned on the numbers of incoming arrows being $a$ (from below) and $b$ (from the left). We consider $\sR(u)$ as a matrix-valued function of the `spectral parameter' $u\in\C$, depending on   parameter $q$.

We also consider the following solution of the reflection equation: 
$$
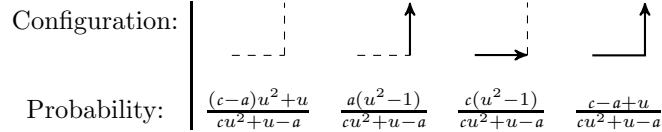
\begin{tabular}{c|cccc}
Configuration:& \tikz{0.7}{
	\draw[dashed] (-1,0) -- (0,0);
	\draw[dashed] (0,0) -- (0,1); 
} & \tikz{0.7}{
	\draw[dashed] (-1,0) -- (0,0);
	\draw[dashed] (0,0) -- (0,1); 
        \draw[path, thick] (0,0)--(0,1); 
} & \tikz{0.7}{
	\draw[dashed] (-1,0) -- (0,0);
	\draw[dashed] (0,0) -- (0,1); 
        \draw[path, thick] (-1,0)--(0,0);
} & \tikz{0.7}{
	\draw[path,thick] (-1,0) -- (0,0)--(0,1); 
}\\ 
\\[0.01cm]
Probability:& $\frac{(\ccc-\aaa)u^2+u}{\ccc u^2+u-\aaa}$ & $\frac{\aaa(u^2-1)}{\ccc u^2+u-\aaa}$ & $\frac{\ccc(u^2-1)}{\ccc u^2+u-\aaa}$ & $\frac{\ccc-\aaa+u}{\ccc u^2+u-\aaa}$ \\ 
\end{tabular} 
$$
We refer to this solution of the reflection equation as the unfused $K$ matrix $\sK(u)=\sK^1(u)\in\End\lb\C^2\rb$, where $\sK(u)_{b}^{c}$ for $b,c\in\{0,1\}$ denotes the probability (given above) that there are $c$ outgoing arrows when there are $b$ incoming arrows. We consider $\sK(u)$ as a matrix-valued function of the `spectral parameter' $u\in\C$, depending on the parameters $\ccc$ and $\aaa$.
 
The next result is well-known, which in particular appears in \cite{Vanicat_fused_MPA} under $(t^2,a,b,c,d)\mapsto(q,\aaa,\bbb,\ccc,\ddd)$.
\begin{proposition}\label{prop:YBE and RE}
We have the Yang-Baxter equation and the reflection equation:
$$
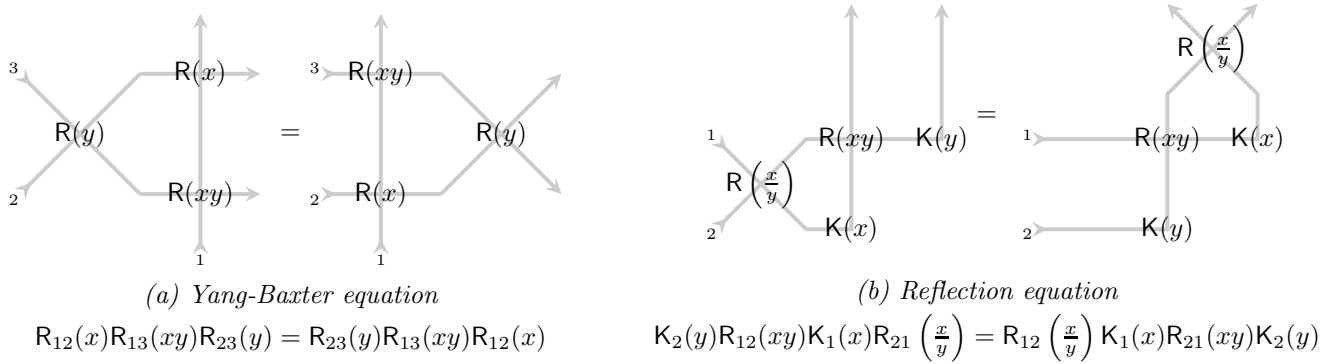
\tikz{0.8}{
\node at (-2.1,-1.1) {\tiny $2$};
\node at (-2.1,1.1) {\tiny $3$};
\node at (1,-2.1) {\tiny $1$};
\draw[lgray,line width=1.5pt,>-]  (-2,-1) -- (0,1); 
\draw[lgray,line width=1.5pt,->]  (0,1) -- (2,1); 
\draw[lgray,line width=1.5pt,>-]  (-2,1) -- (0,-1); 
\draw[lgray,line width=1.5pt,->]  (0,-1) -- (2,-1); 
\draw[lgray,line width=1.5pt,>->]  (1,-2) -- (1,2); 
\node at (-1,0) {$\sR(y)$};
\node at (1,1) {$\sR(x)$};
\node at (1,-1) {$\sR(xy)$}; 
\node at (2.5,0) {$=$};
\node at (2.5,-2.7) {(a)  Yang-Baxter equation};
\node at (2.5,-3.4) {$\sR_{12}(x)\sR_{13}(xy)\sR_{23}(y)=\sR_{23}(y)\sR_{13}(xy)\sR_{12}(x)$}; 
\node at (2.9,-1.1) {\tiny $2$};
\node at (2.9,1.1) {\tiny $3$};
\node at (4,-2.1) {\tiny $1$};
\draw[lgray,line width=1.5pt,>-]  (3,1)--(5,1);
\draw[lgray,line width=1.5pt,->]  (5,1)--(7,-1);
\draw[lgray,line width=1.5pt,>-]  (3,-1)--(5,-1);
\draw[lgray,line width=1.5pt,->]  (5,-1)--(7,1);
\draw[lgray,line width=1.5pt,>->]  (4,-2)--(4,2);
\node at (4,-1) {$\sR(x)$};
\node at (4,1) {$\sR(xy)$};
\node at (6,0) {$\sR(y)$}; 
}  
\quad\quad \quad 
\tikz{0.6}{
\node at (-1.1,1.1) {\tiny $1$};
\node at (-1.1,-1.1) {\tiny $2$};
\node at (5.9,1.1) {\tiny $1$};
\node at (5.9,-1.1) {\tiny $2$};
\draw[lgray,line width=1.5pt,>->] (-1,1)--(1,-1)--(2,-1)--(2,4);
\draw[lgray,line width=1.5pt,>->]  (-1,-1)--(1,1)--(4,1)--(4,4);
\draw[lgray,line width=1.5pt,>->]  (6,1)--(11,1)--(11,2)--(9,4);
\draw[lgray,line width=1.5pt,>->]  (6,-1)--(9,-1)--(9,2)--(11,4);
\node at (5,1.5) {$=$};
\node at (5,-2.4) {(b)  Reflection equation};
\node at (5,-3.4) {$\sK_2(y)\sR_{12}(xy)\sK_1(x)\sR_{21}\left(\frac{x}{y}\right)=\sR_{12}\left(\frac{x}{y}\right)\sK_1(x)\sR_{21}(xy)\sK_2(y)$};
\node at (0,0) {$\sR\lb\frac{x}{y}\rb$};
\node at (2,-1) {$\sK(x)$};
\node at (2,1) {$\sR(xy)$};
\node at (4,1) {$\sK(y)$};
\node at (9,-1) {$\sK(y)$};
\node at (9,1) {$\sR(xy)$};
\node at (11,1) {$\sK(x)$};
\node at (10,3) {$\sR\lb\frac{x}{y}\rb$}; 
}
$$ 
where the graphs above mean equalities of partition functions after fixing the external configurations (i.e. whether there are arrows in the external edges) and summing over all the possible internal configurations. Equivalently, they can also be understood as equalities of compositions of operators, which are given explicitly right below the graphs.  Each path in the graphs corresponds to a distinct copy of the space $\C^2$, labeled by the number at the starting point of the path. The operator $\sR_{ij}(u)$ means $\sR(u)\in\End\lb\C^2\o\C^2\rb$ acting on the $i$-th and $j$-th spaces and  $\sK_{i}(u)$ means $\sK(u)\in\End\lb\C^2\rb$ acting on the $i$-th space.
\end{proposition}

The fusion procedure for $R$ and $K$ matrices will be defined in detail in Section \ref{subsec:Fusion of vertex weights}, which   involve  taking collections of $I$ columns and $I$ rows, and view them as one column and row. One consider the composition of operators as shown in the graphs in part (a) and part (b) below in the $I=3$ case, which respectively corresponds to fusion for a bulk and a boundary vertex. Each intersection of two paths are placed with an $\sR(u)$ and the diagonals  (in the boundary case (b)) are placed with $\sK(u)$.
The spectral parameters are indicated in the graphs, which are chosen as $q$-geometric series on each row and column, from the right to the left, and from top to bottom, except   in case (b), where one needs to take  square roots at the diagonal.

$$
\tikz{0.8}{
\draw[lgray,line width=1.5pt,>->] (-2,0)--(2,0);
\draw[lgray,line width=1.5pt,>->] (-2,1)--(2,1);
\draw[lgray,line width=1.5pt,>->] (-2,-1)--(2,-1);
\draw[lgray,line width=1.5pt,>->] (0,-2)--(0,2);
\draw[lgray,line width=1.5pt,>->] (1,-2)--(1,2);
\draw[lgray,line width=1.5pt,>->] (-1,-2)--(-1,2);
\node at (0,0) {$u$};\node at (-1,1) {$u$};\node at (1,-1) {$u$};
\node at (0,1) {$uq^{-1}$};\node at (1,0) {$uq^{-1}$};
\node at (1,1) {$uq^{-2}$};
\node at (0,-1) {$uq$};\node at (-1,0) {$uq$};
\node at (-1,-1) {$uq^{2}$}; 
\node at (0,-2.5) {(a)  Fusion of bulk vertex};
}
\quad\quad\quad\quad\quad\quad\quad\quad\quad\quad\quad\quad
\tikz{0.9}{
\draw[lgray,line width=1.5pt,>->] (-2,-1)--(-1,-1)--(-1,2);
\draw[lgray,line width=1.5pt,>->] (-2,0)--(0,0)--(0,2);
\draw[lgray,line width=1.5pt,>->] (-2,1)--(1,1)--(1,2);
\node at (-1,1) {$u^2$};\node at (0,1) {$u^2q^{-1}$};\node at (-1,0) {$u^2q$};
\node at (1.2,1) {$\sqrt{u^2q^{-2}}$};\node at (0,0) {$\sqrt{u^2}$};
\node at (-1,-1) {$\sqrt{u^2q^2}$};
\node at (0,-2) {(b)  Fusion of boundary vertex};
}
$$
The state spaces of the combination of $I$ edges (where each edge can contain one arrow or no arrow) is $\lb\C^2\rb^{\o I}$. It turns out that this state space is too big for our purposes. There is a subspace of this state space (referred to as the $q$-exchangeable subspace), which can be identified with $\C^{I+1}$. We will prove that the above composed operators preserve this subspace, in the sense that if the incoming distributions from left and/or below are $q$-exchangeable, then so are the outgoing distributions to the right and/or above. 
By restricting the composed operators to the $q$-exchangeable subspaces, we are able to define the fused $R$ matrix $\sR^I(u)\in\End\lb\C^{I+1}\o\C^{I+1}\rb$ and fused $K$ matrix $\sK^I(u)\in\End\lb\C^{I+1}\rb$. 
We will also make use of another fused operator $\sbK^I(u)\in\End\lb\C^{I+1}\rb$, which is essentially a change of parameters of $\sK^I(u)$.

The next theorem gives explicit expressions for our fused $R$ matrices. This formula is cited from \cite{BGW} and was originally derived in \cite{BM16, kuniba2016stochastic}. 
\begin{theorem}
The fused $R$ matrices $\sR^I(u)\in\End\lb\CI\o\CI\rb$ (that will be defined in  Definition \ref{defn:fused R matrix and K matrix} in Section \ref{subsec:Fusion of vertex weights}) has the following explicit formula:
    For all $0\leq a,b,c,d\leq I$, we have:
    $$
    \sR^I(u)_{a,b}^{c,d}=\mathds{1}_{a+b=c+d}u^{d-b}q^{(d-a)I}\sum_{p=0}^{\min(b,c)}\Phi_{q^{-1}}\lb c-p,c+d-p;u,q^{I}u\rb\Phi_{q^{-1}}\lb p,b;q^{I}/u,q^{I}\rb,
     $$
     where $$\Phi_{q^{-1}}(i,j;x,y):=\left(\frac{y}{x}\right)^i\frac{\lb x;q^{-1}\rb_i\lb\frac{y}{x};q^{-1}\rb_{j-i}}{\lb y;q^{-1}\rb_j}\frac{\lb q^{-1};q^{-1}\rb_j}{\lb q^{-1};q^{-1}\rb_i\lb q^{-1};q^{-1}\rb_{j-i}},$$
     and we use $q$-Pochhammer symbol $(x;s)_n:=(1-x)(1-sx)\dots(1-s^{n-1}x)$ for $n\in\mathbb{N}_0$ (where $\mathbb{N}_0:=\{0\}\cup\z_+$).
\end{theorem}
\begin{proof}
    Observing from the fusion procedure in Section \ref{subsec:Fusion of vertex weights}, we notice that our $\sR^I(u)$ corresponds to the fused $R$ matrices in \cite{BGW} by $q\mapsto1/q$. The result follows from a specialization (for $N=2$ and $L=M=I$) of formula (6.2) in \cite{BGW} (this formula is originally due to \cite{BM16,kuniba2016stochastic}).
\end{proof}
\begin{remark}
There are other explicit formulas for $\sR^I(u)$. An explicit formula is given by the  ${}_4\overline{\phi}_3$ hyper-geometric functions in \cite{mangazeev2014yang} and \cite[Proposition 3.15]{CP}. Another explicit formula appears in \cite{Kuan_stochastic}.
\end{remark}
 
\begin{remark}
    An explicit formula for the fused $K$ matrices was obtained in \cite{ML}. Under a special parameter condition (when the $K$ matrices are upper-triangular), this formula was  proved by induction in \cite{Jimmy_He}. Our fusion procedure for the $K$ matrices is the essentially the same as in these works, modulo some change of parameters. We choose not to pursue the general formula for the fused $K$ matrices, since this is rather technically involved and orthogonal to the focus of this paper.
\end{remark}
\begin{remark}
    In the $I=2$ case, the fused matrices $\sR^2(u)$, $\sK^2(u)$ and $\sbK^2(u)$ are given by equations (3.27), (3.31) and (3.32) in \cite{Vanicat_fused_MPA}, where $z, t^2, a, b, c, d$ therein correspond to $u, q, \aaa, \bbb, \ccc, \ddd$ in this paper.
\end{remark}
\begin{remark}
    In our fusion of $K$ matrices, the spectral parameters  in the bulk take the form of $q$-geometric series involving  $u^2$, while on the diagonal, one needs to take their square roots. 
     This is because our reflection equation is slightly different from \cite[Proposition 2.5]{Jimmy_He}. Our notation is consistent with the literature on ZF and GZ relations  as mentioned in  Section \ref{subsec:preface}.
\end{remark}

Next we define the fused vertex model on a strip, whose bulk and boundary vertex weights are defined by specializing the spectral parameters in a  particular way in the fused operators $\sR^I(u)$, $\sK^I(u)$ and $\sbK^I(u)$. 
\begin{definition}[Fused vertex model on a strip]
\label{def: fused model intro}
    Assume $I\in\z_+$. Suppose we have a bulk parameter $q$, boundary parameters $\aaa,\bbb,\ccc,\ddd$ and a `spectral parameter' $\k$ satisfying:
\be\label{eq:condition irreducible intro}
    0<q<1,\quad0<\k<q^{\frac{I-1}{2}},\quad \aaa,\bbb,\ccc,\ddd>0,\quad \aaa-\ccc>q^{\frac{1-I}{2}}/\kappa,\quad\bbb-\ddd>q^{\frac{1-I}{2}}/\kappa.
\ee
Consider the fused $R$ matrix $\sR^I(u)\in\End\lb\C^{I+1}\o\C^{I+1}\rb$ and $K$ matrices $\sK^I(u), \sbK^I(u)\in\End\lb\C^{I+1}\rb$ (which are matrix-valued functions of $u\in\C$ depending on parameters $q,\aaa,\bbb,\ccc,\ddd$) that will be defined in  Definition \ref{defn:fused R matrix and K matrix} in Section \ref{subsec:Fusion of vertex weights}.
We define, for all $0\leq a,b,c,d\leq I$:
\be\label{eq:fused weights intro}
\R_{a,b}^{c,d}:=\sR^I\lb\k^2\rb_{b,a}^{d,c},\quad \lK_a^d:=\sK^I\lb\k\rb_a^d,\quad \rK_b^c:=\sbK^I(1/\k)_b^c.\ee 
The spin-$\frac{I}{2}$ vertex model on a strip   with these vertex weights will be referred to in this paper as the fused vertex model on a strip. Proposition \ref{prop:stochasticity irreducible} below guarantees that this model is stochastic and irreducible.
\end{definition}
\begin{remark}
    In our definition \eqref{eq:fused weights intro} of vertex weights in the fused vertex model, we are specializing the spectral parameters in a special way, and we are also swapping the indices in the bulk vertex weights. As will become clear later on, this choice is to make sure  that our  model is solvable by the fused matrix ansatz.
    \end{remark}
\begin{remark}
    It is possible that the fused vertex model on a strip could be obtained directly from some version of the inhomogeneous six-vertex model on a strip by treating every $I$ columns and $I$ rows in the model as one column and one row. However we choose not to adopt this approach in this paper.
\end{remark}
 
\begin{proposition}\label{prop:stochasticity irreducible}
    When the parameters $q,\k,\aaa,\bbb,\ccc,\ddd$ satisfy    \eqref{eq:condition irreducible intro}, the vertex weights $\R$, $\lK$ and $\rK$ defined by \eqref{eq:fused weights intro}
    satisfy conditions \eqref{eq:stochasticity intro}  and \eqref{eq:condition irreducible} (which guarantee stochasticity and irreducibility of the model).
\end{proposition}
 
  The above proposition will be proved in Section \ref{subsec:Proof of proposition stochasticity irreducible}.
 \begin{remark}
    Besides \eqref{eq:condition irreducible intro}, it is possible that there are other regions of the parameters on which the fused vertex weights are stochastic. We choose not to pursue this point and only study the model under \eqref{eq:condition irreducible intro}.
\end{remark}
\subsection{Fusion of ZF and GZ relations and stationary measure}
\label{subsec:fusion of ZF and GZ intro}
In order to study the stationary measure of the fused vertex model on a strip using the matrix product ansatz in Theorem \ref{thm:general matrix ansatz higher spin}, one needs to find concrete examples of
$\{M_{j}^\u,M_{j}^\r\}_{j=0}^I$, $\ll W|$ and $|V\rr$ that satisfy
the consistency relations:
\be\label{eq:consistency relations again}
    M_c^\u M_d^\r=\sum_{a,b=0}^I\R_{a,b}^{c,d}M_b^\r M_a^\u, \quad
    \ll W|M_d^\r=\sum_{a=0}^I\lb\lK_a^d\rb\ll W|M_a^\u, \quad
    M_c^\u|V\rr=\sum_{b=0}^I\lb\rK_b^c\rb M_b^\r|V\rr, 
\ee 
for the vertex weights $\R$, $\lK$ and $\rK$ given by \eqref{eq:fused weights intro}.
These relations look overwhelming to solve at first, however, since the vertex weights are obtained from the fusion procedure, our insight is to  
obtain the elements $\{M_{j}^\u,M_{j}^\r\}_{j=0}^I\subset\hA$ by defining a corresponding fusion procedure of the matrix product ansatz. Such fusion was developed in the physics work \cite{Vanicat_fused_MPA}, which was explicitly written only in the spin-$1$ ($I=2$) case but it actually works well in general spin. We will rigorously develop this fusion procedure in arbitrary spin.

More precisely, the fusion of matrix ansatz involves two steps. We first realize that the consistency relations \eqref{eq:consistency relations again} can be seen as a specialization (of the spectral parameters) in the so-called Zamolodchikov-Faddeev (ZF) and Ghoshal-Zamolodchikov (GZ) relations. It is known in the physics literature \cite{ZZ,F,SW97,GZ94,Review_and_progress} that these relations respectively govern the matrix  ansatz of  integrable systems in the bulk and at two  open boundaries. We then develop the fusion procedure of   ZF and GZ relations generalizing \cite{Vanicat_fused_MPA}.
The specialization of spectral parameters in the fused ZF and GZ relations give us \eqref{eq:consistency relations again}.

We first give the definition of the ZF and GZ relations:
\begin{definition}
Assume $\fV$ is a vector space over $\C$. Assume that $R(u)\in\End\lb\fV\o \fV\rb$ and $K(u),\bK(u)\in\End\lb\fV\rb$ are analytic functions of $u$, which respectively have singularities at finite subsets $\mathcal{P}_{R}$, $\mathcal{P}_{K}$ and $\mathcal{P}_{\bK}$ of $\C$. 
Suppose $\hA$ is an abstract algebra over $\mathbb{C}$ which admits a linear representation on a vector space $H$.  

   A function $\M(u)\in \fV\otimes\hA$ on $u\in\mathbb{C}^*:=\C\setminus\{0\}$ satisfies ZF relation with the $R$ matrix $R(u)$ if for all $x,y\in\mathbb{C}^*$ such that $x/y\notin\mathcal{P}_{R}$, we have:
    \be\label{eq: ZF relation}
    \M(y)\o\M(x)=\cR\left(\frac{x}{y}\right)\M(x)\o\M(y), 
    \ee
    where $\cR(u)=PR(u)$   and $P\in\End(\fV\otimes\fV)$ swaps the two factors of the space $\fV$. We remark that in the above ZF relation, the $R$ matrix only acts on   $\fV\o\fV$ component but leaves the $\hA$ factors alone, and the multiplication in the algebra $\hA$ is done implicitly so that both sides of \eqref{eq: ZF relation} are elements in $\fV\o \fV\o\hA$.
     
    Assume that $\ll W|\in H^*$ and $|V\rr\in H$, which we refer to as boundary vectors. 
    We say that $\M(u)$ satisfies the first GZ relation with the $K$ matrix $K(u)$  if for all $u\in\mathbb{C}^*\setminus\mathcal{P}_{K}$, 
    \be \label{eq:first GZ}
        \lw \left( K(u)\M\left(\frac{1}{u}\right)-\M(u) \right)=0. 
    \ee
    We say that $\M(u)$ satisfies the second GZ relation with the $K$ matrix $\bK(u)$  if for all $u\in\mathbb{C}^*\setminus\mathcal{P}_{\bK}$,
\be \label{eq:second GZ}
\left(\overline{K}(u)\M\left(\frac{1}{u}\right)-\M(u) \right) \rv = 0.
\ee
We remark that in the above GZ relations \eqref{eq:first GZ} and \eqref{eq:second GZ}, the $K$ matrices only act on the $\fV$ component but leaves $\hA$ alone, and both boundary vectors $\ll W|$ and $|V\rr$ are paired with the $\hA$ factor and leaves $\fV$ alone. 
\end{definition}
\begin{remark}
We will be interested in the case $\fV=\mathbb{C}^{I+1}$. Under its natural basis, the (vector-valued) function $\M(u)\in\fV\o\hA$ can be written explicitly as:
 $$\M(u)=[M_0(u),\dots,M_I(u)]^T,$$  where $M_j(u)$, $0\leq j\leq I$ are functions of $u\in\mathbb{C}^*$ with values in $\hA$. 
 
 The ZF relation \eqref{eq: ZF relation} can be written explicitly: For all $0\leq c,d\leq I$
 and all $x,y\in\C^*$ such that $x/y\notin\mathcal{P}_{R}$,
    \be \label{eq:ZF relation actually means}
   M_c(y)M_d(x)= \sum_{a,b=0}^I R_{b,a}^{d,c}\left(\frac{x}{y}\right)M_b(x)M_a(y),
    \ee  
where both sides of the equation are elements in $\hA$. We note that \eqref{eq:ZF relation actually means} does not involve tensor products between the $M$'s because, as explained below \eqref{eq: ZF relation}, the multiplication in the algebra $\hA$ is performed implicitly therein, so that both sides of \eqref{eq: ZF relation} are elements in $\fV \o \fV \o \hA$.

    The first GZ relation \eqref{eq:first GZ}  can  be written explicitly: For all $0\leq d\leq I$ and all $u\in\C^*\setminus\mathcal{P}_{K}$,
    \be\label{eq:first GZ relation actually means} 
       \ll W|M_d(u)= \sum_{a=0}^I K_a^d(u)\lw  M_a\left(\frac{1}{u}\right).  
    \ee

    The second GZ relation \eqref{eq:second GZ}  can  be written explicitly: For all $0\leq c\leq I$ and all $u\in\C^*\setminus\mathcal{P}_{\bK}$,
   \be\label{eq:second GZ relation actually means} 
   M_c(u)|V\rr= \sum_{b=0}^I\overline{K}_b^c(u)M_b\left(\frac{1}{u}\right) \rv.
    \ee
\end{remark}

The next result realizes  the consistency relations \eqref{eq:consistency relations again} as specializations of ZF and GZ relations.
 
\begin{proposition}
    \label{prop:solution of compatibility from ZF and GZ}
    Suppose $\fV=\mathbb{C}^{I+1}$.
    Suppose $\M(u)=[M_0(u),\dots,M_I(u)]^T\in \fV\otimes\hA$ satisfies the $\zf$ and $\gz$ relations   with the $R$ and $K$ matrices $R(u)$, $K(u)$ and $\bK(u)$. Then for any $\k\in\mathbb{C}^*$ such that $\k^2\notin\mathcal{P}_R$, $\k\notin\mathcal{P}_K$ and $1/\k\notin\mathcal{P}_{\bK}$, the following elements:
    $$M^\u_{j}=M_j\lb1/\k\rb,\quad
M^\r_{j}=M_j(\k),\quad 0\leq j\leq I$$
satisfy the consistency relations \eqref{eq:consistency relations again} with respect to the vertex weights:
$$
\R_{a,b}^{c,d}=R_{b,a}^{d,c}\lb\k^2\rb,\quad \lK_a^d=K_a^d(\k),\quad
\rK_b^c=\bK_b^c\lb1/\k\rb,\quad 0\leq a,b,c,d\leq I.
$$ 
\end{proposition}
\begin{proof}
    This can be seen by taking $x=\k$ and $y=1/\k$ in the ZF relation \eqref{eq:ZF relation actually means}, taking $u=\k$ in the first  GZ relation \eqref{eq:first GZ relation actually means} and taking $u=1/\k$ in the second   GZ relation \eqref{eq:second GZ relation actually means}.
\end{proof}

We want to find the solutions of ZF and GZ relations for the fused  matrices $\sR^I(u)$, $\sK^I(u)$ and $\sbK^I(u)$. We start with a solution  for the unfused matrices (i.e. when $I=1$) known in the physics literature, then we perform fusion and obtain such solutions for the  fused matrices.
\begin{proposition}[Section 2.4 in \cite{Vanicat_fused_MPA}, for $t^2\mapsto q$]\label{prop:unfused ZF and GZ solution}
Assume $\d,\e\in\hA$, $\ll W|\in H^*$ and $|V\rr\in H$ satisfy:
\be\label{eq:d and e in intro}
\d\e-q\e\d=1-q, \quad
    \ll W|\lb \aaa\e-\ccc\d+1\rb=0 ,\quad
    \lb \bbb\d-\ddd\e+1\rb|V\rr=0. 
\ee
Define $\sM_0(u)=u+\e$ and $\sM_1(u)=\frac{1}{u}+\d$. Then the vector-valued function 
$\sbM(u)=[\sM_0(u),\sM_1(u)]^T\in\C^2\o\hA$ 
satisfies the $\zf$ and $\gz$ relations with the unfused matrices $\sR(u)$, $\sK(u)$ and $\sbK(u)$ given in Definition \ref{def:unfused R and K matrices}.
\end{proposition}
\begin{remark}
    We will later make use of
    a concrete example of $\d$, $\e$, $\ll W|$ and $|V\rr$ satisfying \eqref{eq:d and e in intro}, which is given by the USW representation introduced in \cite{USW04}, see Proposition \ref{prop: USW representation} and Proposition \ref{prop:relations of d e compared with dehp}.
\end{remark} 

\begin{theorem}[Fusion of ZF and GZ relations]
\label{thm:Fusion of ZF and GZ relations}
    Suppose that $\sM_0(u)$ and $\sM_1(u)$ are given by Proposition \ref{prop:unfused ZF and GZ solution} above. For any $I\in\z_+$, we define the functions $\sM^I_{\zeta}(u)$, $0\leq\zeta\leq I$ on $u\in\C$ with values in $\hA$:
    \be\label{eq:fused solution ZF and GZ intro}
    \sM^I_{\zeta}(u):=\sum_{\substack{\zeta_1+\dots+\zeta_I=\zeta \\ \zeta_1,\dots,\zeta_I\in\left\{0,1\right\}}}\prod_{a\in[[1,I]]}^{\lra}\sM_{\zeta_a}\lb uq^{-\frac{I+1}{2}+a}\rb\in\hA,\ee
where the right arrow means that the product is taken from left to right in increasing order of  index $a\in[[1,I]]$. 
Then the vector $\sbM^I(u):=\left[\sM_0^I(u),\dots\sM_I^I(u)\right]^I\in \C^{I+1}\o\hA$ 
satisfies the $\zf$ and $\gz$ relations with the fused matrices $\sR^I(u)\in\End\lb\C^{I+1}\o\C^{I+1}\rb$ and $\sK^I(u), \sbK^I(u)\in\End\lb\C^{I+1}\rb$ defined in Definition \ref{defn:fused R matrix and K matrix}. 
\end{theorem}
\begin{remark}
    The above theorem will be proved in Section \ref{subsec:fusion of ZF and GZ relations proof of main theorem}. The key technical ingredient in this proof are alternative expressions for the fused $R$ and $K$ matrices in Theorem \ref{thm:braided version R K} referred to as the `braided' forms. We believe that  braid forms are interesting on their own and potentially have other uses (see Remark \ref{rmk:braided}). 
\end{remark}
 
\begin{remark}
    The unfused matrices $\sR(u)$, $\sK(u)$ and $\sbK(u)$ each have one or two singularities, as can be observed from \eqref{eq:unfused R and K matrix}. The singularities of the fused matrices $\sR^I(u)$, $\sK^I(u)$, and $\sbK^I(u)$ are precisely the finite set of points $u\in\C$ where certain unfused matrices in their definitions \eqref{eq:fusion R matrix}, \eqref{eq:fusion K matrix}, and \eqref{eq:fusion bar K matrix} become singular.
\end{remark}
We are now able to obtain a concrete matrix product ansatz for the fused vertex model on a strip:
\begin{theorem}[Stationary measure for the fused vertex model on a strip]
\label{thm:concrete matrix ansatz}
    Consider the fused vertex model on a strip with width $N$ in Definition \ref{def: fused model intro}, depending on parameters $q,\k,\aaa,\bbb,\ccc,\ddd$ satisfying: 
    \be\label{eq:condition irreducible intro again}
    0<q<1,\quad0<\k<q^{\frac{I-1}{2}},\quad \aaa,\bbb,\ccc,\ddd>0,\quad \aaa-\ccc>q^{\frac{1-I}{2}}/\kappa,\quad\bbb-\ddd>q^{\frac{1-I}{2}}/\kappa.\ee
On any down-right path $\hP$ with outgoing edges $p_i\in\left\{\u,\r\right\}$, $1\leq i\leq N$, the stationary measure is given by 
\be\label{eq:MPA again}\mu_{\mathcal{P}}(\tau_1,\dots,\tau_N)=\frac{\ll W|M_{\tau_1}^{p_1}\times\dots\times M_{\tau_N}^{p_N}|V\rr}{\ll W|( \sum_{j=0}^IM^{p_1}_j)\times\dots\times (\sum_{j=0}^IM^{p_N}_j)|V\rr},\ee
where  $0\leq\tau_1,\dots,\tau_N\leq I$ are occupation variables,  
$M_j^\u=\sM_j^I\lb1/\k\rb$ and $M_j^\r=\sM_j^I\lb\k\rb$  for $0\leq j\leq I$, where the functions $\sM_j^I(u)\in\hA$ of $u\in\C$ are given in Theorem \ref{thm:Fusion of ZF and GZ relations} above.
\end{theorem}
\begin{remark}
    We remark that the elements $M_j^\u$ and $M_j^\r$ used in the matrix product ansatz above are degree $I$ polynomial-like expressions of $\d$ and $\e$ from \eqref{eq:d and e in intro} (where $\d$ and $\e$ may not be commutative). 
\end{remark}
\begin{proof}
    By Theorem \ref{thm:Fusion of ZF and GZ relations}, $\sbM^I(u):=\left[\sM_0^I(u),\dots\sM_I^I(u)\right]^T$ satisfies the ZF and GZ relations with the fused matrices $\sR^I(u)$, $\sK^I(u)$ and $\sbK^I(u)$. By Proposition \ref{prop:solution of compatibility from ZF and GZ}, the elements $M_j^\u$ and $M_j^\r$ appearing in the matrix product ansatz satisfy the consistency relations \eqref{eq:consistency relations again} with the vertex weights $\R$, $\lK$ and $\rK$ defined in Definition \ref{def: fused model intro}. 
    By Theorem \ref{thm:general matrix ansatz higher spin}, the  unique stationary measure of the system is given by \eqref{eq:MPA again}.
\end{proof}
 
\subsection{Stationary measure in terms of the Askey--Wilson process} 
\label{subsec:AW intro}
To study the asymptotic behavior of the stationary measure of the fused vertex model using the matrix product ansatz provided in Theorem \ref{thm:concrete matrix ansatz} above, we will largely adopt a particular method of handling the matrix ansatz which was developed within a line of research \cite{DEHP93,USW04,BW17} concerning the stationary measure of another model known as the open asymmetric simple exclusion process (ASEP). This method involves expressing the joint generating function of the stationary measure in terms of expectations of an auxiliary Markov process, commonly known as the Askey--Wilson process due to its connections with the so-called Askey--Wilson orthogonal polynomials. 

We will provide an introduction to the background of the Askey--Wilson process in Section \ref{subsec:backgrounds of askey wilson} and to certain aspects of this method in Section \ref{subsec:Matrix product states askey wilson},  however we will not introduce the open ASEP model itself since it is not explicitly needed for our purposes. We refer the interested reader to a nice survey  paper \cite{Corwin_survey}.

The next theorem expresses the joint generating function of the stationary measure for the fused vertex model on a strip in terms of expectations of the Askey--Wilson process:

 \begin{theorem}\label{thm:stationary measure of fused model as AW}
Consider the fused vertex model on a strip defined in Definition \ref{def: fused model intro}, with parameters $q,\k,\aaa,\bbb,\ccc,\ddd$ satisfying \eqref{eq:condition irreducible intro again}, where we recall that $q,\aaa,\bbb,\ccc,\ddd$ are parameters in the fused $R$ and $K$ matrices, and $\k$ plays the role of spectral parameters. 
We will use an alternative parametrization of the model by $q,\k,A,B,C,D$ defined in Definition \ref{def:alternate parametrization of fused model}. 
We assume that $AC<1$.
We use $\lb Y_t\rb_{ t>0 }$ to denote the  Askey--Wilson Markov process with parameters $(A, B, C, D, q)$, which will be defined in Section \ref{subsec:backgrounds of askey wilson}.

For any down-right path $\hP$ on the strip, recall that the outgoing edges of $\hP$ are labelled by $p_1,\dots,p_N\in\left\{\u,\r\right\}$ from the up-left  of $\hP$ to the down-right  of $\hP$.
We   define a set of numbers $\v_1,\dots,\v_N\in\left\{-1,1\right\}$   as $\v_i=1$ if $p_i=\u$ and $\v_i=-1$ if $p_i=\r$ for $1\leq i\leq N$.

The joint generating function of  stationary measure $\mu_{\hP}$ of the fused vertex model on a strip on the down-right path $\hP$ has the following expression in terms of the Askey--Wilson process: For any $0<t_1\leq\dots\leq t_N$:
\be\label{eq:askey wilson formula intro}
\mathbb{E}_{\mu_{\hP}}\lbe\prod_{i=1}^Nt_i^{\tau_i}\rbe=\frac{\mathbb{E}\lbe\prod_{i=1}^N\prod_{a=1}^I\lb 2\sqrt{t_i}Y_{t_i}+t_iq^{\frac{I+1}{2}-a}\k^{\v_i}+q^{-\frac{I+1}{2}+a}\k^{-\v_i}\rb\rbe}{\mathbb{E}\lbe\prod_{a=1}^I\lb2Y_1+q^{\frac{I+1}{2}-a}\k+q^{-\frac{I+1}{2}+a}\k^{-1}\rb^N\rbe},
\ee
where $0\leq\tau_1,\dots,\tau_N\leq I$ are occupation variables indicating the number of arrows occupying edges $p_1,\dots,p_N$,
and the expectations on the RHS of \eqref{eq:askey wilson formula intro}   are expectations under the Askey--Wilson process $\lb Y_t\rb_{ t>0  }$.
\end{theorem}

The above theorem will be proved in Section \ref{subsec:proof of stationary measure as AW}.
 
\subsection{Mean arrow density and phase diagram}
\label{subsec:mean density intro}
We study the asymptotic behavior of a basic (macroscopic) statistical quantity of the stationary measure as the system size $N\rightarrow\infty$. 
This statistical quantity is known as the `mean arrow density' defined as the total number of arrows within the system divided by the system size $N$. This quantity is parallel to the `mean particle density' in the context of open ASEP.  
We will show that, similar to the case of open ASEP (and of the six-vertex model on a strip in \cite{Y}), the limit of the mean arrow density of  fused vertex model on a strip exhibits a phase diagram involving the high density phase, the low density phase and the maximal current phase. We observe a novel and interesting phenomenon: corresponding to the (sequences of) down-right paths with different limit shapes, the stationary measures share the same phase diagram but have different limits for the mean arrow density.  
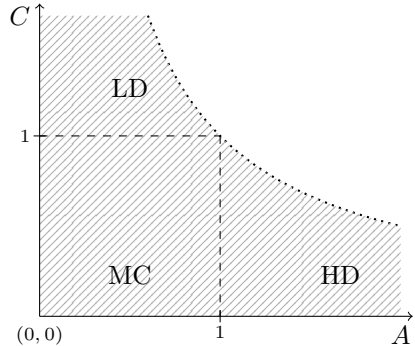
\begin{figure}[h]
\centering
\begin{tikzpicture}[scale=0.8]
\draw[scale = 1,domain=6.8:11,smooth,variable=\x,dotted,thick] plot ({\x},{1/((\x-7)*1/3+2/3)*3+5});
\fill[pattern=north east lines, pattern color=gray!60] (5,5)--(5,10) -- plot [domain=6.8:11]  ({\x},{1/((\x-7)*1/3+2/3)*3+5}) -- (11,5) -- cycle;
 \draw[->] (5,5) to (5,10.2);
 \draw[->] (5.,5) to (11.2,5);
   \draw[-, dashed] (5,8) to (8,8);
   \draw[-, dashed] (8,8) to (8,5); 
   \node [left] at (5,8) {\scriptsize$1$};
   \node[below] at (8,5) {\scriptsize $1$};
     \node [below] at (11,5) {$A$};
   \node [left] at (5,10) {$C$};
  \draw[-] (8,4.9) to (8,5.1);
   \draw[-] (4.9,8) to (5.1,8);
 \node [below] at (5,5) {\scriptsize$(0,0)$};
    \node [above] at (6.5,8.5) {LD}; 
    \node [below] at (10,6) {HD};  
 \node [below] at (6.5,6) {MC};
\end{tikzpicture} 
\caption{Phase diagram of the fused vertex model on a strip. LD, HD, MC respectively stand for
low density, high density and maximal current (phases).}
\label{fig:phase diagram}
\end{figure}
\begin{theorem}\label{thm:mean arrow density}
    Consider the fused vertex model on a strip in Definition \ref{def: fused model intro}. We fix all the parameters of the system but vary the system size (width of the strip) $N$. 
    Assume $AC<1$ and $A,C\notin\{q^{-l}:l\in\mathbb{N}_0\}$. 
    
    Consider a sequence of down-right paths $\{\hP^N\}_{N=1}^\infty$, where each $\hP^N$ lies on the strip with width $N$. We denote by $0\leq\phi_N\leq N$ the total number of horizontal edges within the down-right path  $\hP^N$. We assume that as $N\rightarrow\infty$, the sequence $\phi_N/N$ converges to some number $\lambda\in[0,1]$. 
    Then as $N\rightarrow\infty$, the limit of the mean arrow density of the stationary measure $\mu_{\hP^N}$ of fused vertex model on a strip on  $\hP^N$ is given by:
    \be\label{eq:mean density in theorem}\lim_{N\rightarrow\infty}\mathbb{E}_{\mu_{\hP^N}}\lbe\frac{1}{N}\sum_{i=1}^N\tau_i\rbe=\lambda G(\k)+(1-\lambda)G\lb\frac{1}{\k}\rb,\ee
    where $G$ is the following function:
    \be\label{eq:function eta intro}G(x)=
\begin{cases}
\sum_{a=1}^I \frac{x}{x+q^{-\frac{I+1}{2}+a}} , & A<1, C<1 \text{(maximal current phase)},\\
\sum_{a=1}^I \frac{Ax}{Ax+q^{-\frac{I+1}{2}+a}} , & A>1 \text{ (high density phase)},\\
 \sum_{a=1}^I \frac{x}{x+Cq^{-\frac{I+1}{2}+a}}, & C>1 \text{ (low density phase)}.
 \end{cases}\ee
\end{theorem}

The above theorem will be proved in Section \ref{subsec:mean arrow density and phase diagram}.
\begin{remark}
We note that our phase diagram (Figure \ref{fig:phase diagram}) corresponds to the phase diagram in \cite{Y} for the six-vertex model on a strip (i.e. the $I=1$ case) but with a different parameterization. 
\end{remark}
\begin{remark}
    We have made the assumption $AC<1$ in Theorem \ref{thm:stationary measure of fused model as AW} and Theorem \ref{thm:mean arrow density}. This corresponds to the shaded area in the phase diagram (Figure \ref{fig:phase diagram}), which we refer to as the `fan region' of the fused vertex model on a strip. This constraint is needed to guarantee the existence of the Askey--Wilson Markov process. 
    
    In a recent paper \cite{wang2023askey}, an extension of this method was provided for open ASEP. Instead of expectations of Askey--Wilson processes, the stationary measure of open ASEP was expressed therein in the `shock region' $AC>1$ as integrations with respect to the so-called Askey--Wilson signed measures.  Using such an expression, the density profiles and fluctuations of open ASEP were obtained therein in the shock region. 
    It is possible that, by similar methods, the stationary measure for the fused vertex model on a strip could  be characterized in the shock region $AC>1$ by Askey--Wilson signed measures, from which the limits of mean arrow density could be studied. We choose not to pursue this direction  
    and leave it for future works.
\end{remark}

\subsection{Integrable discrete-time two-step Floquet dynamics}
\label{subsec: Integrable discrete-time two-step Floquet dynamics intro}
We demonstrate that one of the `two-step Floquet dynamics' models introduced in \cite{Vanicat_fused_MPA} can be considered as a special case of the fused vertex model on a strip defined in this paper (Definition \ref{def:interacting particle system} and Definition \ref{def: fused model intro} above). 

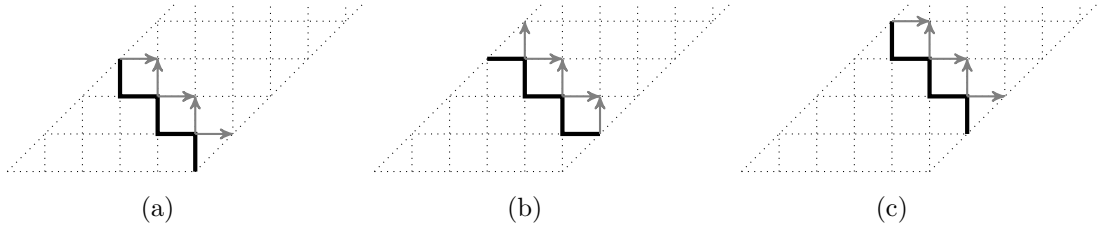
\begin{figure}[h]
\centering
\begin{tikzpicture}[scale=0.5]
\draw[dotted] (0,0)--(4.5,4.5);
\draw[dotted] (0,0)--(5,0);
\draw[dotted] (1,1)--(6,1);
\draw[dotted] (2,2)--(7,2);
\draw[dotted] (3,3)--(8,3);
\draw[dotted] (4,4)--(9,4);
\draw[dotted] (1,0)--(1,1);
\draw[dotted] (2,0)--(2,2);
\draw[dotted] (3,0)--(3,3);
\draw[dotted] (4,0)--(4,4);
\draw[dotted] (5,0)--(5,4.5);
\draw[dotted] (6,1)--(6,4.5);
\draw[dotted] (7,2)--(7,4.5);
\draw[dotted] (8,3)--(8,4.5);
\draw[dotted] (9,4)--(9,4.5);
\draw[dotted] (5,0)--(9.5,4.5);
\draw[ultra thick] (3,3)--(3,2)--(4,2)--(4,1)--(5,1)--(5,0);
\draw[path, gray] (3,3)--(4,3);
\draw[path, gray] (4,2)--(4,3);
\draw[path, gray] (4,2)--(5,2);
\draw[path, gray] (5,1)--(5,2);
\draw[path, gray] (5,1)--(6,1);
\node at (4,-1) {(a)};
\end{tikzpicture}
\begin{tikzpicture}[scale=0.5]
\draw[dotted] (0,0)--(4.5,4.5);
\draw[dotted] (0,0)--(5,0);
\draw[dotted] (1,1)--(6,1);
\draw[dotted] (2,2)--(7,2);
\draw[dotted] (3,3)--(8,3);
\draw[dotted] (4,4)--(9,4);
\draw[dotted] (1,0)--(1,1);
\draw[dotted] (2,0)--(2,2);
\draw[dotted] (3,0)--(3,3);
\draw[dotted] (4,0)--(4,4);
\draw[dotted] (5,0)--(5,4.5);
\draw[dotted] (6,1)--(6,4.5);
\draw[dotted] (7,2)--(7,4.5);
\draw[dotted] (8,3)--(8,4.5);
\draw[dotted] (9,4)--(9,4.5);
\draw[dotted] (5,0)--(9.5,4.5);
\draw[ultra thick] (3,3)--(4,3)--(4,2)--(5,2)--(5,1)--(6,1);
\draw[path, gray] (4,3)--(4,4);
\draw[path, gray] (4,3)--(5,3);
\draw[path, gray] (5,2)--(5,3);
\draw[path, gray] (5,2)--(6,2);
\draw[path, gray] (6,1)--(6,2);
\node at (4,-1) {(b)};
\end{tikzpicture}
\begin{tikzpicture}[scale=0.5]
\draw[dotted] (0,0)--(4.5,4.5);
\draw[dotted] (0,0)--(5,0);
\draw[dotted] (1,1)--(6,1);
\draw[dotted] (2,2)--(7,2);
\draw[dotted] (3,3)--(8,3);
\draw[dotted] (4,4)--(9,4);
\draw[dotted] (1,0)--(1,1);
\draw[dotted] (2,0)--(2,2);
\draw[dotted] (3,0)--(3,3);
\draw[dotted] (4,0)--(4,4);
\draw[dotted] (5,0)--(5,4.5);
\draw[dotted] (6,1)--(6,4.5);
\draw[dotted] (7,2)--(7,4.5);
\draw[dotted] (8,3)--(8,4.5);
\draw[dotted] (9,4)--(9,4.5);
\draw[dotted] (5,0)--(9.5,4.5);
\draw[ultra thick] (4,4)--(4,3)--(5,3)--(5,2)--(6,2)--(6,1);
\draw[path, gray] (4,4)--(5,4);
\draw[path, gray] (5,3)--(5,4);
\draw[path, gray] (5,3)--(6,3);
\draw[path, gray] (6,2)--(6,3);
\draw[path, gray] (6,2)--(7,2);
\node at (4,-1) {(c)};
\end{tikzpicture} 
\caption{Two-step update for a down-right zig-zag path, when $N=5$. The thick paths are the down-right paths, and the gray arrows denote the outgoing edges on those paths.}
\label{fig:Two-step update on a down-right zig-zag path}
\end{figure}

Suppose that the width $N$ of the strip is an odd number, and that the down-right path $\hP$ is the down-right zig-zag path starting from a vertical edge, see (a) in Figure \ref{fig:Two-step update on a down-right zig-zag path} above for the $N=5$ case. The transition at each time $k\mapsto k+1$ in the fused vertex model on a strip updates the down-right path from (a) to (c). This can be understood as a two-step update: first from (a) to (b), and then from (b) to (c). Therefore the transition matrix of this Markov chain can be written as (where $\mathcal{V}=\C^{I+1}$ is the state space of each edge):
\be\label{eq:two step}\UU^{e}\UU^{o}
=B_1 U_{23}U_{45}\dots U_{N-1,N}\, U_{12}U_{34} \dots U_{N-2,N-1} \overline{B}_N\in\End\lb\mathcal{V}^{\o N}\rb,\ee
where 
\be\label{eq:two steps explicit}\UU^{o} = \prod_{k=1}^{(N-1)/2} U_{2k-1,2k}\overline{B}_N,\quad \UU^{e} =  B_1 \prod_{k=1}^{(N-1)/2} U_{2k,2k+1},\ee
correspond to the two steps of the updates discussed above, 
\be\label{eq:weights two step}U=\R P \in\End(\mathcal{V}\o\mathcal{V}),\quad B=\lK\in\End(\mathcal{V}),\quad \overline{B}=\rK\in\End(\mathcal{V}),\ee
are the local operators (where $\R$, $\lK$ and $\rK$ are the fused vertex weights and $P$ swaps the two factors of $\mathcal{V}$),
and $U_{i,j}$ (resp. $B_i$ or $\overline{B}_i$) in \eqref{eq:two step} and \eqref{eq:two steps explicit} above denote the operators $U$ (resp. $B$ or $\overline{B}$) acting on the $(i,j)$-th (resp. $i$-th) copies of $\mathcal{V}$ in $\mathcal{V}^{\o N}$. By \eqref{eq:weights two step} and Definition \ref{def: fused model intro}, we have
 \be\label{eq:weights two}U=\R P=P\sR^I(\k^2),\quad B=\lK=\sK^I(\k),\quad \overline{B}=\rK=\sbK^I\lb1/\k\rb.\ee
To summarize, on the down-right zig-zag path, the Markov chain of the fused vertex model on a strip  can be alternatively defined by transfer matrices by \eqref{eq:two step} and \eqref{eq:weights two}. This Markov chain has been previously studied by \cite{Vanicat_fused_MPA} for $I\in\{1,2\}$ and referred to therein as the `two-step Floquet dynamics' (in the open boundary and asymmetric case). See Section 2.2 and Section 3.1 of \cite{Vanicat_fused_MPA} for the definition, in particular (3.35) and (3.36) therein correspond to \eqref{eq:weights two} and \eqref{eq:two step}. The fusion procedure of vertex weights and of the matrix product ansatz was explicitly given therein in the $I=2$ case.
The partition function and the mean current of the stationary measure was calculated  by the matrix ansatz, see (3.62) and (3.63) of \cite{Vanicat_fused_MPA}.  

As a special case of Theorem \ref{thm:mean arrow density} for the down-right zig-zag path $\hP$, we are able to obtain the asymptotics of  mean density and the phase diagram of the two-step Floquet dynamics with open boundaries (in the asymmetric case). This in particular answers an open question in \cite[Section 4]{Vanicat_fused_MPA}.
\begin{theorem}
Consider the integrable two-step Floquet dynamics with open boundaries (in the asymmetric case) introduced by \cite{Vanicat_fused_MPA}, which is an interacting particle system on the lattice $\{1,\dots,L\}$, where up to $I\in\{1,2\}$ many particles can occupy the same site.  The definition appears in \cite{Vanicat_fused_MPA} in pages 306-307, together with pages 312-313 for $I=1$ case and pages 320-323 for $I=2$ case. The parameters $t^2,\k,a,b,c,d$ in \cite{Vanicat_fused_MPA} correspond to our parameters $q,\k,\aaa,\bbb,\ccc,\ddd$.

As the size of the system $L\rightarrow\infty$, the limit of the mean particle density is given by \eqref{eq:mean density in theorem} for $\lambda=1/2$, and the phase diagram   is given by Figure \ref{fig:phase diagram}, where we recall the definition of $A$ and $C$ in Definition \ref{def:alternate parametrization of fused model}.
\end{theorem}

\subsection*{Outline of the paper}
In Section \ref{sec:2} we define the fusion for $R$ and $K$ matrices and  fusion for ZF and GZ relations. The fused $R$ and $K$ matrices allow the definition of  fused vertex model on a strip in Definition \ref{def: fused model intro}, and the fused ZF and GZ relations allow the construction of its stationary measure by matrix product ansatz in Theorem \ref{thm:concrete matrix ansatz}. In Section \ref{sec:3} we first prove Theorem \ref{thm:stationary measure of fused model as AW} expressing the stationary measure in terms of Askey--Wilson processes. We then use this expression to obtain limits of mean density in Theorem \ref{thm:mean arrow density}.

\subsection*{Acknowledgements} 
The author thanks his advisor, Ivan Corwin, for suggesting this problem and for helpful discussions. The author thanks Amol Aggarwal, Ivan Corwin, Zoe Himwich, and Alisa Knizel for carefully reading the draft and providing valuable suggestions.
The author was supported by Ivan Corwin’s NSF grant DMS-1811143 as well as the Fernholz Foundation’s “Summer Minerva Fellows” program.

\section{Higher spin vertex model on a strip and matrix ansatz}
\label{sec:2}
In Section \ref{subsec:Fusion of vertex weights} we introduce the fusion procedure and define the fused $R$ and $K$ matrices. In Section \ref{subsec:Proof of proposition stochasticity irreducible} we prove Proposition \ref{prop:stochasticity irreducible} that guarantee the stochasticity and irreducibility of the fused vertex model on a strip. In Section \ref{subsec:fusion of ZF and GZ relations proof of main theorem} we develop the fusion procedure for ZF and GZ relations, proving Theorem \ref{thm:Fusion of ZF and GZ relations}. In Section \ref{subsec:proof of MPA} we prove Theorem \ref{thm:general matrix ansatz higher spin} on the matrix product ansatz for higher spin vertex models on a strip.

\subsection{Fusion of vertex weights} 
\label{subsec:Fusion of vertex weights}  
In this subsection, we introduce the fusion procedure and define the  spin-$\frac{I}{2}$ fused $R$ and $K$ matrices from their spin-$\frac{1}{2}$ (unfused) counterparts.

\begin{definition}\label{def:fusion procedure}
We first define some notions that will be useful in fusion.
\begin{enumerate}
\item [(1)] We denote the natural basis of $\CI$ by $\{\eb_0,\dots,\eb_I\}$. 
We identify $\CI$   with the state space of an edge of   spin-$\frac{I}{2}$ vertex model, with $\eb_j$ corresponding to the state with $j$ arrows occupying that edge.
    \item [(2)] We define a surjection $\Pi:\CoI\rightarrow \CI$ by $$\Pi(\eb_{a_1}\o\dots\o\eb_{a_I}):=\eb_{\sum a_j}$$ for all $a_1,\dots,a_I\in\left\{0,1\right\}$. 
    \item [(3)] A vector 
    \be \label{eq: a general vector in VoI}
    v=\sum_{a_1,\dots,a_I\in\left\{0,1\right\}}v(a_1,\dots,a_I)\eb_{a_1}\o\dots\o\eb_{a_I}\in\CoI
    \ee
    is called $q$-exchangeable if its coefficients satisfy
    $$
    qv(\dots,0,1,\dots)= v(\dots,1,0,\dots)
    $$
    on any two nearby positions. The subspace of $q$-exchangeable vectors in $\CoI$ is called the $q$-exchangeable subspace, denoted by $\sy$. 
    We can regard the coefficients $v(a_1,\dots,a_I)$ in vector \eqref{eq: a general vector in VoI} as a $\mathbb{C}$-valued measure on $\left\{0,1\right\}^I$. This measure is $q$-exchangeable if $v$ is $q$-exchangeable.
    \item [(4)] For any $I$-tuple $(a_1,\dots,a_I)\in\left\{0,1\right\}^I$, define
 $$
\inv(a_1,\dots, a_I)=\#\{1\leq i<j\leq I: a_i>a_j\},\quad \tinv(a_1,\dots, a_I)=\#\{1\leq i<j\leq I: a_i<a_j\}.
 $$
It is well-known that, for $0\leq a\leq I$,
$$
Z_{q}(I;a)=\sum_{\substack{a_1,\dots, a_I\in\left\{0,1\right\}\\\sum a_j=a}}q^{\inv(a_1,\dots, a_I)}=\sum_{\substack{a_1,\dots, a_I\in\left\{0,1\right\}\\\sum a_j=a}}q^{\tinv(a_1,\dots, a_I)}=\frac{(q;q)_I}{(q;q)_a (q;q)_{I-a}}.
$$
    \item [(5)] We define an injection $\Pih: \CI\rightarrow\CoI$ by
     $$
    \Pih(\eb_a):=\frac{1}{Z_{q}(I;a)}\sum_{\substack{a_1,\dots, a_I\in\left\{0,1\right\}\\\sum a_j=a}}q^{\inv(a_1,\dots, a_I)}\eb_{a_1}\o\dots\o\eb_{a_I},
     $$
   for  all $0\leq a\leq I$.
    This injection maps $\CI$ onto the $q$-exchangeable subspace $\sy$ of $\CoI$. 
    
    We can observe that $\Pi\circ\Pih=\id_{\CI}$, and that 
    $\Pih\circ\Pi=F\in\End\lb\CoI\rb$ is the projector (i.e. $F^2=F$) onto the subspace $\sy$.
\end{enumerate} 
\end{definition}   
\begin{definition}\label{def:unfused R and K matrices}
    We will use the following unfused (spin-$\frac{1}{2}$) $R$ and $K$ matrices
    $\sR(u)\in\End\lb\Co\o\Co\rb$ and $\sK(u),\sbK(u)\in\End\lb\Co\rb$ with parameters $q,\aaa,\bbb,\ccc,\ddd$:   
\be \label{eq:unfused R and K matrix}
\sR(u)=\begin{bmatrix}
    1 & 0 & 0 & 0\\
    0 & \frac{q(1-u)}{1-qu} & \frac{u(1-q)}{1-qu} & 0 \\
    0 & \frac{1-q}{1-qu} & \frac{1-u}{1-qu} & 0\\
    0 & 0 & 0 & 1
\end{bmatrix},
\quad
\sK(u) = \begin{bmatrix}
          \frac{(\ccc-\aaa)u^2+u}{\ccc u^2+u-\aaa} & \frac{\ccc(u^2-1)}{\ccc u^2+u-\aaa} \\
          \frac{\aaa(u^2-1)}{\ccc u^2+u-\aaa} & \frac{\ccc-\aaa+u}{\ccc u^2+u-\aaa}
         \end{bmatrix}, \quad 
  \sbK(u) = \begin{bmatrix}
          \frac{(\bbb-\ddd)u^2-u}{\bbb u^2-u-\ddd} & \frac{\bbb(u^2-1)}{\bbb u^2-u-\ddd} \\
          \frac{\ddd(u^2-1)}{\bbb u^2-u-\ddd} & \frac{\bbb -\ddd-u}{\bbb u^2-u-\ddd}
         \end{bmatrix}.
\ee
The $R$ matrix $\sR(u)$ above is written under the basis $\left\{\eb_0\o\eb_0, \eb_0\o\eb_1, \eb_1\o\eb_0, \eb_1\o\eb_1\right\}$ of $\Co\o\Co$. The $K$ matrices $\sK(u)$ and $\sbK(u)$ above are written under the basis $\left\{\eb_0,\eb_1\right\}$ of $\Co$.
\end{definition}

From the unfused operators $\sR(u)$, $\sK(u)$ and $\sbK(u)$ above, we will next define the fused operators $\soR(u)\in\End\lb\CoI\o\CoI\rb$ and $\soK(u),\sobK(u)\in\End\lb\CoI\rb$. These are not yet the  spin-$\frac{I}{2}$ fused  $R$ and $K$ matrices that we will make use of, because they are not yet projected to   $q$-exchangeable subspaces.  
\begin{definition}\label{defn:fusion operators}
    We write left/right arrows on product signs to mean that the products are taken from left to right in decreasing/increasing orders of the index. When there are two product signs next to each other we always perform the inner one first.
    We write $\sR_{ij}(u)$ to mean $\sR(u)\in\End\lb\Co\o\Co\rb$ acting on the $i$-th and $j$-th factors of $\CoI$. We write $\sK_i(u)$ (resp. $\sbK_i(u)$) to mean $\sK(u)\in\End\lb\Co\rb$ (resp. $\sbK(u)\in\End\lb\Co\rb$) acting on the $i$-th factor of $\CoI$.
    We define the   permutation operator  $P_{\omega}\in\End\lb\CoI\rb$ corresponding to any $\omega$ in the symmetric group of $\{1,\dots,I\}$.  

We define the fused $R$ operator $\soR(u)\in\End\lb\CoI\o\CoI\rb$ as:
    \be\label{eq:fusion R matrix}
    \soR(u)=\prod_{a\in[[1,I]]}^{\longleftarrow}\prod_{b\in[[1,I]]}^{\longrightarrow}\sR_{b,a+I}\lb uq^{b-a}\rb=\tikz{0.8}{
         
        \foreach\x in {0,...,3}{
		\draw[lgray,line width=1.5pt,>->] (\x,-1) -- (\x,4);
	}
        
        \foreach\y in {0,...,3}{
		\draw[lgray,line width=1.5pt,>->] (-1,\y) -- (4,\y);
	}
         
        \node[below] at (0,-1) {\small $I$};
        \node[below] at (3,-1) {\small$1$};
        \node[below] at (1.5,-1) {\small$\dots$};
        \node[above] at (0,4)  {\small$I$};
        \node[above] at (3,4)  {\small$1$};
        \node[above] at (1.5,4) {\small$\dots$};
         
	\node[left] at (-1,0)  {\small$1$};
        \node[left] at (-1,3)  {\small$I$};
        \node[left] at (-1,1.5)  {\small$\vdots$};
        \node[right] at (4,0)  {\small$1$};
        \node[right] at (4,3)  {\small$I$};
        \node[right] at (4,1.5)  {\small$\vdots$};
       
        \node at (0,0)  {\small$uq^{I-1}$};
        \node at (0,2)  {\small$uq$};
        \node at (1,2)  {\small$u$};
        \node at (3,0)  {\small $u$};
        \node at (3,1)  {\small $uq^{-1}$};
        \node at (2,0)  {\small $uq$};
        \node at (2,1)  {\small$u$};
        \node at (0,3)  {\small$u$};
        \node at (1,3)  {\small$uq^{-1}$};
        \node at (3,3) {\small $uq^{1-I}$};
}  
    \ee
In the above graph the vertical (or horizontal) path labelled $i$ represent the $i$-th (or $(I+i)$-th) factor in $\CoI\o\CoI$ respectively, for $1\leq i\leq I$. The vertex at the intersection of a vertical path labelled $i$ and a horizontal path labelled $j$ represents an  operator $\sR_{i,j+I}$, with spectral parameter given in the graph.

    We define the fused $K$ operator $\soK(u)\in\End\lb\CoI\rb$ as:
    \be\label{eq:fusion K matrix}
        \soK(u)= P_{\binom{1,\dots,I}{I,\dots,1}}\circ\lb\prod_{a\in[[1,I]]}^{\longleftarrow}\sK_a\lb uq^{\frac{I+1}{2}-a}\rb
        \prod_{b\in[[1,a-1]]}^{\longleftarrow}R_{b,a}\lb u^2q^{I+1-a-b}\rb\rb=P_{\binom{1,\dots,I}{I,\dots,1}}\circ\tikz{0.8}{
	\draw[lgray,line width=1.5pt,>->] (-1,0)--(0,0)--(0,4);
	\draw[lgray,line width=1.5pt,>->] (-1,1)--(1,1)--(1,4);
	\draw[lgray,line width=1.5pt,>->] (-1,2)--(2,2)--(2,4);
	\draw[lgray,line width=1.5pt,>->] (-1,3)--(3,3)--(3,4);
        
        \node[left] at (-1,0) {\small $1$};
        \node[left] at (-1,1) {\small $2$};
        \node[left] at (-1,3) {\small $I$};
	\node[left] at (-1,2) {\small $\vdots$};
        
        \node[above] at (0,4) {\small $1$};
        \node[above] at (1,4) {\small $2$};
        \node[above] at (3,4) {\small $I$};
        \node[above] at (2,4) {\small $\dots$};
        
        \node at (0.2,0) {\small $\sqrt{ u^2 q^{I-1}}$};
        \node at (3.2,3) {\small $\sqrt{ u^2 q^{1-I}}$};
        \node at (0,3) {\small $u^2$};
        \node at (0,2) {\small$u^2q$};
        \node at (1,3) {\small $u^2q^{-1}$};
}
    \ee
In the above graph the path labelled $i$ (which first goes rightwards, hit the diagonal and then goes upwards) represent the $i$-th factor in $\CoI$, for $1\leq i\leq I$. The vertex at the intersection of the paths labelled $i$ and $j$ represents an operator $\sR_{i,j}$, for $1\leq i<j\leq I$. The vertex at the turning point of the path labelled by $i$ represents an  operator $\sK_i$. The spectral parameters of these operators are given in the graph.

We define the fused $K$ operator $\sobK(u)\in\End\lb\CoI\rb$ as:
    \be\label{eq:fusion bar K matrix}
\sobK(u)=  P_{\binom{1,\dots,I}{I,\dots,1}}\circ\lb\prod_{a\in[[1,I]]}^{\longrightarrow}\sbK_a\lb uq^{\frac{I+1}{2}-a}\rb
\prod_{b\in[[a+1,I]]}^{\longrightarrow}\sR^{-1}_{b,a}\lb u^2q^{I+1-a-b}\rb\rb=P_{\binom{1,\dots,I}{I,\dots,1}}\circ\tikz{0.8}{
	\draw[lgray,line width=1.5pt,>->] (-1,0)--(0,0)--(0,4);
	\draw[lgray,line width=1.5pt,>->] (-1,1)--(1,1)--(1,4);
	\draw[lgray,line width=1.5pt,>->] (-1,2)--(2,2)--(2,4);
	\draw[lgray,line width=1.5pt,>->] (-1,3)--(3,3)--(3,4);
         
        \node[left] at (-1,0) {\small $I$};
        \node[left] at (-1,2) {\small $2$};
        \node[left] at (-1,3) {\small $1$};
	\node[left] at (-1,1) {\small $\vdots$};
        
        \node[above] at (0,4) {\small $I$};
        \node[above] at (2,4) {\small $2$};
        \node[above] at (3,4) {\small $1$};
        \node[above] at (1,4) {\small $\dots$};
        
        \node at (0.2,0) {\small $\sqrt{u^2q^{1-I}}$};
        \node at (3.2,3) {\small $\sqrt{u^2q^{I-1}}$};
        \node at (0,3) {\small $u^2$};
        \node at (0,2) {\small $u^2q^{-1}$};
        \node at (1,3) {\small $u^2q$};
}
    \ee
In the above graph the path labelled $i$  represent the $i$-th factor in $\CoI$, for $1\leq i\leq I$. The vertex at the intersection of paths labelled $i$ and $j$ represents an  operator $\sR^{-1}_{i,j}$, for $1\leq j<i\leq I$. The vertex at the turning point of path labelled by $i$ represents an operator $\sbK_i$. Spectral parameters are given in the graph.
\end{definition}
\begin{lemma}\label{lem:q-exchangeble} The operator $\soR(u)$ has invariant subspace $\sy\o\sy$, and $\soK(u)$ and $\sobK(u)$ have invariant subspace $\sy$.
\end{lemma}
\begin{proof}
    This is a standard result that was proved for the $R$ matrix, for example, in \cite[Appendix B]{BW}, and for the $K$ matrix in \cite[Section 6.2]{Jimmy_He}, although they appear in slightly different notations.  For this reason, we will provide a sketch of the proof.

    We first sketch the proof for the $K$ matrix. To prove that the $q$-exchangeable subspace $\sy$ is invariant under $\soK(u)$, we need to show that if a $q$-exchangeable probability distribution enters the graph in \eqref{eq:fusion K matrix}, then the outgoing distribution is again $q$-exchangeable.
    We remark that $q$-exchangeability depends on the labeling of edges, and the operator $P_{\binom{1,\dots,I}{I,\dots,1}}$ in \eqref{eq:fusion K matrix} re-labels the outgoing edges at the end.
    Note that 
 $$
\sR(q)_{0,1}^{1,0}=\sR(q)_{1,0}^{1,0}=\frac{1}{q+1},\quad
\sR(q)_{0,1}^{0,1}=\sR(q)_{1,0}^{0,1}=\frac{q}{q+1},
 $$
i.e., for a Yang–Baxter vertex with spectral parameter $q$, the outgoing distribution is always $q$-exchangeable, regardless of the incoming distribution.
We can insert a Yang–Baxter vertex $\sR(q)$ on the left at rows $(i, i+1)$ without changing the distribution entering the graph. We then move this vertex to the right and subsequently to the top, using a Yang–Baxter or reflection equation in Proposition \ref{prop:YBE and RE} at each step. In the end, there is a Yang–Baxter vertex $\sR(q)$ at the top at columns $(i, i+1)$. 
After re-labeling the outgoing edges at the end using the operator $P_{\binom{1,\dots,I}{I,\dots,1}}$, the distribution at outgoing edges $(I - i, I - i + 1)$ is $q$-exchangeable. Letting $i$ range from $1$ to $I - 1$, we conclude that the full outgoing distribution of the operator $\soK(u)$ is $q$-exchangeable. 
This procedure for the case $I = 3$, $i = 1$ is shown in the picture below:
 $$
\tikz{0.9}{
	\draw[lgray,line width=1.5pt,>->] (-1,0)--(-0.75,0)--(-0.25,1)--(1,1)--(1,3);
	\draw[lgray,line width=1.5pt,>->] (-1,1)--(-0.75,1)--(-0.25,0)--(0,0)--(0,3);
	\draw[lgray,line width=1.5pt,>->] (-1,2)--(2,2)--(2,3);
       \node at (0,2) {\small $u^2$};
        \node at (0,1) {\small $u^2q$};
        \node at (1,2) {\small $u^2q^{-1}$};
        \node at (0,0) {\small $uq$};
        \node at (2,2) {\small $uq^{-1}$};
        \node at (1,1) {\small $u$};
        \node at (-0.5,0.5) {\small $q$}; 
}
\quad=\quad\tikz{0.9}{
	\draw[lgray,line width=1.5pt,>->] (-1,0)--(0,0)--(0,1.25)--(1,1.75)--(1,3);
	\draw[lgray,line width=1.5pt,>->] (-1,1)--(1,1)--(1,1.25)--(0,1.75)--(0,3);
	\draw[lgray,line width=1.5pt,>->] (-1,2)--(2,2)--(2,3);
         \node at (0,2) {\small $u^2$};
        \node at (0,1) {\small $u^2q$};
        \node at (1,2) {\small $u^2q^{-1}$};
        \node at (0,0) {\small $u$};
        \node at (2,2) {\small $uq^{-1}$};
        \node at (1,1) {\small $uq$};
        \node at (0.5,1.5) {\small $q$};
}
\quad=\quad\tikz{0.9}{
	\draw[lgray,line width=1.5pt,>->] (-1,0)--(0,0)--(0,2.25)--(1,2.75)--(1,3);
	\draw[lgray,line width=1.5pt,>->] (-1,1)--(1,1)--(1,2.25)--(0,2.75)--(0,3);
	\draw[lgray,line width=1.5pt,>->] (-1,2)--(2,2)--(2,3);
         \node at (0,2) {\small $u^2q^{-1}$};
        \node at (0,1) {\small $u^2q$};
        \node at (1,2) {\small $u^2$};
        \node at (0,0) {\small $u$};
        \node at (2,2) {\small $uq^{-1}$};
        \node at (1,1) {\small $uq$};
        \node at (0.5,2.5) {\small $q$};
}
 $$

The proof for the $R$ matrix is similar to the one above, using only the Yang--Baxter equation and not the reflection equation.
An example of this procedure for $I=3$ is shown in the picture below:
$$
\tikz{0.8}{
       
        \draw[lgray,line width=1.5pt,>->] (0,-1) -- (0,-0.75)--(1,-0.25)-- (1,3);
        \draw[lgray,line width=1.5pt,>->] (1,-1) -- (1,-0.75)--(0,-0.25)-- (0,3);
        \draw[lgray,line width=1.5pt,>->] (2,-1) -- (2,3);
       
        \draw[lgray,line width=1.5pt,>->] (-1,0) -- (3,0);
        \draw[lgray,line width=1.5pt,>->] (-1,1) -- (3,1);
        \draw[lgray,line width=1.5pt,>->] (-1,2) -- (3,2); 
       
        \node at (0,0)  {\small$uq^{2}$};
        \node at (1,0) {\small $uq$};
        \node at (2,0) {\small $u$};
         \node at (0,1)  {\small$uq$};
        \node at (1,1) {\small $u$};
        \node at (2,1) {\small $uq^{-1}$};
         \node at (0,2)  {\small$u$};
        \node at (1,2) {\small $uq^{-1}$};
        \node at (2,2) {\small $uq^{-2}$};
        \node at (0.5,-0.5) {\small $q$};
}  
\quad=\quad
\tikz{0.8}{
       
        \draw[lgray,line width=1.5pt,>->] (0,-1) --(0,0.25)--(1,0.75)-- (1,3);
        \draw[lgray,line width=1.5pt,>->] (1,-1) --(1,0.25)--(0,0.75)-- (0,3);
        \draw[lgray,line width=1.5pt,>->] (2,-1) -- (2,3);
        
        \draw[lgray,line width=1.5pt,>->] (-1,0) -- (3,0);
        \draw[lgray,line width=1.5pt,>->] (-1,1) -- (3,1);
        \draw[lgray,line width=1.5pt,>->] (-1,2) -- (3,2); 
        
        \node at (0,0)  {\small$uq$};
        \node at (1,0) {\small $uq^2$};
        \node at (2,0) {\small $u$};
         \node at (0,1)  {\small$uq$};
        \node at (1,1) {\small $u$};
        \node at (2,1) {\small $uq^{-1}$};
         \node at (0,2)  {\small$u$};
        \node at (1,2) {\small $uq^{-1}$};
        \node at (2,2) {\small $uq^{-2}$};
        \node at (0.5,0.5) {\small $q$};
}  
\quad=\quad
\tikz{0.8}{
       
        \draw[lgray,line width=1.5pt,>->] (0,-1) --(0,1.25)--(1,1.75)-- (1,3);
        \draw[lgray,line width=1.5pt,>->] (1,-1) --(1,1.25)--(0,1.75)-- (0,3);
        \draw[lgray,line width=1.5pt,>->] (2,-1) -- (2,3);
        
        \draw[lgray,line width=1.5pt,>->] (-1,0) -- (3,0);
        \draw[lgray,line width=1.5pt,>->] (-1,1) -- (3,1);
        \draw[lgray,line width=1.5pt,>->] (-1,2) -- (3,2); 
   
        \node at (0,0)  {\small$uq$};
        \node at (1,0) {\small $uq^2$};
        \node at (2,0) {\small $u$};
         \node at (0,1)  {\small$u$};
        \node at (1,1) {\small $uq$};
        \node at (2,1) {\small $uq^{-1}$};
         \node at (0,2)  {\small$u$};
        \node at (1,2) {\small $uq^{-1}$};
        \node at (2,2) {\small $uq^{-2}$};
        \node at (0.5,1.5) {\small $q$};
}  
\quad=\quad
\tikz{0.8}{
        
        \draw[lgray,line width=1.5pt,>->] (0,-1) --(0,2.25)--(1,2.75)-- (1,3);
        \draw[lgray,line width=1.5pt,>->] (1,-1) --(1,2.25)--(0,2.75)-- (0,3);
        \draw[lgray,line width=1.5pt,>->] (2,-1) -- (2,3);
      
        \draw[lgray,line width=1.5pt,>->] (-1,0) -- (3,0);
        \draw[lgray,line width=1.5pt,>->] (-1,1) -- (3,1);
        \draw[lgray,line width=1.5pt,>->] (-1,2) -- (3,2); 
      
        \node at (0,0)  {\small$uq$};
        \node at (1,0) {\small $uq^2$};
        \node at (2,0) {\small $u$};
         \node at (0,1)  {\small$u$};
        \node at (1,1) {\small $uq$};
        \node at (2,1) {\small $uq^{-1}$};
         \node at (0,2)  {\small$uq^{-1}$};
        \node at (1,2) {\small $u$};
        \node at (2,2) {\small $uq^{-2}$};
        \node at (0.5,2.5) {\small $q$};
}  
$$

 Regarding the proof for the $\bK$ matrix, observe from \eqref{eq:unfused R matrices} and \eqref{eq:unfused K matrices} that
$\sR^{-1}(u) = \sR(u)_{q \mapsto 1/q}$ and $\sbK(u) = \sK(u)_{\aaa \mapsto -\ddd, \bbb \mapsto -\ccc}$.
Hence, the result follows from the corresponding result for the $K$ matrix by the change of variables $\aaa \mapsto -\ddd$, $\bbb \mapsto -\ccc$, and $q \mapsto 1/q$.
This concludes the proof. 
\end{proof}

\begin{definition}\label{defn:fused R matrix and K matrix}
    Define  the fused matrices $\sR^I(u)\in\End\lb\CI\o \CI\rb$ and $\sK^I(u),\sbK^I(u)\in\End\lb\CI\rb$:
$$\sR^I(u)=(\Pi_1 \otimes \Pi_2)\circ\soR(u)\circ(\Pih_1 \otimes \Pih_2),\quad \sK^I(u)=\Pi\circ\soK(u)\circ\Pih,\quad \sbK^I(u)=\Pi\circ\sobK(u)\circ\Pih,$$
    i.e.
    \begin{align}
        &\sR^I(u)=(\Pi_1 \otimes \Pi_2)\circ\lb\prod_{a\in[[1,I]]}^{\longleftarrow}\prod_{b\in[[1,I]]}^{\longrightarrow}\sR_{b,a+I}\lb uq^{b-a}\rb\rb\circ(\Pih_1  \otimes \Pih_2 ),\label{eq: actual fusion R}\\
        &\sK^I(u)=\Pi\circ P_{\binom{1,\dots,I}{I,\dots,1}} \circ\lb\prod_{a\in[[1,I]]}^{\longleftarrow}\sK_a\lb uq^{\frac{I+1}{2}-a}\rb
        \prod_{b\in[[1,a-1]]}^{\longleftarrow}\sR_{b,a}\lb u^2q^{I+1-a-b}\rb\rb\circ\Pih,\label{eq: actual fusion K}\\
        &\sbK^I(u)=\Pi\circ P_{\binom{1,\dots,I}{I,\dots,1}}\circ\lb\prod_{a\in[[1,I]]}^{\longrightarrow}\sbK_a\lb uq^{\frac{I+1}{2}-a}\rb
\prod_{b\in[[a+1,I]]}^{\longrightarrow}\sR^{-1}_{b,a}\lb u^2q^{I+1-a-b}\rb\rb\circ\Pih.\label{eq: actual fusion bK} 
    \end{align}

    Here and below, for a linear operator $\Psi$ from a vector space $\mathcal{U}$ to a vector space $\mathcal{W}$, we write $\Psi_1 \otimes \Psi_2$ to denote the linear operator from $\mathcal{U} \otimes \mathcal{U}$ to $\mathcal{W} \otimes \mathcal{W}$, where $\Psi_i$ acts on the $i$-th tensor factor for $i = 1, 2$.
\end{definition}
\begin{remark} 
One can observe from the fusion procedure that    $\sbK^I(u)=\sK^I(u)_{q\mapsto 1/q,\aaa\mapsto-\ddd,\bbb\mapsto-\ccc}$. 
\end{remark}

\subsection{Stochasticity and irreducibility of the model: Proof of Proposition \ref{prop:stochasticity irreducible}}
\label{subsec:Proof of proposition stochasticity irreducible}
In this subsection we prove Proposition \ref{prop:stochasticity irreducible} which states that under the following condition on parameters: 
$$0<q<1,\quad0<\k<q^{\frac{I-1}{2}},\quad \aaa,\bbb,\ccc,\ddd>0,\quad \aaa-\ccc>q^{\frac{1-I}{2}}/\kappa,\quad\bbb-\ddd>q^{\frac{1-I}{2}}/\kappa,$$
  the vertex weights for the fused vertex model:  
$$\R_{a,b}^{c,d}:=\sR^I\lb\k^2\rb_{b,a}^{d,c},\quad \lK_a^d:=\sK^I\lb\k\rb_a^d,\quad \rK_b^c:=\sbK^I(1/\k)_b^c,\quad \text{ for all } 0\leq a,b,c,d\leq I,$$
satisfy the conditions \eqref{eq:stochasticity intro} and \eqref{eq:condition irreducible}, i.e. that the matrices $\R$, $\lK$ and $\rK$ are stochastic, and that all of their entries are strictly positive unless $a+b\neq c+d$ in the matrix $\R$, in which case the entry equals $0$.  
We recall that the matrices $\sR^I\lb\k^2\rb$, $\sK^I\lb\k\rb$ and $\sbK^I(1/\k)$ have been defined by fusion.
One can observe that we only need to prove these properties for all of the (unfused) $R$ and $K$ matrices that appear  as vertices in the graphs \eqref{eq:fusion R matrix} \eqref{eq:fusion K matrix} and \eqref{eq:fusion bar K matrix} in Definition \ref{defn:fusion operators}.
We begin by recalling the unfused matrices:
\begin{align} 
    &\sR(u)=\begin{bmatrix}
    1 & 0 & 0 & 0\\
    0 & q\frac{1-u}{1-qu} & u\frac{1-q}{1-qu} & 0 \\
    0 & \frac{1-q}{1-qu} & \frac{1-u}{1-qu} & 0\\
    0 & 0 & 0 & 1
\end{bmatrix},\quad
\sR^{-1}(u)=\begin{bmatrix}
    1 & 0 & 0 & 0\\
    0 & \frac{1-u}{q-u} & u\frac{q-1}{q-u} & 0 \\
    0 & \frac{q-1}{q-u} & q\frac{1-u}{q-u} & 0\\
    0 & 0 & 0 & 1
\end{bmatrix},\label{eq:unfused R matrices}\\
&\sK(u) = \begin{bmatrix}
          \frac{(\ccc-\aaa)u^2+u}{\ccc u^2+u-\aaa} & \frac{\ccc(u^2-1)}{\ccc u^2+u-\aaa} \\
          \frac{\aaa(u^2-1)}{\ccc u^2+u-\aaa} & \frac{\ccc-\aaa+u}{\ccc u^2+u-\aaa}
         \end{bmatrix}, \quad 
  \sbK(u) = \begin{bmatrix}
          \frac{(\bbb-\ddd)u^2-u}{\bbb u^2-u-\ddd} & \frac{\bbb(u^2-1)}{\bbb u^2-u-\ddd} \\
          \frac{\ddd(u^2-1)}{\bbb u^2-u-\ddd} & \frac{\bbb -\ddd-u}{\bbb u^2-u-\ddd}
         \end{bmatrix}. \label{eq:unfused K matrices}
\end{align}

         We first prove the following claim for the unfused $R$ and $K$ matrices:
\begin{claim}\label{claim:in the proof of stochasticity} Suppose $q\in(0,1)$ and $\aaa,\bbb,\ccc,\ddd>0$. We have the following:
         \begin{enumerate}
              \item [$\bullet$] The middle $4$ entries of $\sR(u)$ are positive if $u\in(0,1)$. 
              \item [$\bullet$] The middle $4$ entries of $\sR^{-1}(u)$ are positive if $u>1$. 
              \item [$\bullet$] The entries of $\sK(u)$ are positive if $u\in(0,1)$ and $\aaa-\ccc>1/u$.
              \item [$\bullet$] The entries of $\sbK(u)$ are positive if $u>1$ and $\bbb-\ddd>u$.
         \end{enumerate}
     \end{claim}
     \begin{proof}[Proof of Claim \ref{claim:in the proof of stochasticity}]
         The first and second can be directly observed. When $\aaa-\ccc>1/u$ and $u\in(0,1)$ we have $(\ccc-\aaa)u+1<0$. Hence $(\ccc-\aaa)u^2+u<0$, $\ccc u^2+u-\aaa<\ccc u^2+u-\aaa u^2<0$ and $\ccc-\aaa+u<\ccc-\aaa+1/u<0$. Moreover $\aaa(u^2-1)<0$ and $\ccc(u^2-1)<0$, therefore the entries of $K(u)$ are positive. When $\bbb-\ddd>u$ and $u>1$ we have $\bbb-\ddd-u>0$, $\bbb u^2-u-\ddd>\bbb-\ddd-u>0$ and $(\bbb-\ddd)u^2-u>u^3-u>0$. Moreover $\bbb(u^2-1)>0$ and $\ddd(u^2-1)>0$, therefore the entries of $\bK(u)$ are positive. We conclude the proof.
     \end{proof}

We record the matrices appearing as vertices of $\sR^I(\kappa^2)$, $\sK^I(\kappa)$ and $\sbK^I(\kappa)$ as follows:
\begin{enumerate}
    \item [$\bullet$] In $\sR^I(\k^2)$ there are $\sR(u)$ for $u=\k^2q^i$, $1-I\leq i\leq I-1$.
    \item [$\bullet$] In $\sK^I(\k)$ there are $\sR(u)$ for $u=\k^2q^i$, $2-I\leq i\leq I-2$ and $\sK(u)$ for $u=\k q^i$, $\frac{1-I}{2}\leq i\leq\frac{I-1}{2}$.
    \item [$\bullet$] In $\sK^I(1/\k)$ there are $\sR^{-1}(u)$ for $u=q^i/\k^2$, $2-I\leq i\leq I-2$ and $\sbK(u)$ for $u=q^i/\k$, $\frac{1-I}{2}\leq i\leq\frac{I-1}{2}$.
\end{enumerate}
Since $\kappa\in(0,q^{\frac{I-1}{2}})$, those $\sR(u)$ that appear have $u\in(0,1)$ and those $\sR^{-1}(u)$ that appear have $u>1$.  
Since $\aaa-\ccc>q^{\frac{1-I}{2}}/\kappa$ and $\bbb-\ddd>q^{\frac{1-I}{2}}/\kappa$, those $\sK(u)$ that appear have $u\in(0,1)$, $\aaa-\ccc>1/u$ and those $\sbK(u)$ that appear have $u>1$, $\bbb-\ddd>u$. By Claim \ref{claim:in the proof of stochasticity} we conclude the proof.

\subsection{Fusion of ZF and GZ relations: Proof of Theorem \ref{thm:Fusion of ZF and GZ relations}} \label{subsec:fusion of ZF and GZ relations proof of main theorem}
In this subsection we prove Theorem \ref{thm:Fusion of ZF and GZ relations} on the fusion of ZF and GZ relations. 

In the proof we will need the following alternative form of the fused $R$ and $K$ matrices.

\begin{theorem}\label{thm:braided version R K} 
We denote by $P_{I,I}$ the operator that swaps the first and second factors in $\CI\o \CI$
or in $\CoI\o\CoI$. Let $\scR^I(u)=P_{I,I}\sR^I(u)$ and $\soRb(u)=P_{I,I}\soR(u)$. Recall that $\scR(u)=P\sR(u)\in\End\lb\Co\o \Co\rb$.
We write $\scR_a(u)$ to mean   operator $\scR(u)$ acting on the $(a,a+1)$ factors of $\CoI\o\CoI$, for $1\leq a\leq 2I-1$. We write $\sK_a(u)$ or $\sbK_a(u)$ to mean   operator $\sK(u)$ or $\sbK(u)$ acting on $a$-th factor of $\CoI$, for $1\leq a\leq I$. Then we have the following expressions of the fused operators introduced in Definition \ref{defn:fusion operators}:
\begin{align}
   &\soRb(u)=\prod_{a\in[[1,I]]}^{\longleftarrow}\prod_{b\in[[a,a+I-1]]}^{\longrightarrow}\scR_b\lb uq^{b+1-2a}\rb,\label{eq:braided fusion oR}\\
   &\soK(u)=\prod^{\longleftarrow}_{a\in[[1,I]]}\sK_1\lb uq^{\frac{I+1}{2}-a}\rb\prod_{b\in[[1,a-1]]}^{\longleftarrow}\scR_{a-b}\lb u^2q^{I+1-a-b}\rb,\label{eq:braided fusion oK}\\
   &\sobK(u)=\prod_{a\in[[1,I]]}^{\longrightarrow}\sbK_I\lb uq^{\frac{I+1}{2}-a}\rb\prod_{b\in[[a+1,I]]}^{\longrightarrow}\scR^{-1}_{I+a-b}\lb u^2q^{I+1-a-b}\rb.\label{eq:braided fusion obK} 
\end{align}  

As a result, the fused $R$ and $K$ matrices in Definition \ref{defn:fused R matrix and K matrix} admit the following alternative expressions:
\begin{align}
    &\scR^I(u)=(\Pi_1 \otimes\Pi_2)\circ\lb\prod_{a\in[[1,I]]}^{\longleftarrow}\prod_{b\in[[a,a+I-1]]}^{\longrightarrow}\scR_b\lb uq^{b+1-2a}\rb\rb\circ(\Pih_1 \otimes\Pih_2), \label{eq:braided fusion R}\\
    &\sK^I(u)=\Pi\circ\lb\prod^{\longleftarrow}_{a\in[[1,I]]}\sK_1\lb uq^{\frac{I+1}{2}-a}\rb\prod_{b\in[[1,a-1]]}^{\longleftarrow}\scR_{a-b}\lb u^2q^{I+1-a-b}\rb\rb\circ\Pih,\label{eq:braided fusion K}\\
    &\sbK^I(u)=\Pi\circ\lb\prod_{a\in[[1,I]]}^{\longrightarrow}\sbK_I\lb uq^{\frac{I+1}{2}-a}\rb\prod_{b\in[[a+1,I]]}^{\longrightarrow}\scR^{-1}_{I+a-b}\lb u^2q^{I+1-a-b}\rb\rb\circ\Pih.\label{eq:braided fusion bK} 
\end{align} 
\end{theorem}
\begin{remark}\label{rmk:braided}
The alternative forms for   fused $R$ and $K$ matrices provided in \eqref{eq:braided fusion R}, \eqref{eq:braided fusion K} and \eqref{eq:braided fusion bK} above are known as the `braided' forms.  The braid form for   fused $R$ matrix can be found in \cite[Section 5.2]{dA22} in different notations. We are unable to find the braid form for the fused $K$ matrix in the literature. 
An explicit formula was obtained in \cite{Kuan_stochastic} for the fused $R$ matrix using the braided form combined with techniques of Hecke algebras. We expect that  the braid form for   fused $K$ matrix could also yield such explicit formula.
\end{remark}
\begin{proof}
 The proof is essentially a multiple (inductive) use of two basic identities:
$$P_{\omega}\sR_{ab}(u)P_{\omega^{-1}}=\sR_{\omega(a)\omega(b)}(u),\quad
P_{\omega}\sK_a(u)P_{\omega^{-1}}=\sK_{\omega(a)}(u),$$
for any $\omega$ in the symmetric group and indices $a,b$.
We can see the spectral parameters are always unchanged, so they
do not play a role as long as they are in the correct order. We  omit spectral parameters in the proof. 

We begin with the proof of the braided form \eqref{eq:braided fusion oR} for the fused $R$ matrix. Note that the left-hand side of \eqref{eq:braided fusion oR} is $\soRb(u)=P_{I,I}\soR(u)$, where $P_{I,I}\in\End\lb\CoI\o\CoI\rb$ swaps the two factors of $\CoI$, and hence can also be written as  $P_{I,I}=P_{\binom{1,\dots,I,I+1,\dots,2I}{I+1,\dots,2I,1,\dots,I}}\in\End\lb\lb\C^2\rb^{\otimes 2I}\rb$.   
 We compare the formula \eqref{eq:braided fusion oR} for $\soRb(u)=P_{I,I}\soR(u)$ with the formula \eqref{eq:fusion R matrix} for $\soR(u)$, and
we only need to prove:
\be \label{eq:barided R matrix we need to prove}
\prod_{a\in[[1,I]]}^{\longleftarrow}\prod_{b\in[[a,a+I-1]]}^{\longrightarrow}\scR_b=P_{\binom{1,\dots,I,I+1,\dots,2I}{I+1,\dots,2I,1,\dots,I}}\prod_{a\in[[1,I]]}^{\longleftarrow}\prod_{b\in[[1,I]]}^{\longrightarrow}\sR_{b,a+I}.
\ee
In the following we will use $(i,j|k)$ for $i\leq j$ to denote the element $\omega$ in the symmetric group such that $\omega(i)=i+k, \dots, \omega(j)=j+k$, $\omega(j+1)=i,\dots,\omega(j+k)=i+k-1$. 
\begin{claim} 
    For any $i\leq j$ 
    we have
    $$\prod_{a\in[[i,j]]}^{\lra}\scR_a=P_{(i,j|1)}\prod_{a\in[[i,j]]}^{\lra}\sR_{a,j+1}.$$
\end{claim}

We prove this claim by backwards induction on $i$. When $i=j$ it reduces to $\scR_j=P_{j,j+1}\sR_{j,j+1}$. Suppose the claim holds for $i+1\in[[2,j]]$ and we prove that it holds for $i$:
\begin{equation*}
    \begin{split}
        \prod_{a\in[[i,j]]}^{\lra}\scR_a &=\scR_i\prod_{a\in[[i+1,j]]}^{\lra}\scR_a
        =P_{i,i+1}\sR_{i,i+1}P_{(i+1,j|1)}\prod_{a\in[[i+1,j]]}^{\lra}\sR_{a,j+1}\\
        &=P_{i,i+1}P_{(i+1,j|1)}\sR_{i,j+1}\prod_{a\in[[i+1,j]]}^{\lra}\sR_{a,j+1}
        =P_{(i,j|1)}\prod_{a\in[[i,j]]}^{\lra}\sR_{a,j+1}.
    \end{split}
\end{equation*}
Hence the claim holds for $i$ and we conclude its proof. Take $j=i+I-1$ in the claim and we get
\be \label{eq:claim R first} 
\prod_{a\in[[i,i+I-1]]}^{\lra}\scR_a=P_{(i,i+I-1|1)}\prod_{a\in[[i,i+I-1]]}^{\lra}\sR_{a,i+I}.
\ee
\begin{claim}
   For $1\leq i\leq I$ we have
    $$\prod_{a\in[[1,i]]}^{\lla}\prod_{b\in[[a,a+I-1]]}^{\lra}\scR_b=P_{(1,I|i)}\prod_{a\in[[1,i]]}^{\lla}\prod_{b\in[[1,I]]}^{\lra} \sR_{b,a+I}.$$
\end{claim}
We prove this claim by induction on $i$. When $i=1$ it reduces to \eqref{eq:claim R first}. Suppose it is true for $i-1\geq 1$. 
\begin{equation*}
    \begin{split}
        \prod_{a\in[[1,i]]}^{\lla}\prod_{b\in[[a,a+I-1]]}^{\lra}\scR_b &=\prod_{b\in[[i,i+I-1]]}^{\lra}\scR_b\times\prod_{a\in[[1,i-1]]}^{\lla}\prod_{b\in[[a,a+I-1]]}^{\lra}\scR_b\\
        &=P_{(i,i+I-1|1)}\sR_{i,i+I}\dots \sR_{i+I-1,i+I} P_{(1,I|i-1)} \prod_{a\in[[1,i-1]]}^{\lla}\prod_{b\in[[1,I]]}^{\lra} \sR_{b,a+I}\\
        &=P_{(i,i+I-1|1)}P_{(1,I|i-1)}\sR_{1,i+I}\dots \sR_{I,i+I}\prod_{a\in[[1,i-1]]}^{\lla}\prod_{b\in[[1,I]]}^{\lra} \sR_{b,a+I}\\
        &=P_{(1,I|i)}\prod_{a\in[[1,i]]}^{\lla}\prod_{b\in[[1,I]]}^{\lra} \sR_{b,a+I}.
    \end{split}
\end{equation*}
Hence the claim holds for $i$ and we conclude its proof. We take $i=I$ in this claim and get equation \eqref{eq:barided R matrix we need to prove}. We conclude the proof of the braided form \eqref{eq:braided fusion oR} of fused $R$ matrix.

To prove the braided form \eqref{eq:braided fusion oK} for fused $K$ matrix, compare it with  \eqref{eq:fusion K matrix} and we only need to prove:
$$
\prod_{a\in[[1,I]]}^{\lla}\sK_1\scR_1\dots\scR_{a-1}=P_{\binom{1,\dots,I}{I,\dots,1}}\prod_{a\in[[1,I]]}^{\lla}\sK_a\sR_{a-1,a}\dots \sR_{1,a}.
$$
This can be seen as starting from $\sK_1=\sK_1$ and multiplying the following identity through $2\leq a\leq I$:
\be\label{eq:identity in the proof of braided fusion K}
\sK_1\scR_1\dots\scR_{a-1}P_{\binom{1,\dots,a-1}{a-1,\dots,1}}=P_{\binom{1,\dots,a}{a,\dots,1}}\sK_{a}\sR_{a-1,a}\dots \sR_{1,a}.
\ee
We re-write the right hand side of the above identity as $$\sK_1\sR_{21}\dots \sR_{a,1}P_{\binom{1,\dots,a}{a,\dots,1}},$$
so the identity \eqref{eq:identity in the proof of braided fusion K} is equivalent to
$$\scR_1\dots\scR_{a-1}=\sR_{21}\dots \sR_{a,1}P_{(1,a-1|1)}.$$
This can be seen as starting with $\scR_1=\sR_{21}P_{12}$ and multiplying
 the following identity through $2\leq i\leq a-1$:
$$P_{(1,i-1|1)}\scR_i=\sR_{i+1,1}P_{(1,i|1)}.$$
The above identity reduces to $\scR_i=P_{i,i+1}\sR_{i,i+1}$. Hence we conclude the proof of \eqref{eq:braided fusion oK}.

To prove the braided form of fused $\bK$ matrix \eqref{eq:braided fusion obK} we note that
$$\scR_a^{-1}=\lb P_{a,a+1}\sR_{a,a+1}\rb^{-1}=\sR_{a,a+1}^{-1}P_{a,a+1}=P_{a,a+1}\sR_{a+1,a}^{-1}$$
the  proof is then word-by-word parallel with the proof of braided form of fused $K$ matrix \eqref{eq:braided fusion oK} if we replace each index $a$ by $I+1-a$ for $1\leq a\leq I$ and each operator $\sR$ by $\sR^{-1}$. 
\end{proof}

We now begin the proof of Theorem \ref{thm:Fusion of ZF and GZ relations}.
\begin{proof}[Proof of Theorem \ref{thm:Fusion of ZF and GZ relations}]
    We want to show that   $\sbM^I(u):=\left[\sM_0^I(u),\dots\sM_I^I(u)\right]^I\in \C^{I+1}\o\hA$ defined as: 
    \be\label{eq:fused solution ZF and GZ in proof}\sM^I_{\zeta}(u):=\sum_{\substack{\zeta_1+\dots+\zeta_I=\zeta \\ \zeta_1,\dots,\zeta_I\in\left\{0,1\right\}}}\prod_{a\in[[1,I]]}^{\lra}\sM_{\zeta_a}\lb uq^{-\frac{I+1}{2}+a}\rb\in\hA,\quad 0\leq\zeta\leq I,\ee
    satisfies the ZF and GZ relations with the fused matrices $\sR^I(u)\in\End\lb\C^{I+1}\o\C^{I+1}\rb$ and $\sK^I(u), \sbK^I(u)\in\End\lb\C^{I+1}\rb$ defined in Definition \ref{defn:fused R matrix and K matrix}.
    We are able to write $\sbM^I(u)$ in an alternate form:
    \be\label{eq:fused_q_M_vector}
\sbM^I(u)=\Pi\circ\bigotimes_{a\in[[1,I]]}^{\longrightarrow}\sbM\lb uq^{-\frac{I+1}{2}+a}\rb\in \CI\o \hA,
\ee
where $\sbM(u)=[\sM_0(u),\sM_1(u)]^T\in\C^2\o\hA$.  
We remark that in the above equation, the factor $\hA^{\o I}$ get (implicitly) contracted to $\hA$ by the multiplication in the algebra $\hA$, and  then $\Pi$ projects $\CoI$ to $\CI$.

We next show that 
\begin{equation}\label{eq:fused M in q exchangable subspace}
    \bigotimes_{a\in[[1,I]]}^{\longrightarrow}\sbM\lb uq^{-\frac{I+1}{2}+a}\rb\in \sy\o \hA,
\end{equation}
where we recall that $\sy$ is the $q$-exchangeable subspace of $\CoI$. In view of Definition \ref{def:fusion procedure} (3), it suffices to prove that for any $a\in[[1,I-1]]$, we have \[q\sM_0\left( uq^{-\frac{I+1}{2}+a}\right)\sM_1\left(uq^{-\frac{I+1}{2}+a+1}\right)=\sM_1\left( uq^{-\frac{I+1}{2}+a}\right)\sM_0\left(uq^{-\frac{I+1}{2}+a+1}\right).\]
This identity can be verified directly using $\sM_0(u)=u+\e$, $\sM_1(u)=\frac{1}{u}+\d$, and the relation $\d\e-q\e\d=1-q$.

We recall from Proposition \ref{prop:unfused ZF and GZ solution} that $\sbM(u)$ 
satisfies the $\zf$ and $\gz$ relations with the unfused matrices $\sR(u)\in\End\lb\Co\o\Co\rb$ and $\sK(u),\sbK(u)\in\End\lb\Co\rb$.

We make use of the braided forms of the fused $R$ and $K$ matrices given in Theorem \ref{thm:braided version R K}. 
In view of Lemma \ref{lem:q-exchangeble}, the operator $\soR(u)$ has the invariant subspace $\sy\o\sy$, while $\soK(u)$ and $\sobK(u)$ have the invariant subspace $\sy$. Note also that the projector $F=\Pih\circ\Pi$ acts as the identity operator on $\mathrm{Im}(F)=\sy\subset\CoI$. Using \eqref{eq:fused M in q exchangable subspace}, one can therefore get rid of the various $\Pi$ and $\Pih$. Hence, in order to prove the fused ZF relation  
    \be\label{eq:fused_ZF_relation}
    \scR^I\left(\frac{x}{y}\right)\sbM^I(x)\o\sbM^I(y)=\sbM^I(y)\o\sbM^I(x),
    \ee
    we only need to prove 
    \begin{equation*}
        \begin{split}
            \prod_{i\in[[1,I]]}^{\lla}\scR_i\lb\frac{x}{y}q^{1-i}\rb\dots\scR_{i+I-1}\lb\frac{x}{y}q^{I-i}\rb\bigotimes_{a\in[[1,I]]}^{\longrightarrow}\sbM\lb xq^{-\frac{I+1}{2}+a}\rb\o\bigotimes_{b\in[[1,I]]}^{\longrightarrow}\sbM\lb yq^{-\frac{I+1}{2}+b}\rb\\
            =\bigotimes_{b\in[[1,I]]}^{\longrightarrow}\sbM\lb yq^{-\frac{I+1}{2}+b}\rb\o\bigotimes_{a\in[[1,I]]}^{\longrightarrow}\sbM\lb xq^{-\frac{I+1}{2}+a}\rb.
        \end{split}
    \end{equation*}
This can be seen as an induction through the following identity for $1\leq i\leq I$:
\begin{equation*}
        \begin{split}
            \scR_i\lb\frac{x}{y}q^{1-i}\rb\dots\scR_{i+I-1}\lb\frac{x}{y}q^{I-i}\rb\bigotimes_{b\in[[1,i-1]]}^{\longrightarrow}\sbM\lb yq^{-\frac{I+1}{2}+b}\rb\o\bigotimes_{a\in[[1,I]]}^{\longrightarrow}\sbM\lb xq^{-\frac{I+1}{2}+a}\rb\o\bigotimes_{b\in[[i,I]]}^{\longrightarrow}\sbM\lb yq^{-\frac{I+1}{2}+b}\rb\\
            =\bigotimes_{b\in[[1,i]]}^{\longrightarrow}\sbM\lb yq^{-\frac{I+1}{2}+b}\rb\o\bigotimes_{a\in[[1,I]]}^{\longrightarrow}\sbM\lb xq^{-\frac{I+1}{2}+a}\rb\o\bigotimes_{b\in[[i+1,I]]}^{\longrightarrow}\sbM\lb yq^{-\frac{I+1}{2}+b}\rb,
        \end{split}
    \end{equation*}
which is itself a result of using of the unfused ZF relation:
\be\label{eq:unfused ZF relation in the proof}
\scR\left(\frac{x}{y}\right)\sbM(x)\o\sbM(y)=\sbM(y)\o\sbM(x)
\ee
$I$ times, where
each step transfers the $(I+i)$-th factor  $\sbM\lb yq^{-\frac{I+1}{2}+i}\rb$ in the tensor product leftwards by one. We conclude the proof of the identity and hence the proof of fused ZF relation \eqref{eq:fused_ZF_relation}.

To prove the fused GZ relation
    \be\label{eq: fused GZ relation in the proof}
    \lw \sK^I(u)\sbM^I\lb\frac{1}{u}\rb=\ll W|\sbM^I(u),
    \ee
we only need to prove
\begin{equation*}
        \begin{split}
            \lw\prod_{i\in[[1,I]]}^{\lla}\sK_1\lb uq^{\frac{I+1}{2}-i}\rb\scR_1\lb u^2q^{I+2-2i}\rb\dots\scR_{i-1}\lb u^2q^{I-i}\rb\bigotimes_{a\in[[1,I]]}^{\lra}\sbM\lb\frac{1}{u}q^{-\frac{I+1}{2}+a}\rb\\
            =\lw\bigotimes_{a\in[[1,I]]}^{\lra}\sbM\lb uq^{-\frac{I+1}{2}+a}\rb.
        \end{split}
    \end{equation*}
This can be seen as an induction through the following identity for $1\leq i\leq I$:
\begin{equation}\label{eq:identity in the proof of fused GZ relation}
    \begin{split}
        \lw \sK_1\lb uq^{\frac{I+1}{2}-i}\rb\scR_1\lb u^2q^{I+2-2i}\rb\dots\scR_{i-1}\lb u^2q^{I-i}\rb\bigotimes_{a\in[[1,i-1]]}^{\lla}\sbM\lb uq^{\frac{I+1}{2}-a}\rb\o\bigotimes_{a\in[[i,I]]}^{\lra}\sbM\lb\frac{1}{u}q^{-\frac{I+1}{2}+a}\rb\\
        =\lw \bigotimes_{a\in[[1,i]]}^{\lla}\sbM\lb uq^{\frac{I+1}{2}-a}\rb\o\bigotimes_{a\in[[i+1,I]]}^{\lra}\sbM\lb\frac{1}{u}q^{-\frac{I+1}{2}+a}\rb
    \end{split}
\end{equation}
The above equation can be seen as follows. Firstly the $R$ matrices transfer the $i$-th factor $\sbM\lb\frac{1}{u}q^{-\frac{I+1}{2}+i}\rb$ in the tensor product all the way to the left, where each step uses the unfused ZF relation \eqref{eq:unfused ZF relation in the proof}, i.e.
\begin{equation*}
    \begin{split}
        \scR_1\lb u^2q^{I+2-2i}\rb\dots\scR_{i-1}\lb u^2q^{I-i}\rb\bigotimes_{a\in[[1,i-1]]}^{\lla}\sbM\lb uq^{\frac{I+1}{2}-a}\rb\o\bigotimes_{a\in[[i,I]]}^{\lra}\sbM\lb\frac{1}{u}q^{-\frac{I+1}{2}+a}\rb\\
        =\sbM\lb\frac{1}{u}q^{-\frac{I+1}{2}+i}\rb\o\bigotimes_{a\in[[1,i-1]]}^{\lla}\sbM\lb uq^{\frac{I+1}{2}-a}\rb\o  \bigotimes_{a\in[[i+1,I]]}^{\lra}\sbM\lb\frac{1}{u}q^{-\frac{I+1}{2}+a}\rb
    \end{split}
\end{equation*}
Then the $K$ matrix $\sK_1\lb uq^{\frac{I+1}{2}-i}\rb$ acts on the first factor $\sbM\lb\frac{1}{u}q^{-\frac{I+1}{2}+i}\rb$ of the tensor product and converts the spectral parameter to its inverse $\sbM\lb uq^{\frac{I+1}{2}-i}\rb$ by the unfused GZ relation:
\be \label{eq:unfused GZ relation in the proof}
   \lw \sK(u)\sbM\lb\frac{1}{u}\rb=\ll W|\sbM(u).
\ee
i.e. 
\begin{equation*}
    \begin{split}
        \lw \sK_1\lb uq^{\frac{I+1}{2}-i}\rb\sbM\lb\frac{1}{u}q^{-\frac{I+1}{2}+i}\rb\o\bigotimes_{a\in[[1,i-1]]}^{\lla}\sbM\lb uq^{\frac{I+1}{2}-a}\rb\o  \bigotimes_{a\in[[i+1,I]]}^{\lra}\sbM\lb\frac{1}{u}q^{-\frac{I+1}{2}+a}\rb\\
        =\lw \bigotimes_{a\in[[1,i]]}^{\lla}\sbM\lb uq^{\frac{I+1}{2}-a}\rb\o\bigotimes_{a\in[[i+1,I]]}^{\lra}\sbM\lb\frac{1}{u}q^{-\frac{I+1}{2}+a}\rb.
    \end{split}
\end{equation*}
We conclude the proof of the identity \eqref{eq:identity in the proof of fused GZ relation} and hence the proof of the fused GZ relation \eqref{eq: fused GZ relation in the proof}.

We remark that in the above proof we have implicitly used two basic facts:
\begin{enumerate}
    \item [$\bullet$] If $m_1,m_2\in \mathcal{V}\o\hA$ and $\ll W|m_1=\ll W|m_2$ then $\ll W|\Phi m_1=\ll W|\Phi m_2$ for any $\Phi\in\End(\mathcal{V})$.
    \item [$\bullet$] If $m_1,m_2\in \mathcal{V}\o\hA$, $m\in \mathcal{V}'\o\hA$ and $\ll W|m_1=\ll W|m_2$ then $\ll W|m_1\o m=\ll W|m_2\o m$.
\end{enumerate}
The first fact allows us to induct on the identity \eqref{eq:identity in the proof of fused GZ relation} and the second fact allows us to apply the unfused GZ equation \eqref{eq:unfused GZ relation in the proof} only on the first factor of a tensor product.

The proof of the second fused GZ relation
\be\label{eq:second_fused_GZ_relation}
\sbK^I(u)\sbM^I\lb\frac{1}{u}\rb \rv=\sbM^I(u) |V\rr
\ee
follows from an induction that is parallel with the above proof of the first fused GZ relation.  
\end{proof}
\subsection{Proof of Theorem \ref{thm:general matrix ansatz higher spin}}
  \label{subsec:proof of MPA}
In this subsection we offer the deferred proof of Theorem \ref{thm:general matrix ansatz higher spin}. 
This proof is very similar with the proof of \cite[Theorem 2.6]{Y} in the spin-$\frac{1}{2}$ ($I=1$) case. 

Assume  there are elements $M_{j}^\u$ and $M_{j}^\r$, $0\leq j\leq I$ in the $\C$-algebra $\hA$ 
and  boundary vectors $\ll W|\in H^*$ and $|V\rr\in H$ (recall that $H$ is a linear representation space of $\hA$) satisfying consistency relations:  
\be\label{eq:consistency relations in proof}M_c^\u M_d^\r=\sum_{a,b=0}^I\R_{a,b}^{c,d}M_b^\r M_a^\u, \quad
    \ll W|M_d^\r=\sum_{a=0}^I\lb\lK_a^d\rb\ll W|M_a^\u, \quad
    M_c^\u|V\rr=\sum_{b=0}^I\lb\rK_b^c\rb M_b^\r|V\rr. \ee
    We consider the following collection of $\mathbb{C}$-valued measures $\{\mu_\hP\}$ (each measure has total mass $1$)  on the state space $[[0,r]]^N$,
    indexed by down-right paths $\hP$ on the strip,
      given by the matrix product state: 
    \be\label{eq:MPA in proof}
    \mu_{\mathcal{P}}(\tau_1,\dots,\tau_N)=\frac{\ll W|M_{\tau_1}^{p_1}\times\dots\times M_{\tau_N}^{p_N}|V\rr}{\ll W|( \sum_{j=0}^IM^{p_1}_j)\times\dots\times (\sum_{j=0}^IM^{p_N}_j)|V\rr},
    \ee
    where $p_i\in\{\u,\r\}$, $1\leq i\leq N$ are outgoing edges of $\hP$ labelled from the up-left of $\hP$ to the down-right of $\hP$.
    We prove the following stationarity of  collection $\mu_\hP$ under the evolution of spin-$\frac{I}{2}$ vertex model on a strip:
    \begin{claim}\label{claim:in proof of MPA}
    For any down-right paths $\hP$ and $\hQ$ such that $\hQ$ sits above $\hP$,  we have for all  $\tau'\in[[0,I]]^N$,
    \be\label{eq:transition probability}
    \sum_{\tau\in[[0,I]]^N} P_{\hP,\mathcal{Q}}(\tau,\tau')\mu_\hP(\tau)=\mu_{\mathcal{Q}}(\tau').
    \ee
    \end{claim}
    
    Observe that for the translated path $\Upsilon_1\hP=\hP+(1,1)$, the measure $\mu_{\Upsilon_1\hP}$ is the same as $\mu_\hP$ (as $\mathbb{C}$-valued measures on $[[0,I]]^N$), since the outgoing edges of $\Upsilon_1$ are also $p_1,\dots,p_N\in\left\{\u,\r\right\}$ (so that the elements $M_{\tau_i}^{p_i}$ in  \eqref{eq:MPA in proof} are the same).
    Assume that the Claim \ref{claim:in proof of MPA} holds, we can take $\mathcal{Q}=\Upsilon_1\hP$ and hence $\mu_\hP$ is an eigenvector with eigenvalue $1$ of the transition matrix $P_{\hP,\Upsilon_1\hP}(\tau,\tau')$ of the (irreducible) Markov chain indexed by $\hP$. By Perron-Frobenius theorem, $\mu_\hP$ is the unique stationary probability measure of this system. 
    
    We introduce three types of `local moves' of a down-right path, where the thick lines denote locally the down-right path: 
$$
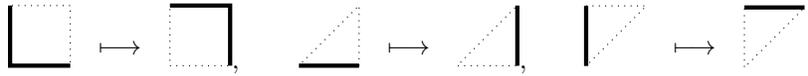
\begin{tikzpicture}[scale=0.8]
		\draw[dotted] (0,0) -- (1,0)--(1,1)--(0,1)--(0,0);
		\draw[ultra thick] (0,1) -- (0,0) -- (1,0);
		\end{tikzpicture}
  \quad 
\raisebox{5pt}{$\longmapsto$}
\quad 
\begin{tikzpicture}[scale=0.8]
		\draw[dotted] (0,0) -- (1,0)--(1,1)--(0,1)--(0,0);
		\draw[ultra thick] (0,1) -- (1,1) -- (1,0);
		\end{tikzpicture}
,\quad\quad
\begin{tikzpicture}[scale=0.8]
		\draw[dotted] (0,0) -- (0,1)--(-1,0)--(0,0);
            \draw[ultra thick] (0,0) -- (-1,0);
		\end{tikzpicture}
  \quad 
\raisebox{5pt}{$\longmapsto$}
\quad 
\begin{tikzpicture}[scale=0.8]
		\draw[dotted] (0,0) -- (0,1)--(-1,0)--(0,0);
		\draw[ultra thick] (0,0) -- (0,1);
		\end{tikzpicture}
,\quad\quad
\begin{tikzpicture}[scale=0.8]
		\draw[dotted] (0,0) -- (1,0)--(0,-1)--(0,0);
            \draw[ultra thick] (0,0) -- (0,-1);
		\end{tikzpicture}
  \quad 
\raisebox{5pt}{$\longmapsto$}
\quad 
\begin{tikzpicture}[scale=0.8]
		\draw[dotted] (0,0) -- (1,0)--(0,-1)--(0,0);
		\draw[ultra thick] (0,0) -- (1,0);
		\end{tikzpicture},
$$
By sequentially performing three local moves, a path $\mathcal{P}$ can be updated to any other path $\mathcal{Q}$ above it.  Hence \eqref{eq:transition probability} can be guaranteed by its special case when $\mathcal{Q}=\widetilde{\hP}$ is a local move of $\hP$:  For all $\tau'\in[[0,I]]^N$,
\be\label{eq:compatibility with local move in proof}
    \sum_{\tau\in[[0,I]]^N} P_{\hP,\widetilde{\hP}}(\tau,\tau')\mu_\hP(\tau)=\mu_{\widetilde{\hP}}(\tau').
\ee

As we put matrix product form \eqref{eq:MPA in proof} of $\mu_\hP$ and $\mu_{\widetilde{\hP}}$ in the above equation, all of the terms coincide except for two that went through the local move in the bulk, or one that went through local move at the left/right boundary. We only need to keep track of the terms that have been updated.
In the following diagrams the thick paths denote locally down-right paths $\hP$ and $\widetilde{\hP}$, and gray arrows denote locally outgoing configurations.
    \begin{enumerate}
        \item [$\bullet$]The bulk local move
    $$
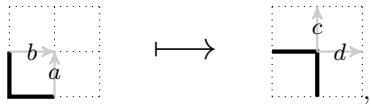
\begin{tikzpicture}[scale=0.6]
		\draw[dotted] (0,0)--(2,0);\draw[dotted] (0,1)--(2,1);\draw[dotted] (0,2)--(2,2);
            \draw[dotted] (0,0)--(0,2);\draw[dotted] (1,0)--(1,2);\draw[dotted] (2,0)--(2,2);
		\draw[ultra thick] (0,1) -- (0,0) -- (1,0);
            \draw[path,lgray] (0,1)--(1,1); \draw[path,lgray] (1,0)--(1,1);
            \node at (0.5,1) {\small $b$};\node at (1,0.5) {\small $a$};
		\end{tikzpicture}
  \quad \quad
\raisebox{15pt}{\scalebox{1.5}{$\longmapsto$}}
\quad \quad
\begin{tikzpicture}[scale=0.6]
		\draw[dotted] (0,0)--(2,0);\draw[dotted] (0,1)--(2,1);\draw[dotted] (0,2)--(2,2);
            \draw[dotted] (0,0)--(0,2);\draw[dotted] (1,0)--(1,2);\draw[dotted] (2,0)--(2,2);
		\draw[ultra thick] (0,1) -- (1,1) -- (1,0);
            \draw[path,lgray] (1,1)--(2,1); \draw[path,lgray] (1,1)--(1,2);
            \node at (1,1.5) {\small $c$};\node at (1.5,1) {\small $d$};
            \end{tikzpicture},
$$
happens with probability $\R_{a,b}^{c,d}$, which 
takes $M_b^\r M_a^\u$ in the matrix ansatz of $\mu_{\hP}$ to $M^{\u}_c M_d^\r$ in the matrix ansatz of $\mu_{\widetilde{\hP}}$. This gives us the first relation in \eqref{eq:consistency relations in proof}: For all $0\leq c,d\leq I$,
$$M_c^\u M_d^\r=\sum_{a,b=0}^I\R_{a,b}^{c,d}M_b^\r M_a^\u.$$
        \item [$\bullet$]The local move at the left boundary
$$
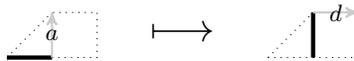
\begin{tikzpicture}[scale=0.6]
		\draw[dotted] (0,0) -- (1,1);\draw[dotted] (0,0) -- (1,0);
          \draw[dotted] (1,0)--(2,0)--(2,1)--(1,1)--(1,0);
            \draw[ultra thick] (0,0) -- (1,0);
            \draw[path,lgray] (1,0)--(1,1);
            \node at (1,0.5) {\small $a$};
		\end{tikzpicture}
  \quad\quad
\raisebox{7pt}{\scalebox{1.5}{$\longmapsto$}}
\quad\quad
\begin{tikzpicture}[scale=0.6]
		\draw[dotted] (0,0) -- (1,1);\draw[dotted] (0,0) -- (1,0);
          \draw[dotted] (1,0)--(2,0)--(2,1)--(1,1)--(1,0);
            \draw[ultra thick] (1,0) -- (1,1);
            \draw[path,lgray] (1,1)--(2,1);
            \node at (1.5,1) {\small $d$};
		\end{tikzpicture}
$$
happens with probability $\lK_a^d$, which
takes $\ll W|M_a^\u$ in the matrix ansatz of $\mu_{\hP}$ to $\ll W|M_d^\r$ in the matrix ansatz of $\mu_{\widetilde{\hP}}$. This gives us the second relation in \eqref{eq:consistency relations in proof}: For all $0\leq d\leq I$,
$$\ll W|M_d^\r=\sum_{a=0}^I\lb\lK_a^d\rb\ll W|M_a^\u.$$ 
    \end{enumerate}

The case of the right boundary local move is very similar. 
We conclude the proof.

\section{Stationary measure in terms of Askey--Wilson process  and phase diagram}\label{sec:3}
In Section \ref{subsec:backgrounds of askey wilson} we introduce the backgrounds of Askey--Wilson measures and processes. It is known that certain types of matrix product states can be written as expectations under the Askey--Wilson process, which will be introduced in Section \ref{subsec:Matrix product states askey wilson}. In Section \ref{subsec:proof of stationary measure as AW} we prove Theorem \ref{thm:stationary measure of fused model as AW}   expressing the joint generating function of the stationary measure of the fused vertex model on a strip in terms of the Askey--Wilson process.
In Section \ref{subsec:mean arrow density and phase diagram} we use this expression
to prove Theorem \ref{thm:mean arrow density} on the asymptotic behavior of the `mean arrow density' of the stationary measure as the system size $N\rightarrow\infty$, whereby recovering the phase diagram.
\subsection{Backgrounds on Askey--Wilson process}\label{subsec:backgrounds of askey wilson}
The Askey--Wilson measures (originally introduced by \cite{askey1985some}) depend on five parameters $(\aa,\bb,\cc,\dd,q)$, where $q\in(-1,1)$ and $\aa, \bb, \cc, \dd$ are either real, or two of them form a complex conjugate pair, or they form two complex conjugate pairs.
In addition we require:
\be\label{eq:condition parameters AW measures}
\aa\bb, \aa\cc, \aa\dd, \bb\cc, \bb\dd, \cc\dd, q\aa\bb,  q\aa\cc, q\aa\dd, q\bb\cc, q\bb\dd, q\cc\dd, \aa\bb\cc\dd, q\aa\bb\cc\dd\in\mathbb{C}\setminus[1,\infty).\ee
The Askey--Wilson measure is a probability measure of mixed type on $\mathbb{R}$:
\be\label{eq:AW measure}\nu(dy;\aa, \bb, \cc, \dd, q)=f(y, \aa, \bb, \cc, \dd, q)dy+\sum_{z\in F(\aa,\bb,\cc,\dd,q)}p(z)\delta_z(dy).\ee
The absolutely continuous part of \eqref{eq:AW measure} is supported on $[-1,1]$ with density 
\begin{equation}\label{eq:defn of continuous part density}
    f(y,\aa,\bb,\cc,\dd,q)=\frac{(q,\aa\bb,\aa\cc,\aa\dd,\bb\cc,\bb\dd,\cc\dd;q)_{\infty}}{2\pi(\aa\bb\cc\dd;q)_{\infty}\sqrt{1-y^2}}\left\vert\frac{\lb e^{2i\th_y};q\rb_{\infty}}{\lb \aa e^{i\th_y},\bb e^{i\th_y},\cc e^{i\th_y},\dd e^{i\th_y};q\rb_{\infty}}\right\vert^2,
\end{equation}
where $y=\cos\th_y$ and we set $f(y,\aa,\bb,\cc,\dd,q)=0$ for $|y|\geq 1$. 
Here and below, for complex $z$ and $n\in\mathbb{Z}_+\cup\{\infty\}$, we use the $q$-Pochhammer symbol:
$$
 (z;q)_n=\prod_{j=0}^{n-1}\,\lb 1-z q^j\rb, \quad (z_1,\cdots,z_k;q)_n=\prod_{i=1}^k(z_i;q)_n.
$$
The atomic part of \eqref{eq:AW measure} is supported on a finite or empty set $F(\aa,\bb,\cc,\dd,q)$ of atoms generated by numbers $\chi\in\left\{\aa,\bb,\cc,\dd\right\}$ such that $|\chi|>1$. In this case $\chi$ must be real and generates its own set of atoms:
\begin{equation}\label{eq:defn of atoms}
    y_j=\frac{1}{2}\left(\chi q^j+\frac{1}{\chi q^j}\right) \text{ for }j=0,1\dots\text{ such that } |\chi q^j| > 1.
\end{equation}
When $\chi=\aa$, the corresponding masses are
\begin{align}
\label{eq:atom masses 0}    p(y_0;\aa,\bb,\cc,\dd,q)&=\frac{\qp{\aa^{-2},\bb\cc,\bb\dd,\cc\dd}}{\qp{\bb/\aa,\cc/\aa,\dd/\aa,\aa\bb\cc\dd}},\\
 \label{eq:atom masses all}   p(y_j;\aa,\bb,\cc,\dd,q)&=p(y_0;\aa,\bb,\cc,\dd,q)\frac{\qps{\aa^2,\aa\bb,\aa\cc,\aa\dd}_j\,\lb1-\aa^2q^{2j}\rb}{\qps{q,q\aa/\bb,q\aa/\cc,q\aa/\dd}_j\lb1-\aa^2\rb}\left(\frac{q}{\aa\bb\cc\dd}\right)^j,\quad j\geq 1.
\end{align}
For $\chi\in\{\bb,\cc,\dd\}$ the corresponding masses of the atoms are given by similar formulas with $\aa$ and $\chi$ swapped.

\begin{remark}
A more general condition on the parameters $(\aa,\bb,\cc,\dd,q)$ under which the Askey--Wilson measure defined above is a probability measure (i.e., with the continuous part density \eqref{eq:defn of continuous part density} and all atom masses \eqref{eq:atom masses 0} and \eqref{eq:atom masses all} being finite and nonnegative, and with total mass one) is given in \cite[Lemma 3.1]{BW10}; see also \cite[Section 1.4]{BW17} for a later reference. The conditions on $(\aa,\bb,\cc,\dd,q)$ assumed above correspond to a special case of \cite[Lemma 3.1]{BW10}, namely when $m_1 = m_2 = 0$. 
\end{remark}

The Askey--Wilson processes introduced by \cite{BW10,BW17}
depend on five parameters $(\AA,\BB,\CC,\DD,q)$, where $q\in(-1,1)$ and $\AA,\BB,\CC,\DD$ are either real or $(\AA,\BB)$ or $(\CC,\DD)$ form complex conjugate pairs, and in addition
$$\AA\CC,\AA\DD,\BB\CC,\BB\DD,q\AA\CC,q\AA\DD,q\BB\CC,q\BB\DD,\AA\BB\CC\DD,q\AA\BB\CC\DD\in\mathbb{C}\setminus[1,\infty).$$
The Askey--Wilson process $\lb Y_t\rb_{t\in I}$ is a time-inhomogeneous Markov process defined on the interval
$
  I= \left( \max\{0,\CC\DD,q\CC\DD\},\,\tfrac{1}{\max\{0,\AA\BB,q\AA\BB\}}\right),
$ with marginal distributions
$$
  \pi_t(dx)=\nu\lb dx;\AA\sqrt{t},\BB\sqrt{t},\CC/\sqrt{t},\DD/\sqrt{t},q\rb,\quad \forall t\in I,
$$ and transition probabilities
$$
   \mathsf{P}_{s,t}(x,dy)=\nu\lb dy;\AA\sqrt{t},\BB\sqrt{t},\sqrt{s/t}\lb x+\sqrt{x^2-1}\rb,\sqrt{s/t}\lb x-\sqrt{x^2-1}\rb\rb,\quad \forall s<t,\quad s,t\in I.
$$
The marginal distribution $\pi_t(dx)$ may have atoms at 
\begin{equation}\label{eq:defn of all possible atoms}
    \frac{1}{2}\left(\AA\sqrt{t} q^j+\frac{1}{\AA\sqrt{t} q^j}\right),\quad \frac{1}{2}\left(\BB\sqrt{t} q^j+\frac{1}{\BB\sqrt{t} q^j}\right),\quad\frac{1}{2}\left(\frac{\CC q^j}{\sqrt{t}}+\frac{\sqrt{t}}{\CC q^j}\right),\quad \frac{1}{2}\left(\frac{\DD q^j}{\sqrt{t}}+\frac{\sqrt{t}}{\DD q^j}\right),
\end{equation}
and the transition probabilities $\mathsf{P}_{s,t}(x,dy)$ may also have atoms.

\subsection{Matrix product states in terms of Askey--Wilson processes}\label{subsec:Matrix product states askey wilson}
We introduce the result in \cite{BW17} which expresses a certain type of matrix product states in terms of expectations of  Askey--Wilson process.

Suppose that there are real numbers $q,\alpha,\beta,\gamma,\delta$ satisfying: 
\be   \label{eq:conditions open ASEP}
\alpha,\beta>0,\quad \gamma,\de\geq 0,\quad 0\leq q<1.
\ee   

We consider the following relations involving matrices $\D$ and $\E$, row vector $\ll W|$ and column vector $|V\rr$:  
\be\label{eq:DEHP}
 \D\E-q\E\D=\D+\E, \quad
        \ll W|(\alpha\E-\gamma\D)=\ll W|, \quad
        (\beta\D-\de\E)|V\rr=|V\rr,
\ee  
We require that these matrices and vectors have the same (possibly infinite) dimension.
  \eqref{eq:DEHP} is commonly referred to as the DEHP algebra and was introduced in the seminal work \cite{DEHP93}.
\begin{remark}
    In the context of open ASEP, parameters $q,\alpha,\beta,\gamma,\delta$
play the role of  particle jump rates, and the DEHP algebra \eqref{eq:DEHP} plays the role of consistency relations of the matrix product ansatz. 
\end{remark}

\begin{definition}\label{defn:parameterization open ASEP}
    We will use the following (alternative) parameterization:
\be\label{eq:defining ABCD}
A=\k_+(\beta,\de),\quad B=\k_-(\beta,\de),\quad C=\k_+(\alpha,\gamma),\quad D=\k_-(\alpha,\gamma),
\ee
where 
$$
\k_{\pm}(u,v)=\frac{1}{2u}\lb1-q-u+v\pm\sqrt{(1-q-u+v)^2+4uv}\rb.
$$
One can check that for any given $q\in[0,1)$, \eqref{eq:defining ABCD} gives a bijection 
$$\left\{(\alpha,\beta,\gamma,\delta):\alpha,\beta>0, \gamma,\delta\geq 0\right\}\overset{\sim}{\longleftrightarrow}\left\{(A,B,C,D): A,C \geq 0, B,D\in(-1,0]\right\}.$$
\end{definition}

The following result gives a concrete example of $\D$, $\E$, $\ll W|$ and $|V\rr$ satisfying the DEHP algebra \eqref{eq:DEHP}, which is commonly referred to as the USW representation and was first introduced in \cite{USW04}.
\begin{proposition}[USW representation of the $\dehp$ algebra, see \cite{USW04,BW10}]\label{prop: USW representation}
Suppose that the parameters $q,\alpha,\beta,\gamma,\delta$ satisfy \eqref{eq:conditions open ASEP}.
Suppose that $\alpha_n,\beta_n,\gamma_n,\delta_n,\ep_n,\varphi_n$ are given in terms of $q,A,B,C,D$ by the formulas in \cite[page 1243]{BW10}, for $n\in\mathbb{N}_0$. Consider the following tridiagonal matrices:
\begin{align*}
    &\xx=\left[\begin{matrix}
        \gamma_0 & \ep_1 & 0 &\dots \\
        \alpha_0 &  \gamma_1& \ep_2 &\dots \\
        0 & \alpha_1 & \gamma_2& \dots\\
        \vdots &\vdots & \vdots& \ddots\\
      \end{matrix}\right], \quad \yy=\left[\begin{matrix}
        \delta_0 & \varphi_1 & 0 &\dots \\
        \beta_0 & \delta_1 & \varphi_2 & \dots\\
        0 & \beta_1 & \delta_2 & \dots\\
        \vdots & \vdots& \vdots& \ddots\\
      \end{matrix}\right],\\
      &\E=\frac{1}{1-q}\mathbf{I}+\frac{1}{\sqrt{1-q}}\yy,\quad \D=\frac{1}{1-q}\mathbf{I}+\frac{1}{\sqrt{1-q}}\xx,
\end{align*}
and boundary vectors 
$$\ll W|=(1,0,0,\dots),\quad |V\rr=(1,0,0,\dots)^T.$$ 
Then $\D$, $\E$, $\ll W|$ and $|V\rr$ satisfy the $\dehp$ algebra \eqref{eq:DEHP}.
\end{proposition}

\begin{theorem}[Theorem 1 in \cite{BW17}]\label{thm:matrix product state as Askey Wilson}
Assume that $q,\alpha,\beta,\gamma,\delta$ satisfy \eqref{eq:conditions open ASEP} and that $AC<1$ (where $A,B,C,D$ are defined in Definition \ref{defn:parameterization open ASEP}).  
Suppose that $\D$, $\E$, $\ll W|$, $|V\rr$ form the $\usw$ representation  of the $\dehp$ algebra given in Proposition \ref{prop: USW representation}. 
    Then for any $0< t_1\leq\dots\leq t_N$, we have
    $$\lw\prod_{i=1}^{\substack{N\\\longrightarrow}}(\E+t_i\D)\rv=\frac{1}{(1-q)^N}\mathbb{E}\lbe\prod_{i=1}^N\lb 1+t_i+2\sqrt{t_i}Y_{t_i}\rb\rbe,$$
    where $\lb Y_t\rb_{ t>0 }$ is the Askey–Wilson process with parameters $(A, B, C, D, q)$.
\end{theorem}

We will utilize the following corollary of the above theorem, which itself has not appeared in the literature. 
\begin{corollary}\label{cor:matrix product state as Askey Wilson}
Our assumptions are the same as in Theorem \ref{thm:matrix product state as Askey Wilson}. Then for any $v_1,\dots,v_N\in\mathbb{R}$ and $0< t_1\leq\dots\leq t_N$, we have
$$
\lw\prod_{i=1}^{\substack{N\\\longrightarrow}}\lb(1-q)\lb\E+t_i\D\rb+v_i\rb\rv=\mathbb{E}\lbe\prod_{i=1}^N\lb1+t_i+2\sqrt{t_i}Y_{t_i}+v_i\rb\rbe.
$$
\end{corollary}
\begin{proof}
    This follows from expanding the bracket and using Theorem \ref{thm:matrix product state as Askey Wilson} multiple times.
\end{proof}

\subsection{Stationary measure in terms of the Askey--Wilson process}
\label{subsec:proof of stationary measure as AW}
In this subsection we prove Theorem \ref{thm:stationary measure of fused model as AW} expressing the stationary measure of the fused vertex model on a strip in terms of Askey--Wilson process.

We begin by recalling the matrix product ansatz of stationary measure given in Theorem \ref{thm:concrete matrix ansatz}. 
Suppose $\hA$ is a $\C$-algebra with linear representation space $H$. 
Assume $\d,\e\in\hA$, $\ll W|\in H^*$ and $|V\rr\in H$ satisfy:
\be\label{eq:d and e again}
\d\e-q\e\d=1-q, \quad
    \ll W|\lb \aaa\e-\ccc\d+1\rb=0 ,\quad
    \lb \bbb\d-\ddd\e+1\rb|V\rr=0. 
\ee
We then define for $0\leq\zeta\leq I$,
\be  
\sM^I_{\zeta}(u):=\sum_{\substack{\zeta_1+\dots+\zeta_I=\zeta \\ \zeta_1,\dots,\zeta_I\in\left\{0,1\right\}}}\prod_{a\in[[1,I]]}^{\lra}\sM_{\zeta_a}\lb uq^{-\frac{I+1}{2}+a}\rb\in\hA,
\ee
where $\sM_0(u)=u+\e$ and $\sM_1(u)=\frac{1}{u}+\d$. Define $M_j^\u=\sM_j^I\lb1/\k\rb$ and $M_j^\r=\sM_j^I\lb\k\rb$  for $0\leq j\leq I$.
Then the stationary measure of the fused vertex model on a strip on any down-right path $\hP$ is given by: 
\be\label{eq:MPA in proof of AW}\mu_{\mathcal{P}}(\tau_1,\dots,\tau_N)=\frac{\ll W|M_{\tau_1}^{p_1}\times\dots\times M_{\tau_N}^{p_N}|V\rr}{\ll W|( \sum_{j=0}^IM^{p_1}_j)\times\dots\times (\sum_{j=0}^IM^{p_N}_j)|V\rr},\ee
where  $0\leq\tau_1,\dots,\tau_N\leq I$ are occupation variables on outgoing edges $p_1,\dots p_N\in\left\{\u,\r\right\}$ of $\hP$.

We first observe that the relations \eqref{eq:d and e again} is a linear transformation of the DEHP algebra.
\begin{proposition}\label{prop:relations of d e compared with dehp}
Relations \eqref{eq:d and e again} between $\d$, $\e$, $\ll W|$ and $|V\rr$ are equivalent to the $\dehp$ algebra \eqref{eq:DEHP} between
$\D=\frac{1}{1-q}(1+\d)$, $\E=\frac{1}{1-q}(1+\e)$, $\ll W|$ and $|V\rr$, 
    with parameters: 
    \be   \label{eq:def of alpha beta gamma delta}
    (\alpha,\beta,\gamma,\delta)=\lb\frac{(1-q)\aaa}{\aaa-\ccc-1}, \frac{(1-q)\bbb}{\bbb-\ddd-1}, \frac{(1-q)\ccc}{\aaa-\ccc-1}, \frac{(1-q)\ddd}{\bbb-\ddd-1}\rb.\ee
\end{proposition}
\begin{proof}
    This can be seen by taking $\d=(1-q)\D-1$ and $\e=(1-q)\E-1$ into  \eqref{eq:d and e again}.
\end{proof}
\begin{definition}\label{def:alternate parametrization of fused model}
We recall  the fused vertex model on a strip defined in Definition \ref{def: fused model intro}, which has parameters $q,\k,\aaa,\bbb,\ccc,\ddd$ satisfying:
$$0<q<1,\quad0<\k<q^{\frac{I-1}{2}},\quad \aaa,\bbb,\ccc,\ddd>0,\quad \aaa-\ccc>q^{\frac{1-I}{2}}/\kappa,\quad\bbb-\ddd>q^{\frac{1-I}{2}}/\kappa.$$
We will use the alternative parameterization $q,\k,A,B,C,D$ given by \eqref{eq:def of alpha beta gamma delta} and then by \eqref{eq:defining ABCD}: 
$$(\aaa,\bbb,\ccc,\ddd)\overset{\eqref{eq:def of alpha beta gamma delta}}{\longmapsto}(\alpha,\beta,\gamma,\delta)\overset{\eqref{eq:defining ABCD}}{\longmapsto}(A,B,C,D).$$
For any fixed $0<q<1$ and $0<\k<q^{\frac{I-1}{2}}$, this is a bijection from 
$\{(\aaa,\bbb,\ccc,\ddd): \aaa,\bbb,\ccc,\ddd>0, \aaa-\ccc>q^{\frac{1-I}{2}}/\kappa, \bbb-\ddd>q^{\frac{1-I}{2}}/\kappa\}$
to a certain sub-region of $\left\{(A,B,C,D): A,C \geq 0, B,D\in(-1,0]\right\}$.
\end{definition}

\begin{proof}[Proof of Theorem \ref{thm:stationary measure of fused model as AW}]
    By \eqref{eq:MPA in proof of AW}, we have that for any $t_1,\dots,t_N>0$,
    \be\label{eq:generating function in proof}\mathbb{E}_{\mu_{\hP}}\lbe\prod_{i=1}^Nt_i^{\tau_i}\rbe=\sum_{0\leq\tau_1,\dots,\tau_N\leq I}\mu_{\hP}(\tau_1,\dots,\tau_N)\prod_{i=1}^Nt_i^{\tau_i}
=\frac{\Biggl< W\Bigg{|}\prod\limits_{i\in[[1,N]]}^{\lra}
\lb\sum\limits_{\zeta=0}^I t_i^{\zeta}M_\zeta^{p_i}\rb
\Bigg{|}V\Biggr>}{\Biggl< W\Bigg{|}\prod\limits_{i\in[[1,N]]}^{\lra}
\lb\sum\limits_{\zeta=0}^I M_\zeta^{p_i}\rb
\Bigg{|}V\Biggr>}.\ee
For any $t>0$ we have:
\begin{equation*}
        \begin{split}
            \sum_{\zeta=0}^It^\zeta M_\zeta^\r&=\sum_{\zeta=0}^It^\zeta \sM_\zeta^I\lb\k \rb=\sum_{\zeta=0}^It^\zeta\sum_{\substack{\zeta_1+\dots+\zeta_I=\zeta \\ \zeta_i\in\left\{0,1\right\}}}\prod_{a\in[[1,I]]}^{\lra}\sM_{\zeta_a}\lb 
            q^{-\frac{I+1}{2}+a}\k\rb\\
            &=\sum_{\zeta_1,\dots,\zeta_I\in\left\{0,1\right\}}\prod_{a\in[[1,I]]}^{\lra}\lb t^{\zeta_a}\sM_{\zeta_a}\lb 
            q^{-\frac{I+1}{2}+a}\k \rb\rb\\
            &=\prod_{a\in[[1,I]]}^{\lra}\lb \sM_0\lb q^{-\frac{I+1}{2}+a}\k\rb+t\sM_1\lb q^{-\frac{I+1}{2}+a}\k\rb\rb\\
            &=\prod_{a\in[[1,I]]}^{\lra}\lb t\d+\e+tq^{\frac{I+1}{2}-a}\k^{-1}+
            q^{-\frac{I+1}{2}+a}\k \rb\\
            &=\prod_{a\in[[1,I]]}^{\lra}\lb(t\D+\E)(1-q)+tq^{\frac{I+1}{2}-a}\k^{-1}+
            q^{-\frac{I+1}{2}+a}\k-t-1\rb.
        \end{split}
\end{equation*} 
The formula for $\sum_{\zeta=0}^It^\zeta M_\zeta^\u$ is parallel with above, with the only difference that $\k$ on the RHS is replaced by $1/\k$. Therefore we have:
$$
\sum_{\zeta=0}^It^\zeta M_\zeta^{p_i}=\prod_{a\in[[1,I]]}^{\lra}\lb(t\D+\E)(1-q)+tq^{\frac{I+1}{2}-a}\k^{\v_i}+
            q^{-\frac{I+1}{2}+a}\k^{-\v_i}-t-1\rb,
$$
where we recall that $\v_i=1$ if $p_i=\u$ and $\v_i=-1$ if $p_i=\r$.

For any $0<t_1\leq\dots\leq t_N$, by Corollary \ref{cor:matrix product state as Askey Wilson}, we have:
    \be\label{eq:matrix product state written in Askey Wilson}
    \begin{split}
        \Biggl< W\Bigg{|}\prod_{i\in[[1,N]]}^{\lra}\lb\sum_{\zeta=0}^It_i^{\zeta }M_{\zeta}^{p_i}\rb\Bigg{|}V\Biggr>
        &=\Biggl< W\Bigg{|}\prod_{i\in[[1,N]]}^{\lra}\prod_{a\in[[1,I]]}^{\lra}\lb(t_i\D+\E)(1-q)+t_iq^{\frac{I+1}{2}-a}\k^{\v_i}+q^{-\frac{I+1}{2}+a}\k^{-\v_i}
        -t_i-1\rb\Bigg{|}V\Biggr>\\
        &=\mathbb{E}\lbe\prod_{i=1}^N\prod_{a=1}^I\lb 2\sqrt{t_i}Y_{t_i}+t_iq^{\frac{I+1}{2}-a}\k^{\v_i}+q^{-\frac{I+1}{2}+a}\k^{-\v_i} \rb\rbe.
    \end{split}
    \ee
 Taking $t_1=\dots=t_N=1$ in the above equation, we get:
 \be\label{eq:matrix product state normalization}
 \begin{split}
 \Biggl< W\Bigg{|}\prod_{i\in[[1,N]]}^{\lra}\lb\sum_{\zeta=0}^I M_{\zeta}^{p_i}\rb\Bigg{|}V\Biggr>&=
 \mathbb{E}\lbe\prod_{i=1}^N\prod_{a=1}^I\lb 2Y_{1}+ q^{\frac{I+1}{2}-a}\k^{\v_i}+q^{-\frac{I+1}{2}+a}\k^{-\v_i} \rb\rbe\\
 &=\mathbb{E}\lbe\prod_{a=1}^I\lb2Y_1+q^{\frac{I+1}{2}-a}\k+q^{-\frac{I+1}{2}+a}\k^{-1}\rb^N\rbe,
 \end{split}
 \ee
where the last equality above follows from 
$$\prod_{a=1}^I\lb2Y_1+q^{\frac{I+1}{2}-a}\k+q^{-\frac{I+1}{2}+a}\k^{-1}\rb
=\prod_{a=1}^I\lb2Y_1+q^{\frac{I+1}{2}-a}\k^{-1}+q^{-\frac{I+1}{2}+a}\k\rb,$$
which follows from changing the index $a$ to $I+1-a$.

 Theorem \ref{thm:stationary measure of fused model as AW} follows from putting \eqref{eq:matrix product state written in Askey Wilson} and \eqref{eq:matrix product state normalization} into \eqref{eq:generating function in proof}. We conclude the proof.
\end{proof}

\subsection{Mean arrow density and phase diagram: Proof of Theorem \ref{thm:mean arrow density}}\label{subsec:mean arrow density and phase diagram}
We prove Theorem \ref{thm:mean arrow density} on the limits of the mean arrow density of the stationary measure of the fused vertex model on a strip. 

We recall from the statement of the theorem that $\hP^N$ is a down-right path on the strip with width $N$, with $0\leq\phi_N\leq N$ many horizontal edges, and that we are assuming $\phi_N/N$ converges to $\lambda\in[0,1]$ as $N\rightarrow\infty$.
The next result offers an expression of the mean density in terms of the `partition function':
\begin{lemma}\label{lem:mean as partition}
   For each $N\in\z_+$, we define the following function on $t>0$:
    \be   \label{eq:partition function in terms of askey wilson}
    Z_N(t)=\mathbb{E}\lbe\prod_{a=1}^{I}\lb2\sqrt{t}Y_t+tq^{\frac{I+1}{2}-a}\k+
    q^{-\frac{I+1}{2}+a}\k^{-1} \rb^{\phi_N}
    \prod_{a=1}^{I}\lb2\sqrt{t}Y_t+tq^{\frac{I+1}{2}-a}\k^{-1}+
    q^{-\frac{I+1}{2}+a}\k \rb^{N-\phi_N}
    \rbe.
    \ee
Then we have:
\be\label{eq:mean density as partition}\mathbb{E}_{\mu_{\hP^N}}\lbe\frac{1}{N}\sum_{i=1}^N\tau_i\rbe=\frac{\pt Z_N(t)\po}{NZ_N(1)}.\ee
\end{lemma} 
\begin{proof}
    Taking $t_1=\dots=t_N=t$ in \eqref{eq:matrix product state written in Askey Wilson}, we have:
    $$Z_N(t)=\Biggl< W\Bigg{|}\prod_{i\in[[1,N]]}^{\lra}\lb\sum_{\zeta=0}^It^{\zeta }M_{\zeta}^{p_i}\rb\Bigg{|}V\Biggr>
    =\sum_{0\leq\tau_1,\dots,\tau_N\leq I}t^{\sum_{i=1}^N\tau_i}\ll W|M_{\tau_1}^{p_1}\dots M_{\tau_N}^{p_N}|V\rr$$
    Therefore we have:
    $$\pt Z_N(t)\po=\sum_{0\leq\tau_1,\dots,\tau_N\leq I}\lb\sum_{i=1}^N\tau_i\rb\ll W|M_{\tau_1}^{p_1}\dots M_{\tau_N}^{p_N}|V\rr, 
    \quad
     Z_N(1)=\Biggl< W\Bigg{|} \lb\sum_{\zeta=0}^I M_{\zeta}^{p_1}\rb\dots \lb\sum_{\zeta=0}^I M_{\zeta}^{p_N}\rb\Bigg{|}V\Biggr>$$
     From the matrix product ansatz expression \eqref{eq:MPA in proof of AW} of the stationary measure $\mu_{\hP^N}$ we conclude \eqref{eq:mean density as partition}.
\end{proof}

We now begin the proof of Theorem \ref{thm:mean arrow density}.
\begin{proof}[Proof of Theorem \ref{thm:mean arrow density}]
By Lemma \ref{lem:mean as partition}, to obtain the limits of the mean arrow density, we need to analyze the $N\rightarrow\infty$ asymptotic behaviors of $Z_N(1)$ and $\pt Z_N(t)\po$.

We first define some functions that will be useful in the proof.
   For $1\leq a\leq I$, we define:
   \be \label{eq:specific form of h}
   h_a^\u(t,y)=2\sqrt{t}y+tq^{\frac{I+1}{2}-a}\k+q^{-\frac{I+1}{2}+a}\k^{-1},\quad
   h_a^\r(t,y)=2\sqrt{t}y+tq^{\frac{I+1}{2}-a}\k^{-1}+q^{-\frac{I+1}{2}+a}\k,
   \ee 
and then define $h^\u(t,y)=\prod_{a=1}^Ih_a^\u(t,y)$ and $h^\r(t,y)=\prod_{a=1}^Ih_a^\r(t,y)$.

Denote $\psi_N:=N-\phi_N$.
By \eqref{eq:partition function in terms of askey wilson} we can write:
\be\label{eq:partition function in proof}
    \begin{split}
        Z_N(t)=&\mathbb{E}\lbe h^\u(t,Y_t)^{\phi_N}h^\r(t,Y_t)^{\psi_N}\rbe=\int_{-\infty}^\infty h^\u(t,y)^{\phi_N}h^\r(t,y)^{\psi_N}\nu\lb dy;A\sqrt{t},B\sqrt{t},C/\sqrt{t},D/\sqrt{t},q\rb\\
        =&\int_{-1}^1h^\u(t,y)^{\phi_N}h^\r(t,y)^{\psi_N}f\lb y,A\sqrt{t},B\sqrt{t},C/\sqrt{t},D/\sqrt{t},q\rb dy\\
        &+\sum_{y_j(t)\in F\lb A\sqrt{t},B\sqrt{t},C/\sqrt{t},D/\sqrt{t},q\rb} h^\u(t,y_j(t))^{\phi_N}h^\r(t,y_j(t))^{\psi_N} 
  p\lb y_j(t);A\sqrt{t},B\sqrt{t},C/\sqrt{t},D/\sqrt{t},q\rb,
    \end{split}
\ee
where   $F(\ttA,\ttB,\ttC,\ttD,q)$ is the set of atoms generated by $\ttA,\ttB,\ttC,\ttD$ and
\begin{equation}\label{eq:continuous density}
    f\lb y,\ttA,\ttB,\frac{C}{\sqrt{t}},\frac{D}{\sqrt{t}},q\rb=\frac{\lb q,tAB,AC,AD,BC,BD,CD/t;q\rb_{\infty}}{2\pi(ABCD;q)_{\infty}\sqrt{1-y^2}}\left\vert\frac{\lb e^{2i\th_y};q\rb_{\infty}}{\lb\ttA e^{i\th_y},\ttB e^{i\th_y},\frac{C}{\sqrt{t}} e^{i\th_y},\frac{D}{\sqrt{t}} e^{i\th_y};q\rb_{\infty}}\right\vert^2
\end{equation}
is the continuous part density, where $y=\cos\th_y\in[-1,1]$.

We make some observations of the Askey--Wilson measure $\nu\lb dy;A\sqrt{t},B\sqrt{t},C/\sqrt{t},D/\sqrt{t},q\rb$ for $t$ close to $1$. Recall that $A,C \geq 0$, $-1<B,D \leq0$ and  $AC<1$. Hence the atoms are generated by $A\sqrt{t}$ in the high density phase and by $C/\sqrt{t}$ in the low density phase. Since, by our assumptions, $A,C\notin\{q^{-l}:l\in\mathbb{N}_0\}$, for $t$ in small neighborhood of $1$, the number of atoms is constant, and the positions of atoms $y_j(t)$ and the masses $p\lb y_j(t);A\sqrt{t},B\sqrt{t},C/\sqrt{t},D/\sqrt{t},q\rb$ are smooth functions of $t$. The Askey--Wilson measure is supported in $[-1,\infty)$, where $h^\u(t,y)$ and $h^\r(t,y)$ are strictly positive and strictly increasing functions of $y$. In the rest of the proof we denote the continuous part density \eqref{eq:continuous density} by $g(t,y)$. By Fact 3.13 in the proof of \cite[Theorem 3.11]{Y}, there exists a smooth function $\th(t,z)$ defined on a small neighborhood of $\{1\}\times\{|z|=1\}\subset\mathbb{R}\times\C$ which cannot take $0$ as its value, such that:
\be \label{eq:gty}
g(t,y):=f\lb y,\ttA,\ttB,C/\sqrt{t},D/\sqrt{t},q\rb=\sqrt{1-y^2}\th(t,z),
\ee
where $z=e^{i\th_y}$ for $y=\cos\theta_y\in[-1,1]$. As a corollary, there exists a small $\ep>0$, such that functions $g(t,y)$, $\pt g(t,y)$ and $\pt g(t,y)/g(t,y)$ are bounded on the region $(t,y)\in(1-\ep,1+\ep)\times[-1,1]$. In particular, one can take differentiation under the integral sign in $\pt\lb\int_{-1}^1h^\u(t,y)^{\phi_N}h^\r(t,y)^{\psi_N}g(t,y)dy\rb\Big{\vert}_{t=1}$.

By the observation above, we are now able to obtain the limits of mean density:

\begin{enumerate}
         \item [(\upperRomannumeral{1})] \label{1} (High density phase $A>1$).
When $t$ is close to $1$, atoms are generated by $\ttA$.
From the discussions in the paragraph around \eqref{eq:gty}, we can observe that as $N\rightarrow\infty$, both $Z_N(1)$ and $\pt Z_N(t)\po$ are dominated by the largest atom  $y_0(t)=\frac{1}{2}\lb A\sqrt{t}+\frac{1}{A\sqrt{t}}\rb$: 
\begin{align}
Z_N(1)&\sim h^\u(1,y_0(1))^{\phi_N}h^\r(1,y_0(1))^{\psi_N}p(y_0(1);A,B,C,D,q),
\label{eq:partition function high density}\\
\pt Z_N(t)\po&\sim \pt\lb h^\u(t,y_0(t))^{\phi_N}h^\r(t,y_0(t))^{\psi_N}\rb\po
p(y_0(1); A,B,C,D,q),\label{eq:partition function derivative high density} 
\end{align} 
where we use $u(N)\sim v(N)$ to denote $\lim_{N\rightarrow\infty}u(N)/v(N)=1$.  
 We will provide the proofs of \eqref{eq:partition function high density} and \eqref{eq:partition function derivative high density} at the end of the proof of the high density phase, and will next use them to derive the limits of the mean density. Note that, by \eqref{eq:atom masses 0}, we have
$$
p(y_0(1);A,B,C,D,q)=\frac{\qp{A^{-2},BC,BD,CD}}{\qp{B/A,C/A,D/A,ABCD}}>0
$$ 
since $\qp{z}>0$ when $z\in(-\infty,1)$ (since $q\in[0,1)$), and all the entries in the $q$-Pochhammer symbols above lie in $(-\infty,1)$ because $A>1$, $C\geq0$, $AC<1$, and $B,D\in(-1,0]$.
Hence, by \eqref{eq:partition function high density} and \eqref{eq:partition function derivative high density}, we obtain: 
\be \label{eq:in proof of mean density}
    \begin{split}
        \lim_{N\rightarrow\infty}\frac{\pt Z_N(t)\po}{NZ_N(1)} 
&=\lambda \frac{\pt h^\u(t,y_0(t))\po}{h^\u(1,y_0(1))}
+(1-\lambda) \frac{\pt h^\r(t,y_0(t))\po}{h^\r(1,y_0(1))},
    \end{split}
\ee 
where we have used $\lim_{N\rightarrow\infty}\phi_N/N=\lambda$ and $\lim_{N\rightarrow\infty}\psi_N/N=1-\lambda$.
For each $1\leq a\leq I$ we have:
\begin{equation*}
    \begin{split}
        \frac{\pt h_a^\u(t,y_0(t))\po}{h_a^\u(1,y_0(1))}&=\frac{\pt \lb2\sqrt{t}y_0(t)+tq^{\frac{I+1}{2}-a}\k+q^{-\frac{I+1}{2}+a}\k^{-1}\rb\po}{2y_0(t)+q^{\frac{I+1}{2}-a}\k+q^{-\frac{I+1}{2}+a}\k^{-1}}\\
        &=\frac{\pt\lb At+1/A+ tq^{\frac{I+1}{2}-a}\k+q^{-\frac{I+1}{2}+a}\k^{-1}\rb\po}{A+1/A+q^{\frac{I+1}{2}-a}\k+q^{-\frac{I+1}{2}+a}/\k^{-1}} 
        =\frac{A\k}{A\k+q^{-\frac{I+1}{2}+a}}.
    \end{split}
\end{equation*}
Therefore 
$$
\frac{\pt h^\u(t,y_0(t))\po}{h^\u(1,y_0(1))}=\sum_{a=1}^I\frac{\pt h_a^\u(t,y_0(t))\po}{h_a^\u(1,y_0(1))}=\sum_{a=1}^I\frac{A\k}{A\k+q^{-\frac{I+1}{2}+a}}.
$$
Similarly,
$$
\frac{\pt h^\r(t,y_0(t))\po}{h^\r(1,y_0(1))}=\sum_{a=1}^I\frac{A\k^{-1}}{A\k^{-1}+q^{-\frac{I+1}{2}+a}}.
$$
Hence by Lemma \ref{lem:mean as partition} and  \eqref{eq:in proof of mean density}, we have
\be\label{eq:limit mean density HD phase}\lim_{N\rightarrow\infty}\mathbb{E}_{\mu_{\hP^N}}\lbe\frac{1}{N}\sum_{i=1}^N\tau_i\rbe
=\lim_{N\rightarrow\infty}\frac{\pt Z_N(t)\po}{NZ_N(1)}
=\lambda\sum_{a=1}^I\frac{A\k}{A\k+q^{-\frac{I+1}{2}+a}} 
+(1-\lambda) \sum_{a=1}^I\frac{A\k^{-1}}{A\k^{-1}+q^{-\frac{I+1}{2}+a}}.\ee

We now explain the proofs of \eqref{eq:partition function high density} and \eqref{eq:partition function derivative high density}. Since \( A > 1 \), \( A \notin \{ q^{-l} : l \in \mathbb{N}_0 \} \), \( C \in [0, 1) \), and \( B, D \in (-1, 0] \), for \( t \) in a small neighborhood of 1, all the atoms of the Askey--Wilson measure \( \nu\left( dy; A\sqrt{t}, B\sqrt{t}, C/\sqrt{t}, D/\sqrt{t}, q \right) \) are generated by \( A\sqrt{t} \) and are given by \( y_0(t) > \dots > y_k(t) > 1 \). Here, $y_j(t)=\frac{1}{2}\lb q^jA\sqrt{t}+\frac{1}{q^jA\sqrt{t}}\rb$ for $0\leq j\leq k$, and \( k  \in \mathbb{N}_0 \) is uniquely determined by $q^kA>1$ and $q^{k+1}A<1$. We take \( t = 1 \) in the formula \eqref{eq:partition function in proof} for \( Z_N(t) \) and obtain:
\be \label{eq:ZN1 hd phase}
    \begin{split}
        Z_N(1) =& h^\u(1, y_0(1))^{\phi_N} h^\r(1, y_0(1))^{\psi_N} p\left( y_0(1); A, B, C, D, q \right) \\
        &+ \sum_{j=1}^k h^\u(1, y_j(1))^{\phi_N} h^\r(1, y_j(1))^{\psi_N} p\left( y_j(1); A, B, C, D, q \right) \\
        &+ \int_{-1}^1 h^\u(1, y)^{\phi_N} h^\r(1, y)^{\psi_N} g(1, y) \, dy,
    \end{split}
\ee
where $g(t, y)$ is defined by \eqref{eq:gty}. Since $h^\u(1, y)$ and $h^\r(1, y)$ are strictly positive and strictly increasing functions over $y \in [-1, \infty)$, and since $\phi_N + \psi_N = N$, we have that the first term of \eqref{eq:ZN1 hd phase}, divided by $h^\u(1, y_0(1))^{\phi_N} h^\r(1, y_0(1))^{\psi_N}$, is $p(y_0(1); A, B, C, D, q) > 0$, independent of $N$, and every other term of \eqref{eq:ZN1 hd phase}, divided by it, converges to 0 as $N \to \infty$. We conclude the proof of \eqref{eq:partition function high density}:
$$
Z_N(1) \sim h^\u(1, y_0(1))^{\phi_N} h^\r(1, y_0(1))^{\psi_N} p(y_0(1); A, B, C, D, q).
$$
Next, we take the derivative at $t = 1$ of $Z_N(t)$ given by \eqref{eq:partition function in proof}. We have:
\be \label{eq:ZN1 derivative hd phase}
    \begin{split}
        \pt Z_N(t)\po &= \pt \left( h^\u(t, y_0(t))^{\phi_N} h^\r(t, y_0(t))^{\psi_N} \right) \po p\left( y_0(1); A, B, C, D, q \right) \\
        &+ h^\u(1, y_0(1))^{\phi_N} h^\r(1, y_0(1))^{\psi_N} \pt p\left( y_0(t); A\sqrt{t}, B\sqrt{t}, C/\sqrt{t}, D/\sqrt{t}, q \right) \po \\
        &+ \sum_{j=1}^k \pt \left( h^\u(t, y_j(t))^{\phi_N} h^\r(t, y_j(t))^{\psi_N} \right) \po p\left( y_j(1); A, B, C, D, q \right) \\
        &+ \sum_{j=1}^k h^\u(1, y_j(1))^{\phi_N} h^\r(1, y_j(1))^{\psi_N} \pt p\left( y_j(t); A\sqrt{t}, B\sqrt{t}, C/\sqrt{t}, D/\sqrt{t}, q \right) \po \\
        &+ \int_{-1}^1 \pt \left( h^\u(t, y)^{\phi_N} h^\r(t, y)^{\psi_N} \right) \po g(1, y) \, dy \\
        &+ \int_{-1}^1 h^\u(1, y)^{\phi_N} h^\r(1, y)^{\psi_N} \pt g(t, y) \po \, dy.
    \end{split}
\ee 
The first term of \eqref{eq:ZN1 derivative hd phase}, divided by $N h^\u(1, y_0(1))^{\phi_N} h^\r(1, y_0(1))^{\psi_N}$, converges to the positive constant given by the right-hand side of \eqref{eq:limit mean density HD phase} times $p(y_0(1); A, B, C, D, q) > 0$, as $N \to \infty$.
By the discussion below \eqref{eq:gty}, both $g(1, y)$ and $\pt g(t, y) \po $ are bounded over $y\in[-1,1]$. 
Note also that, from \eqref{eq:specific form of h},  both  $\pt h^\u(t, y) \po $ and $\pt h^\r(t, y) \po $ are bounded over $y\in[-1,1]$. 
Hence, one can observe that every 
other term of \eqref{eq:ZN1 derivative hd phase} except for the first one, divided by  $N h^\u(1, y_0(1))^{\phi_N} h^\r(1, y_0(1))^{\psi_N}$, converges to 0 as $N \to \infty$. We conclude the proof of \eqref{eq:partition function derivative high density}:
$$
\pt Z_N(t)\po \sim \pt \left( h^\u(t, y_0(t))^{\phi_N} h^\r(t, y_0(t))^{\psi_N} \right) \po p(y_0(1); A, B, C, D, q).
$$

\item [(\upperRomannumeral{2})] \label{2} (Low density phase $C>1$). 
When $t$ is close to $1$, atoms are generated by $\frac{C}{\sqrt{t}}$. Both $Z_N(1)$ and $\pt Z_N(t)\po$ are dominated by the largest atom $y_0(t)=\frac{1}{2}\lb \frac{C}{\sqrt{t}}+\frac{\sqrt{t}}{C}\rb$. By a calculation that is very similar with the high density phase, we are able to obtain:
$$\lim_{N\rightarrow\infty}\mathbb{E}_{\mu_{\hP^N}}\lbe\frac{1}{N}\sum_{i=1}^N\tau_i\rbe
=\lim_{N\rightarrow\infty}\frac{\pt Z_N(t)\po}{NZ_N(1)}
=\lambda\sum_{a=1}^I \frac{\k}{\k+Cq^{-\frac{I+1}{2}+a}} +(1-\lambda) \sum_{a=1}^I \frac{\k^{-1}}{\k^{-1}+Cq^{-\frac{I+1}{2}+a}}.$$
 \item [(\upperRomannumeral{3})] \label{3} (Maximal current phase $A<1$, $C<1$).
 When $t$ is close to $1$ there is no atom, i.e.
\be\label{eq:partition function maximal current} 
Z_N(t)=\int_{-1}^1h^\u(t,y)^{\phi_N}h^\r(t,y)^{\psi_N}g(t,y)dy.\ee
We can observe that as $N\rightarrow\infty$ we have:
\be\label{eq:partition function derivative maximal current}
\begin{split}
\lim_{N\rightarrow\infty}\frac{\pt Z_N(t)\po}{NZ_N(1)}
&=\lim_{N\rightarrow\infty}\frac{\int_{-1}^1\pt\lb h^\u(t,y)^{\phi_N}h^\r(t,y)^{\psi_N}\rb\po g(1,y)dy}{N\int_{-1}^1h(1,y)^{N}g(1,y)dy}\\
&=\lambda\lim_{N\rightarrow\infty}\frac{\int_{-1}^1 \frac{\pt h^\u(t,y)\po}{h(1,y)} h(1,y)^{N}g(1,y)dy}{\int_{-1}^1h(1,y)^{N}g(1,y)dy}
+(1-\lambda)\lim_{N\rightarrow\infty}\frac{\int_{-1}^1  \frac{\pt h^\r(t,y)\po}{h(1,y)}h(1,y)^{N}g(1,y)dy}{\int_{-1}^1h(1,y)^{N}g(1,y)dy},
\end{split} \ee
where we notice that $h^\u(1,y)=h^\r(1,y)$ and denote them both by $h(1,y)$.

We observe that:
$$
\frac{\pt h^\u(t,y)\po}{h(1,y)}=\sum_{a=1}^I\frac{\pt h_a^\u(t,y)\po}{h_a^\u(1,y)}
=\sum_{a=1}^I\lb\frac{1}{2}+\frac{q^{\frac{I+1}{2}-a}\k-q^{-\frac{I+1}{2}+a}\k^{-1}}{2h_a^\u(1,y)}\rb.
$$
Hence we have:
\be\label{eq:evaluation of derivative in proof}
\frac{\int_{-1}^1 \frac{\pt h^\u(t,y)\po}{h(1,y)} h(1,y)^{N}g(1,y)dy}{\int_{-1}^1h(1,y)^{N}g(1,y)dy}=\sum_{a=1}^I\lb\frac{1}{2}+\frac{q^{\frac{I+1}{2}-a}\k-q^{-\frac{I+1}{2}+a}\k^{-1}}{2}\frac{\int_{-1}^1\frac{h(1,y)^N}{h_a^\u(1,y)}g(1,y)dy}{\int_{-1}^1h(1,y)^Ng(1,y)dy}\rb.\ee

We use a similar method as \cite[section 4.2]{BW19} (see also \cite{Y}) to take limits.
We write:
\begin{equation*}
\begin{split}
    g(1,y)&=\frac{(q,AB,AC,AD,BC,BD,CD;q)_{\infty}}{2\pi(ABCD;q)_{\infty}\sqrt{1-y^2}}\left\vert\frac{(e^{2i\th_y};q)_{\infty}}{(Ae^{i\th_y},Be^{i\th_y},Ce^{i\th_y},D e^{i\th_y};q)_{\infty}}\right\vert^2\\
    &=\sqrt{1-y^2}\frac{2(q,AB,AC,AD,BC,BD,CD;q)_{\infty}}{\pi(ABCD;q)_{\infty}\vert(A e^{i\th_y},B e^{i\th_y},C e^{i\th_y},D e^{i\th_y};q)_{\infty}\vert^2}|(qe^{2i\th_y};q)_\infty|^2.
\end{split}
\end{equation*}
Set $y=1-\frac{u}{2N}$. Fix $u>0$, then as $N\rightarrow\infty$ we have $e^{i\th_y}\rightarrow 1$   hence $(qe^{2i\th_y};q)_\infty\rightarrow (q;q)_\infty$. Therefore
$$ g\lb1, 1-\frac{u}{2N}\rb\sim 2\mathfrak{c}\sqrt{\frac{u}{N}}, \quad\text{and}\quad g\lb1, 1-\frac{u}{2N}\rb\leq M\sqrt{\frac{u}{N}},$$
where  $M$ is a large enough constant and 
$$\mathfrak{c}=\frac{(q;q)_\infty^3(AB,AC,AD,BC,BD,CD;q)_\infty}{\pi(ABCD;q)_{\infty}(A,B,C,D;q)_\infty^2}.$$

In the rest of the proof we write $r_a=2+q^{\frac{I+1}{2}-a}\k+q^{-\frac{I+1}{2}+a}\k^{-1}$. 
We have
$$h\lb 1,1-\frac{u}{2N}\rb=\prod_{a=1}^I\lb r_a-\frac{u}{N}\rb=\prod_{a=1}^Ir_a\prod_{a=1}^I\lb1-\frac{u}{Nr_a}\rb.$$
Therefore
\begin{equation*}
    \begin{split}
        \int_{-1}^1h(1,y)^N g(1,y)dy &=\int_0^\infty 1_{u\leq 4N}h\lb 1,1-\frac{u}{2N}\rb^N g\lb1, 1-\frac{u}{2N}\rb\frac{du}{2N}\\
        &=\frac{\lb\prod_{a=1}^Ir_a\rb^N}{2N^{\frac{3}{2}}}\int_0^\infty 1_{u\leq 4N}\prod_{a=1}^I\lb 1-\frac{u}{Nr_a}\rb^N g\lb1, 1-\frac{u}{2N}\rb\sqrt{N}du.
    \end{split}
\end{equation*}
By $g\lb1, 1-\frac{u}{2N}\rb\leq M\sqrt{\frac{u}{N}}$ the right hand side can be bounded by a constant times $\exp\lb-\sum_{a=1}^I\frac{u}{r_a}\rb\sqrt{u}$, which is integrable over $u\in(0,\infty)$. Hence we can use dominated convergence theorem to take $N\rightarrow\infty$. We use $g\lb1, 1-\frac{u}{2N}\rb\sim 2\mathfrak{c}\sqrt{\frac{u}{N}}$
to get:
\be\label{eq:first integral proof}\int_{-1}^1h(1,y)^N g(1, y)dy\sim\frac{\lb\prod_{a=1}^Ir_a\rb^N}{2N^{\frac{3}{2}}}\int_0^\infty\exp\lb-\sum_{a=1}^I\frac{u}{r_a}\rb2\mathfrak{c}\sqrt{u}du.\ee
This integral can be explicitly evaluated using the formula $\int_0^\infty e^{-su}\sqrt{u}du=\frac{1}{8}\sqrt{\frac{\pi}{s^3}}$ for $s>0$. However we do not need to explicitly evaluate this integral. For any $1\leq b\leq N$,  we have:
\begin{equation*}
    \begin{split}
        \int_{-1}^1\frac{h(1,y)^N}{h_b^\u(1,y)}g(1, y)dy &=\int_0^\infty 1_{u\leq 4N}\frac{h\lb 1,1-\frac{u}{2N}\rb^N}{h_b^\u\lb 1,1-\frac{u}{2N}\rb} g\lb1, 1-\frac{u}{2N}\rb\frac{du}{2N}\\
        &=\frac{\lb\prod_{a=1}^Ir_a\rb^N}{2N^{\frac{3}{2}}r_b}\int_0^\infty 1_{u\leq 4N}\frac{\prod_{a=1}^I\lb 1-\frac{u}{Nr_a}\rb^N}{\lb 1-\frac{u}{Nr_b}\rb}g\lb1, 1-\frac{u}{2N}\rb\sqrt{N}du.
    \end{split}
\end{equation*}
The right hand side can be bounded by a constant times $\exp\lb-\sum_{a=1}^I\frac{u}{r_a}+\frac{u}{2r_b}\rb\sqrt{u}$, which is integrable over $u\in(0,\infty)$. Hence we can use dominated convergence theorem to take $N\rightarrow\infty$:
\be\label{eq:second integral proof}\int_{-1}^1\frac{h(1,y)^N}{h_b^\u(1,y)}g(1, y)dy\sim\frac{\lb\prod_{a=1}^Ir_a\rb^N}{2N^{\frac{3}{2}}r_b}\int_0^\infty\exp\lb-\sum_{a=1}^I\frac{u}{r_a}\rb2\mathfrak{c}\sqrt{u}du.\ee

Combining \eqref{eq:first integral proof} and \eqref{eq:second integral proof}, we get   for any $1\leq b\leq N$:
$$\lim_{N\rightarrow\infty}\frac{\int_{-1}^1\frac{h(1,y)^N}{h_b(1,y)}g(1,y)dy}{\int_{-1}^1h(1,y)^Ng(1,y)dy}=\frac{1}{r_b}.$$
By \eqref{eq:evaluation of derivative in proof}, one can evaluate:
$$
\lim_{N\rightarrow\infty}\frac{\int_{-1}^1 \frac{\pt h^\u(t,y)\po}{h(1,y)} h(1,y)^{N}g(1,y)dy}{\int_{-1}^1h(1,y)^{N}g(1,y)dy}
=\sum_{a=1}^I\lb\frac{1}{2}+\frac{q^{\frac{I+1}{2}-a}\k-q^{-\frac{I+1}{2}+a}\k^{-1}}{2r_a}\rb
=\sum_{a=1}^I \frac{\k}{\k+q^{-\frac{I+1}{2}+a}}.
$$
Similarly, we also have
$$
\lim_{N\rightarrow\infty}\frac{\int_{-1}^1 \frac{\pt h^\r(t,y)\po}{h(1,y)} h(1,y)^{N}g(1,y)dy}{\int_{-1}^1h(1,y)^{N}g(1,y)dy}
=\sum_{a=1}^I \frac{\k^{-1}}{\k^{-1}+q^{-\frac{I+1}{2}+a}}.
$$
Therefore, by \eqref{eq:partition function derivative maximal current}, we have:
$$
\lim_{N\rightarrow\infty}\mathbb{E}_{\mu_{\hP^N}}\lbe\frac{1}{N}\sum_{i=1}^N\tau_i\rbe
=\lim_{N\rightarrow\infty}\frac{\pt Z_N(t)\po}{NZ_N(1)}
=\lambda\sum_{a=1}^I \frac{\k}{\k+q^{-\frac{I+1}{2}+a}}+(1-\lambda)\sum_{a=1}^I \frac{\k^{-1}}{\k^{-1}+q^{-\frac{I+1}{2}+a}}.
$$

Combining the three phases above, we conclude the proof.
\end{enumerate}
\end{proof}


\begin{thebibliography}{10}

\bibitem{askey1985some}
R.~Askey and J.~Wilson.
\newblock {\em Some basic hypergeometric orthogonal polynomials that generalize Jacobi polynomials}, volume 319.
\newblock American Mathematical Soc., 1985.

\bibitem{barraquand2018stochastic}
G.~Barraquand, A.~Borodin, I.~Corwin, and M.~Wheeler.
\newblock Stochastic six-vertex model in a half-quadrant and half-line open asymmetric simple exclusion process.
\newblock {\em Duke Math. J.}, 2018.

\bibitem{borodin2017family}
A.~Borodin.
\newblock On a family of symmetric rational functions.
\newblock {\em Adv. Math.}, 306:973--1018, 2017.

\bibitem{BGW}
A.~Borodin, V.~Gorin, and M.~Wheeler.
\newblock Shift-invariance for vertex models and polymers.
\newblock {\em Proc. Lond. Math. Soc.}, 124(2):182--299, 2022.

\bibitem{borodin2018higher}
A.~Borodin and L.~Petrov.
\newblock Higher spin six vertex model and symmetric rational functions.
\newblock {\em Selecta Math. (N.S.)}, 24(2):751--874, 2018.

\bibitem{BW}
A.~Borodin and M.~Wheeler.
\newblock Colored stochastic vertex models and their spectral theory.
\newblock {\em Ast\'erisque}, 437, 2022.

\bibitem{BM16}
G.~Bosnjak and V.~Mangazeev.
\newblock Construction of {R}-matrices for symmetric tensor representations related to {$U_q(\widehat{sl_n})$}.
\newblock {\em J. Phys. A}, 49(49):495204, 2016.

\bibitem{BW19}
W.~Bryc and Y.~Wang.
\newblock Limit fluctuations for density of asymmetric simple exclusion processes with open boundaries.
\newblock {\em Ann. Inst. Henri Poincaré Probab. Stat.}, 55(4):2169--2194, 2019.

\bibitem{BW10}
W.~Bryc and J.~Weso{\l}owski.
\newblock {Askey--Wilson} polynomials, quadratic harnesses and martingales.
\newblock {\em Ann. Probab.}, pages 1221--1262, 2010.

\bibitem{BW17}
W.~Bryc and J.~Weso{\l}owski.
\newblock Asymmetric simple exclusion process with open boundaries and quadratic harnesses.
\newblock {\em J. Stat. Phys.}, 167:383--415, 2017.

\bibitem{Corwin_survey}
I.~Corwin.
\newblock Some recent progress on the stationary measure for the open {KPZ} equation.
\newblock {\em Toeplitz Operators and Random Matrices: In Memory of Harold Widom}, pages 321--360, 2022.

\bibitem{CP}
I.~Corwin and L.~Petrov.
\newblock Stochastic higher spin vertex models on the line.
\newblock {\em Comm. Math. Phys.}, 343:651--700, 2016.

\bibitem{Review_and_progress}
N.~Cramp{\'e}, E.~Ragoucy, and M.~Vanicat.
\newblock Integrable approach to simple exclusion processes with boundaries. {Review} and progress.
\newblock {\em J. Stat. Mech. Theory Exp.}, 2014(11):P11032, 2014.

\bibitem{dA22}
L.~d'Andecy.
\newblock Fusion for the {Yang--Baxter} equation and the braid group.
\newblock {\em Winter Braids Lecture Notes}, 7:1--49, 2020.

\bibitem{DEHP93}
B.~Derrida, M.~Evans, V.~Hakim, and V.~Pasquier.
\newblock Exact solution of a {1D} asymmetric exclusion model using a matrix formulation.
\newblock {\em J. Phys. A}, 26(7):1493, 1993.

\bibitem{F}
L.~Faddeev.
\newblock Quantum completely integrable models in field theory.
\newblock In {\em 40 Years in Mathematical Physics}, pages 187--235. World Scientific, 1995.

\bibitem{frappat2007complete}
L.~Frappat, R.~Nepomechie, and E.~Ragoucy.
\newblock A complete {Bethe} ansatz solution for the open spin-s {XXZ} chain with general integrable boundary terms.
\newblock {\em J. Stat. Mech. Theory Exp.}, 2007(09):P09009, 2007.

\bibitem{GZ94}
S.~Ghoshal and A.~Zamolodchikov.
\newblock Boundary {S} matrix and boundary state in two-dimensional integrable quantum field theory.
\newblock {\em Internat. J. Modern Phys. A}, 9(21):3841--3885, 1994.

\bibitem{Jimmy_He}
J.~He.
\newblock Shift invariance of half space integrable models.
\newblock {\em Probab. Theory Related Fields}, pages 1--71, 2025.

\bibitem{jimbo1985q}
M.~Jimbo.
\newblock A {$q$}-difference analogue of {$U(g)$} and the {Yang--Baxter} equation.
\newblock {\em Lett. Math. Phys.}, 10:63--69, 1985.

\bibitem{karowski1979bound}
M.~Karowski.
\newblock On the bound state problem in 1+ 1 dimensional field theories.
\newblock {\em Nuclear Phys. B}, 153:244--252, 1979.

\bibitem{kirillov1987exact}
A.~Kirillov and N.~Reshetikhin.
\newblock Exact solution of the integrable {XXZ} {Heisenberg} model with arbitrary spin. {I}. the ground state and the excitation spectrum.
\newblock {\em J. Phys. A}, 20(6):1565, 1987.

\bibitem{kuan2018algebraic}
J.~Kuan.
\newblock An algebraic construction of duality functions for the stochastic {$\mathcal{U}_q(A_n^{(1)})$} vertex model and its degenerations.
\newblock {\em Comm. Math. Phys.}, 359:121--187, 2018.

\bibitem{Kuan_stochastic}
J.~Kuan.
\newblock Stochastic {B}axterisation of a fused {H}ecke algebra.
\newblock {\em arXiv preprint arXiv:2012.11564}, 2020.

\bibitem{kulish1981yang}
P.~Kulish, N.~Reshetikhin, and E.~Sklyanin.
\newblock {Yang--Baxter} equation and representation theory: I.
\newblock {\em Lett. Math. Phys.}, 5:393--403, 1981.

\bibitem{kulish1992algebraic}
P.~Kulish and E.~Sklyanin.
\newblock Algebraic structures related to reflection equations.
\newblock {\em J. Phys. A}, 25(22):5963, 1992.

\bibitem{kulish2005quantum}
P.~Kulish and E.~Sklyanin.
\newblock Quantum spectral transform method recent developments.
\newblock In {\em Integrable Quantum Field Theories: Proceedings of the Symposium Held at Tv{\"a}rminne, Finland, 23--27 March, 1981}, pages 61--119. Springer, 2005.

\bibitem{kuniba2016stochastic}
A.~Kuniba, V.~Mangazeev, S.~Maruyama, and M.~Okado.
\newblock Stochastic {$R$} matrix for {$U_q(A_n^{(1)})$}.
\newblock {\em Nuclear Phys. B}, 913:248--277, 2016.

\bibitem{mangazeev2014yang}
V.~Mangazeev.
\newblock On the {Yang--Baxter} equation for the six-vertex model.
\newblock {\em Nuclear Phys. B}, 882:70--96, 2014.

\bibitem{ML}
V.~Mangazeev and X.~Lu.
\newblock Boundary matrices for the higher spin six vertex model.
\newblock {\em Nuclear Phys. B}, 945:114665, 2019.

\bibitem{mezincescu1992fusion}
L.~Mezincescu and R.~Nepomechie.
\newblock Fusion procedure for open chains.
\newblock {\em J. Phys. A}, 25(9):2533, 1992.

\bibitem{SW97}
T.~Sasamoto and M.~Wadati.
\newblock Stationary state of integrable systems in matrix product form.
\newblock {\em J. Phys. Soc. Japan}, 66(9):2618--2627, 1997.

\bibitem{sklyanin1988boundary}
E.~Sklyanin.
\newblock Boundary conditions for integrable quantum systems.
\newblock {\em J. Phys. A}, 21(10):2375, 1988.

\bibitem{USW04}
M.~Uchiyama, T.~Sasamoto, and M.~Wadati.
\newblock Asymmetric simple exclusion process with open boundaries and {Askey--Wilson} polynomials.
\newblock {\em J. Phys. A}, 37(18):4985, 2004.

\bibitem{Vanicat_fused_MPA}
M.~Vanicat.
\newblock Integrable {F}loquet dynamics, generalized exclusion processes and ``fused'' matrix ansatz.
\newblock {\em Nuclear Phys. B}, 929:298--329, 2018.

\bibitem{wang2023askey}
Y.~Wang, J.~Weso{\l}owski, and Z.~Yang.
\newblock Askey--{W}ilson signed measures and open {ASEP} in the shock region.
\newblock {\em Int. Math. Res. Not. IMRN}, 2024(15):11104--11134, 2024.

\bibitem{Y}
Z.~Yang.
\newblock Stationary measure for six-vertex model on a strip.
\newblock {\em Electron. J. Probab.}, 29:1--28, 2024.

\bibitem{ZZ}
A.~Zamolodchikov and A.~Zamolodchikov.
\newblock Factorized {$S$}-matrices in two dimensions as the exact solutions of certain relativistic quantum field theory models.
\newblock {\em Ann. Physics}, 120(2):253--291, 1979.

\end{thebibliography}
\end{document}